\documentclass[a4paper,UKenglish,cleveref, autoref, thm-restate]{lipics-v2021}

\usepackage[T1]{fontenc}
\usepackage{lmodern}
\usepackage{url}

\usepackage{float} 
\usepackage{microtype}
\usepackage[utf8]{inputenc}
\usepackage{amsmath, amsthm, amssymb, mathtools, stmaryrd}
\SetSymbolFont{stmry}{bold}{U}{stmry}{m}{n}
\usepackage{tikz}
\usepackage{xspace}
\usepackage{todonotes}
\usepackage{enumerate}
\usepackage{listings}
\usepackage{wrapfig}
\usepackage{comment}

\usetikzlibrary{backgrounds}
\usetikzlibrary{patterns}
\usetikzlibrary{decorations.pathreplacing,calligraphy}
\usetikzlibrary{fadings}

\nolinenumbers

\makeatletter
\newcommand{\@abbrev}[3]{
	\def\c@a@def##1{
		\if ##1.
		\relax
		\else
		\@ifdefinable{\@nameuse{#1##1}}{\@namedef{#1##1}{#2##1}}
		\expandafter\c@a@def
		\fi
	}
	\c@a@def #3.
}
\@abbrev{bb}{\mathbb}{ABCDEFGHIJKLMNOPQRSTUVWXYZ}
\@abbrev{bf}{\mathbf}{ABCDEFGHIJKLMNOPQRSTUVWXYZabcdefghijklmnopqrstuvwxyz}
\@abbrev{bit}{\boldsymbol}{
	ABCDEFGHIJKLMNOPQRSTUVWXYZabcdefghijklmnopqrstuvwxyz}
\@abbrev{mc}{\mathcal}{ABCDEFGHIJKLMNOPQRSTUVWXYZ}
\@abbrev{mb}{\mathbb}{ABCDEFGHIJKLMNOPQRSTUVWXYZ}
\@abbrev{mf}{\mathfrak}{
	ABCDEFGHIJKLMNOPQRSTUVWXYZabcdefghijklmnopqrstuvwxyz}
\@abbrev{rm}{\mathrm}{ABCDEFGHIJKLMNOPQRSTUVWXYZabcdefghijklmnopqrstuvwxyz}
\@abbrev{scr}{\mathscr}{
	ABCDEFGHIJKLMNOPQRSTUVWXYZabcdefghijklmnopqrstuvwxyz}
\@abbrev{sf}{\mathsf}{ABCDEFGHIJKLMNOPQRSTUVWXYZabcdefghijklmnopqrstuvwxyz}
\@abbrev{b}{\bar}{abcdefghijklmnopqrstuvwxyz}
\makeatother

\newcommand{\AC}{\ensuremath{\textsc{AC}}}

\newcommand{\FO}{\ensuremath{\textsc{FO}}}

\newcommand{\CPT}{\ensuremath{\text{CPT}}}

\newcommand{\classmonPCx}[1]{\ensuremath{\mathbf{\mathcal{K}_{\monPCx 
				k}}}}

\newcommand{\G}{\ensuremath{\mathfrak{G}}}

\newcommand{\HF}{\text{\upshape HF}}
\newcommand{\Aut}{\mathbf{Aut}}
\newcommand{\dist}{\text{dist}}

\newcommand{\Sym}{{\mathbf{Sym}}}

\renewcommand{\phi}{\varphi}

\newcommand{\lra}{\longrightarrow}

\renewcommand{\phi}{\varphi}
\renewcommand{\theta}{\vartheta}

\renewcommand{\AA}{{\mathfrak A}}
\newcommand{\BB}{{\mathfrak B}}

\newcommand{\GG}{{\mathfrak G}}
\newcommand{\HH}{\mathfrak H}

\renewcommand{\epsilon}{\varepsilon}

\newcommand{\ptime}{\mbox{\sc Ptime}}

\newcommand{\tc}{\text{\upshape tc}}

\newcommand{\Aa}{{\cal A}}
\newcommand{\Bb}{{\cal B}}
\newcommand{\Cc}{{\cal C}}

\newcommand{\Ee}{{\cal E}}

\newcommand{\Hh}{{\cal H}}

\newcommand{\Kk}{{\cal K}}

\newcommand{\Oo}{{\cal O}}
\newcommand{\Pp}{{\cal P}}
\newcommand{\Qq}{{\cal Q}}

\newcommand{\Ss}{{\cal S}}

\newcommand{\Vv}{{\cal V}}

\newcommand{\Xx}{{\cal X}}

\renewcommand{\bar}{\overline}

\renewcommand{\G}{G} 
\newcommand{\Orb}{\mathbf{Orb}}
\newcommand{\Alt}{\mathbf{Alt}}

\newcommand{\Stab}{\mathbf{Stab}}
\newcommand{\StabP}{\Stab^\bullet}

\newcommand{\CFI}{\text{\upshape CFI}}
\newcommand{\maxOrb}{\mathbf{maxOrb}}
\renewcommand{\H}{\Hh}		

\newcommand{\cutout}[1]{}

\newcommand{\tp}{\text{tp}}

\newcommand{\Ker}{\textbf{Ker}}

\renewcommand{\Im}{\textbf{Im}}
\newcommand{\rk}{\textbf{rk}}

\newcommand{\tw}{\textbf{tw}}

\newcommand{\Sp}{\mathbf{SP}}

\newcommand{\Orbit}{\mathbf{Orbit}}

\newcommand{\Pair}{\textsf{\upshape Pair}}

\newcommand{\Union}{\textsf{\upshape Union}}

\title{ Lower bounds for Choiceless Polynomial Time via Symmetric XOR-circuits  } 
\titlerunning{     }  

\author{Benedikt Pago}{Mathematical Foundations of Computer Science, RWTH Aachen University, Germany }{pago@logic.rwth-aachen.de}{}{}


\authorrunning{B. Pago} 

\Copyright{Benedikt Pago} 

\acknowledgements{I would like to thank Daniel Wiebking for his help with the group-theoretic proofs in the appendix.}

\begin{document}
	\maketitle

\begin{abstract}	
	One of the central questions in finite model theory is whether there exists a logic that captures polynomial time. 
	After rank logic and more general linear-algebraic logics have been separated from polynomial time, the most important candidate logics that remain are Choiceless Polynomial Time (CPT) and an extension thereof with a witnessed symmetric choice operator.\\
	In this paper, we make progress towards separating CPT from polynomial time by firstly establishing a connection between the expressive power of CPT and the existence of certain symmetric circuit families, and secondly, proving lower bounds against 
	these circuits.\\
	We focus on the isomorphism problem of unordered Cai-Fürer-Immerman-graphs (the CFI-query) as a potential candidate for separating CPT from P. Results by Dawar, Richerby and Rossman, and subsequently by Pakusa, Schalthöfer and Selman show that the CFI-query is CPT-definable on linearly ordered and preordered base graphs with small colour classes. We define a class of CPT-algorithms, that we call ``CFI-symmetric algorithms'', which generalises all the known ones, and show that such algorithms can only define the CFI-query on a given class of base graphs if there exists a family of symmetric XOR-circuits with certain properties. These properties include that the circuits have the same symmetries as the base graphs, are of polynomial size, and satisfy certain fan-in restrictions.
	Then we prove that such circuits with slightly strengthened requirements (i.e.\ stronger symmetry and fan-in and fan-out restrictions) do not exist for the $n$-dimensional hypercubes as base graphs. This almost separates the CFI-symmetric algorithms from polynomial time -- up to the gap that remains between the circuits whose existence we can currently disprove and the circuits whose existence is necessary for the definability of the CFI-query by a CFI-symmetric algorithm.\\
	To an extent, the connection between XOR-circuits and CPT-algorithms for the CFI-query can also be generalised to non-CFI-symmetric algorithms (even though such algorithms are currently not known). 
\end{abstract}

\section{Introduction}	
\emph{Choiceless Polynomial Time} (CPT) \cite{blass1999} is one of the most prominent candidate logics in finite model theory for capturing polynomial time. It can be seen as an extension of \emph{fixed-point logic with counting} \cite{dawar2015nature} with \emph{hereditarily finite sets} as data structures. This allows in principle to simulate arbitrary computations in CPT -- the only restrictions are that the computed h.f.\ sets are symmetric under the automorphisms of the input structure and polynomially bounded in size.\\
Other less studied candidates are logics with with witnessed choice constructs, such as fixed-point logic with witnessed symmetric choice and interpretations \cite{witnessedIFPlowerBound} and CPT extended with witnessed symmetric choice \cite{witnessedCPT}. Prior to Lichter's breakthrough \cite{lichter}, which separates rank logic from $\ptime$ using a variation of the famous Cai-Fürer-Immerman (CFI) construction \cite{caifurimm92}, logics with linear-algebraic operators \cite{linearLogics} were also considered reasonable candidates. However, as outlined in \cite{DGL22}, the results from \cite{lichter} and \cite{linearLogics} together imply that \emph{no} set of isomorphism-invariant linear algebraic operators can be used to define a logic capturing $\ptime$. Strong lower bounds for \emph{Choiceless Polynomial Time}, though, have remained elusive. Motivated by Gurevich's conjecture that no logic at all can capture polynomial time, in this article we make progress towards separating CPT from $\ptime$. 
For an overview of the unresolved problem "Is there a logic for $\ptime$?" in general, as well as the logic CPT in particular, see for example \cite{gradel2015polynomial}, \cite{grohe2008quest}, \cite{pakusa2015linear}, or \cite{svenja}.\\

Concerning CPT lower bounds, not very much is known so far: There is a non-definability result for a \emph{functional} problem in $\ptime$, namely it is impossible to define the dual of a given finite vector space in CPT \cite{rossman2010choiceless}. What we would like to have is, however, the inexpressibility of a polynomial time \emph{decision problem}. We focus on a standard benchmark from finite model theory, namely the \emph{CFI-query}. It asks to output, given a CFI-graph, whether it is odd or even. What this means is explained in Section \ref{sec:CFI} (it is equivalent to the graph isomorphism problem on these instances). The CFI-query is decidable in polynomial time but not in fixed-point logic with counting nor in rank logic (for generalised CFI-structures). It is open whether it is CPT-definable on \emph{unordered} instances, and our goal is to eventually answer this question in the negative. Our approach starts off from positive results: There do exist CPT-algorithms for linearly ordered and preordered versions of the CFI-query \cite{dawar2008, pakusaSchalthoeferSelman} and also CFI-graphs over base graphs of linear degree \cite{pakusaSchalthoeferSelman}. All these algorithms are based on the construction of hereditarily finite sets which somehow encode the parity of the given CFI-graph in their structure. These sets have been called \emph{super-symmetric} in \cite{dawar2008}.\\
In \cite{pago} it was shown that there exist unordered CFI-graphs (over $n$-dimensional hypercubes) whose degree is not linear and which cannot be preordered in CPT in such a way that the preorder-based algorithm from \cite{pakusaSchalthoeferSelman} (or the total-order-based one from \cite{dawar2008}) could be applied. This shows that these known choiceless algorithms for preordered versions of the CFI-query do not generalise to the unordered case because the necessary combinatorial objects (said preorders) are not symmetric enough: The main limiting factor of CPT is that it cannot define objects which break the symmetry of the input structure ``too much'' -- this is also why unordered (and hence highly symmetric) CFI-instances seem promising for lower bounds.\\
In the present paper, we take this strategy further: We define a general class of CPT-algorithms for the CFI-query, which encompasses all the known ones mentioned above, and show that their expressiveness depends on the existence of certain symmetric combinatorial objects, namely circuits with Boolean XOR-gates. We show that the CFI-query over a given class $\Kk$ of base graphs is only definable by an algorithm from that class if there exists a family of polynomial-size XOR-circuits with certain properties and, in particular, with the same symmetries as the graphs in $\Kk$ (Theorem \ref{thm:mainInformal}). This means that the non-definability of the CFI-query over $\Kk$ can be shown by proving the non-existence of symmetric circuits with the required properties. Indeed, we almost achieve this goal: Our second main result is a lower bound against such circuits; it shows that if we take as $\Kk$ the family of $n$-dimensional hypercubes and make the circuit properties slightly more restrictive than required by our Theorem \ref{thm:mainInformal}, then no circuit family can satisfy all these properties simultaneously. Thus, we come close to showing that the CFI-query over unordered hypercubes is undefinable by any CPT-algorithm from the general class we are considering. 

\paragraph*{Results}
We define in Section \ref{sec:CFIsymmetricSets} a class of CPT-algorithms for the CFI-query that contains all the currently known ones and prove that solving the CFI-query on a given class $\Kk$ of unordered base graphs by means of such an algorithm presupposes the existence of certain symmetric XOR-circuits. Following \cite{dawar2008}, we denote CFI-instances over a base graph $G$ as $\GG^S$, where $S$ is the set of vertices whose CFI-gadget is odd (see Section \ref{sec:CFI}). In the following theorem, we consider circuits whose input gates are labelled with the edges of the base graph $G_n$, and all internal gates are XOR gates. The group $\Aut(G_n) \leq \Sym(V(G_n))$ is the automorphism group of the base graph $G_n$. The $\Aut(G_n)$-orbit of the circuit refers to the set of all images of the circuit under relabelings of the input gates with permutations in $\Aut(G_n)$. A circuit is \emph{sensitive} to an input bit if flipping that bit changes the output of the circuit. The \emph{fan-in dimension} of a circuit is a parameter that we define in this paper; it is a generalisation of the fan-in degree. All these notions concerning our circuits are presented in detail in Section \ref{sec:XOR_circuits}.
\begin{theorem}[Main Theorem, informal]
	\label{thm:mainInformal}
	Let $(G_n = (V_n,E_n))_{n \in \bbN}$ be a sequence of base graphs. Let $\GG^S_n$ be a CFI-graph over $G_n$, and let $\tw_n$ denote the treewidth of $G_n$. If there exists a $\CPT$-program $\Pi$ that is \emph{super-symmetric} and \emph{CFI-symmetric} and decides the CFI-query on the instances $\GG_n^S$, for all $n \in \bbN$, then there also exists a family $(C_n)_{n \in \bbN}$ of XOR-circuits such that  
	\begin{enumerate}
		\item The number of gates in $C_n$ is polynomial in $|\GG^S_n|$.
		\item The $\Aut(G_n)$-orbit of the circuit has size polynomial in $|\GG^S_n|$. 
		\item $C_n$ is \emph{sensitive} to $\Omega(\tw_n)$ many input bits. 
		\item The \emph{fan-in dimension} of $C_n$ is $\Oo(\log |\GG_n^S|)$.
	\end{enumerate}	
\end{theorem}	

The detailed version of this is Theorem \ref{thm:XOR_lowerboundProgram} in Section \ref{sec:XOR_circuits}.
The terms \emph{super-symmetric} and \emph{CFI-symmetric} refer to the properties of a h.f.\ set that is constructed by the program $\Pi$ in order to decide the CFI-query. \emph{Super-symmetry} is a property of h.f.\ sets that goes back to \cite{dawar2008} and means that a h.f.\ set is fixed by \emph{all} edge flips of a CFI-structure $\GG^S$, not only by those which are automorphisms of $\GG^S$ (see Section \ref{sec:CFI_automorphisms}). \emph{CFI-symmetry} is a concept that we define in this paper and which describes the internal structure and ``local symmetries'' of a h.f.\ set. The CFI-algorithms from \cite{dawar2008} and \cite{pakusaSchalthoeferSelman} are based on h.f.\ sets which are both super-symmetric and CFI-symmetric. Arguably, both these properties are crucial for the success of all these algorithms: The way they work is that they aggregate all the CFI-gadgets of $\GG^S$ into one big h.f.\ set $\mu$ which is symmetric under all flips of edge gadgets. Then the vertices $e_0$ and $e_1$ in the edge gadgets in $\mu$ are replaced with the constants $0$ and $1$. Since CPT is choiceless, it cannot arbitrarily choose which of the vertices $e_0, e_1$ is replaced with which constant; it has to do it in both ways. The super-symmetry of $\mu$ ensures that both these replacements lead to the same object. This is necessary to avoid the creation of an exponential number of h.f.\ sets when all combinations are tried out. After all atoms in $\mu$ have been replaced with $0$ and $1$, the parity of the original CFI-graph $\GG^S$ can be extracted. The property of $\mu$ that we call CFI-symmetry essentially says that $\mu$ is composed out of sub-objects which have a similar behaviour as CFI-gadgets. This seems to be a natural design pattern for super-symmetric objects which encode the parity of $\GG^S$ but one could also conceive super-symmetric objects which are not CFI-symmetric (or vice versa). In short, super-symmetry is the main property of h.f.\ sets that makes all known CFI-algorithms work, and CFI-symmetry is the established design principle to achieve super-symmetry.\\

Theorem \ref{thm:mainInformal} reduces the question about the CPT-definability of the CFI-query to the question for the existence of certain non-trivial combinatorial objects, namely polynomial-size symmetric XOR-circuits. Its proof is based on a translation of h.f.\ sets over CFI-structures into XOR-circuits (Theorem \ref{thm:XOR_mainCircuitTheorem}). This is then combined with a lower bound from \cite{dawar2008} on the \emph{support} size of the h.f.\ sets required to decide the CFI-query -- the support size somehow measures how asymmetric a set is (see Section \ref{sec:CFIsupports}).
As our second main result shows, it is not at all clear that the required circuit families do exist for all base graphs, and so, obtaining lower bounds for CPT via lower bounds for symmetric circuits may indeed be possible.\\

 Before we come to this second result, in Section \ref{sec:generalisedCircuitConstruction}, we also prove a version of the above theorem without the restriction to \emph{CFI-symmetric} algorithms. It requires that the h.f.\ sets that are constructed by the algorithm have certain \emph{symmetric bases}, which we prove to be a more general property than CFI-symmetry. However, since all currently known CPT-algorithms for the CFI-query are CFI-symmetric, it is not clear that this more general (and much more complicated to prove) version of the theorem will actually be necessary at some point in order to separate CPT from $\ptime$. It could be that the class of choiceless algorithms satisfying the said symmetric basis condition is a strictly bigger class than the CFI-symmetric ones, but it could also be that every CPT-algorithm for the CFI-query is equivalent to a CFI-symmetric one (in which case Theorem \ref{thm:mainInformal} would in principle be sufficient to separate CPT from P via symmetric circuit lower bounds).\\

Our second main result, which we prove in Section \ref{sec:chapterXOR2}, shows that if we choose the $n$-dimensional hypercubes as the family of base graphs, and impose slightly stronger conditions on the circuits, then it is not possible to satisfy all of them together.
\begin{theorem}
	\label{thm:mainLowerBound}
	Let $(\H_n)_{n \in \bbN}$ be the family of $n$-dimensional hypercubes and let $\tw_n$ denote the treewidth of $\H_n$. Let $\HH_n^S$ denote a CFI-structure over $\H_n$. There exists \emph{no} family of symmetric XOR-circuits $(C_n)_{n \in \bbN}$ such that: 
	\begin{enumerate}
		\item The number of gates in $C_n$ is polynomial in $|\HH^S_n|$.
		\item The $\Aut(\H_n)$-orbit of the circuit has size exactly one. 
		\item $C_n$ is \emph{sensitive} to $\Omega(\tw_n)$ input bits. 
		\item For any two gates $g,h$ in $C_n$ such that $h$ is a parent of $g$, it holds $|\Orbit_{(h)}(g)| \in \Oo(\log |\HH_n^S|)$ and $|\Orbit_{(g)}(h)| \in \Oo(\log |\HH_n^S|)$.
	\end{enumerate}	
	Here, $\Orbit_{(h)}(g)$ denotes the orbit of the gate $g$ with respect to the subgroup of $\Aut(\H_n)$ that fixes the gate $h$ (and vice versa for $\Orbit_{(g)}(h)$).
\end{theorem}	

If the four circuit properties were the same as in Theorem \ref{thm:mainInformal}, then this would separate the class of super- and CFI-symmetric choiceless algorithms from $\ptime$. The difference between the two theorems is that here, the circuit has orbit size one, i.e.\ it is stabilised by the whole group $\Aut(\H_n)$, whereas in Theorem \ref{thm:mainInformal}, the orbit of the circuit is only required to be polynomial. Moreover, here, we have a logarithmic bound on the parents and children (per orbit) of every gate, whereas in Theorem \ref{thm:mainInformal}, the logarithmic bound is on the \emph{fan-in dimension} of the gates. We define this notion in Section \ref{sec:XOR_circuits}; we do not know if logarithmic fan-in dimension implies the orbit-wise logarithmic bound on the number of children (or vice versa), and probably, it does not imply the bound on the number of parents. So the ``gap'' between our two main results concerns how symmetric the circuits have to be and how restricted the connectivity between two consecutive circuit layers is. The proof of Theorem \ref{thm:mainLowerBound} involves group theoretic techniques based on those from \cite{dawarAnderson}.

\paragraph*{Related work}
The study of lower bounds for symmetric circuits has proven to be fruitful in many contexts: Anderson and Dawar established families of highly symmetric Boolean circuits with majority gates as a computation model equivalent to fixed-point logic with counting \cite{dawarAnderson}. A generalisation of these circuits also captures rank logic \cite{dawarGreg}. Our results regarding CPT are weaker than those in the sense that we do not give a circuit \emph{characterisation} of CPT but only a structural description of the relevant h.f.\ sets that CPT can use to decide the CFI-query. Lower bounds against our circuits seem to be generally harder to obtain than for the circuits from \cite{dawarAnderson} and \cite{dawarGreg} because our symmetry requirements are weaker (for interesting classes of base graphs, at least).\\
Besides these connections to logics from finite model theory, symmetric circuits are also interesting in the context of the VP vs VNP question. Dawar and Wilsenach have shown super-polynomial lower bounds on symmetric arithmetic circuits for the permanent \cite{dawar2020symmetric} and determinant \cite{symmetricCircuitsDeterminant}, for different symmetry groups. Lifting these lower bounds to less symmetric circuits for the permanent would be a step forward towards separating VP from VNP.\\
Other examples for symmetric circuit lower bounds concern $\AC^0$-circuits for the parity function \cite{rossman2019subspace}, and Boolean circuits for the multiplication of permutation matrices \cite{rossmanHe}. An interesting aspect about Rossman's lower bound for $\AC^0$-circuits computing parity is the symmetry group he considers: Contrary to the other mentioned results, the symmetry group is in this case not a large permutation group on the input variables but a Boolean vector space which acts on the set of input literals $\{X_1,\overline{X}_1,...,X_n,\overline{X}_n\}$ by swapping specified literals $X_i$ with their respective negations. This is reminiscent of the flips of CFI-gadgets that we encounter in the present article. However, a direct connection between Rossman's lower bound and ours does not seem to exist because his lower bound concerns the parity function, which can be easily expressed in our setting with a single XOR-gate.\\

Another research direction that is connected with this topic is on extensions of CPT. Lichter and Schweitzer have developed Choiceless Polynomial Time with \emph{witnessed symmetric choice} \cite{witnessedCPT}. This is a logic that allows to make arbitrary choices from definable orbits of the structure, as long as the automorphisms that witness a choice set to be an orbit are also definable. This logic captures polynomial time on all classes of structures where it can define the isomorphism problem, so the witnessed choice operator essentially reduces canonisation to isomorphism testing. A question is in how far our proposed lower bound approach via symmetric circuits also applies to CPT with witnessed choice. As shown in \cite{witnessedCPT}, CPT with witnessed choices has no difficulties to define the CFI-query on structures with a single orbit, i.e.\ unordered CFI-graphs. These are, however, precisely the example for which we have the circuit lower bound, so it seems like the circuit approach exploits a weakness of CPT that does not exist in the witnessed choice extension.

\section{Preliminaries}

\paragraph*{Bounded variable counting logic}

 For $k \in \bbN$, $\Cc^k$ denotes the $k$-variable fragment of first-order logic with \emph{counting quantifiers}. The counting quantifiers in this logic are of the form $\exists^i x \phi(x)$, for every $i \in \bbN$, expressing that at least $i$ elements of the structure satisfy $\phi(x)$. Note that such counting quantifiers can be simulated in ordinary $\FO$ but this requires more than one variable. 
Two structures $\AA$ and $\BB$ are called \emph{$\Cc^k$-equivalent}, denoted $\AA \equiv_{\Cc^k} \BB$, if they satisfy exactly the same $\Cc^k$-sentences.\\

The standard tool to prove $\Cc^k$-equivalence of two given structures is the \emph{bijective $k$-pebble game}. It is played on a pair of structures $(\AA,\BB)$
by two players, Spoiler and Duplicator. Duplicator has a winning strategy if and only if $\AA \equiv_{\Cc^k} \BB$ \cite{hella1996logical}. The game proceeds as follows: A position in the game is a set of pebble-pairs $\pi \subseteq A \times B$ of size at most $k$. In each round, Spoiler may pick up any number of pebble-pairs and remove them from the board such that in the resulting position $\pi'$, less than $k$ pebble-pairs remain. Then Duplicator specifies a bijection $f: A \lra B$ such that for every $(a,b) \in \pi'$, $f(a) = b$. Spoiler now puts down a new pebble on some element $a \in A$ of his choice, and the corresponding pebble in $B$ is placed on $f(a)$. If the resulting set of pebble-pairs does not induce a local isomorphism, then Spoiler wins. Duplicator has a winning strategy if she can enforce to play forever without losing. A position $\pi = \{ (a_1,b_1),...,(a_k,b_k)  \}$ is said to induce a local isomorphism if the mapping $g$ that maps each $a_i$ to $b_i$, for $i \in [k]$, is an isomorphism from the induced substructure of $\AA$ with universe $\{a_1,...,a_k\}$ into the induced substructure of $\BB$ with universe $\{b_1,...,b_k\}$.\\

The positions $\pi$ from which Duplicator has a winning strategy are given by those tuples that have the same \emph{$\Cc^k$-type} in $\AA$ and $\BB$. The \emph{$\Cc^k$-type} of a tuple $\bar{a} \in A^{\leq k}$ is the collection of all $\Cc^k$-formulas that are satisfied by $\bar{a}$ in $\AA$. It is known that in each finite structure, every $\Cc^k$-type is definable by a single $\Cc^k$-formula, so even though a type is an infinite collection of formulas, it is semantically equivalent to one finite $\Cc^k$-formula, if the structure is fixed \cite{gradel1993inductive}.

\paragraph*{Hereditarily finite sets}
Let $A$ be a finite set of atoms. Usually, the atoms will be the universe of a structure $\AA$ (and by convention, whenever a structure is called $\AA$, then $A$ denotes its universe).\\
The set of hereditarily finite objects over $A$, $\HF(A)$, is defined as $\bigcup_{i \in \bbN} \HF_i(A)$, where $\HF_0(A) := A \cup \{\emptyset\}, \HF_{i+1}(A) := \HF_{i}(A) \cup 2^{\HF_{i}(A)}$. The size of a h.f.\ set $x \in \HF(a)$ is measured in terms of its \emph{transitive closure} $\tc(x)$: The set $\tc(x)$ is the least transitive set such that $x \in \tc(x)$. Transitivity means that for every $a \in \tc(x)$, $a \subseteq \tc(x)$. Intuitively, one can view $\tc(x)$ as the set of all sets that appear as elements at some nesting depth within $x$. 

\paragraph*{Choiceless Polynomial Time}
By CPT we always mean Choiceless Polynomial Time \emph{with counting}. For details and various ways to define CPT formally, we refer to the literature: A concise survey can be found in \cite{gradel2015polynomial}. The work in which Blass, Gurevich and Shelah originally introduced CPT as an abstract state machine model is \cite{blass1999} from 1999; later, more ``logic-like'' presentations of CPT were invented, such as Polynomial Interpretation Logic \cite{grapakschalkai15, svenja} and BGS-logic \cite{rossman2010choiceless, gradel2015polynomial}. In short, CPT is like the better-studied \emph{fixed-point logic with counting} \cite{dawar2015nature} plus a mechanism to construct isomorphism-invariant \emph{hereditarily finite sets} of polynomial size. When a CPT-sentence $\Pi$ is evaluated in a finite structure $A$, then $\Pi$ may augment $A$ with hereditarily finite sets over its universe. The total number of distinct sets appearing in them (i.e.\ the sum over the sizes of the transitive closures of the h.f.\ sets) and the number of computation steps is bounded by $p(|A|)$, where $p(n)$ is a polynomial that is explicitly part of the sentence $\Pi$. For the sake of illustration, we sketch the definition of BGS-logic:\\

The sentences of BGS-logic are called \emph{programs}. A program is a tuple $\Pi = (\Pi_{\text{step}}(x),\\ \Pi_{\text{halt}}(x),\Pi_{\text{out}}(x),p(n))$. Here, $\Pi_{\text{step}}(x)$ is a BGS-term, $\Pi_{\text{halt}}$ and $\Pi_{\text{out}}$ are BGS-formulas, and $p(n)$ is a polynomial that bounds the time and space used by the program. BGS-terms take as input hereditarily finite sets and output a hereditarily finite set. Examples of such terms are $\Pair(x,y)$, which evaluates to $\{x,y\}$, or $\Union(x) = \bigcup_{y \in x} y$. Furthermore, if $s$ and $t$ are terms, $x$ is a variable, and $\phi$ a formula, then $\{ s(x) \ : \ x \in t \ : \ \phi(t) \}$ is a comprehension term. It applies the term $s$ to all elements of the set defined by $t$ that satisfy $\phi$, and outputs the set of the resulting objects $s(x)$. When a program is evaluated in a given finite structure $A$, then the term $\Pi_{\text{step}}(x)$ is iteratively applied to its own output, starting with $x_0 = \emptyset$. The iteration stops in step $i$ if the computed set $x_i = (\Pi_{\text{step}})^i(\emptyset)$ satisfies $A \models \Pi_{\text{halt}}(x_i)$. The formula $\Pi_{\text{out}}$ defines, in dependence of $x_i$, whether the run is accepting or rejecting, that is, whether $A \models \Pi$ or not. If the length of the run or the size of the transitive closure of $x_i$ exceeds $p(|A|)$ at some point, then the computation is aborted, and $A \not\models \Pi$.\\

The h.f.\ sets that appear in the run of a program $\Pi$ on a structure $\AA$ are called the sets that are \emph{activated} by $\Pi$ on input $\AA$. Formally, the set of active objects is the union over the transitive closures of all the iteration stages $x_i$. The precise definition is not important for the purposes of this article and there exist multiple slightly varying definitions in the literature \cite{dawar2008, rossman2010choiceless, svenja} which all essentially describe the same concept.\\

\paragraph*{Symmetry groups}
The two key properties of CPT that we exploit for lower bounds are its polynomial boundedness and symmetry-invariance.
For a structure $\AA$ with universe $A$, we denote by $\Aut(\AA) \leq \Sym(A)$ its automorphism group. Any $\pi \in \Sym(A)$ also acts naturally on $\HF(A)$ by renaming the atoms of the h.f.\ set according to $\pi$. With this, the symmetry-invariance of CPT can be summarised as follows: 
\begin{proposition}
	\label{prop:symmetryInvarianceCPT}
	Let $\AA$ be a structure, $x \in \HF(A)$ a hereditarily finite set over $\AA$, $\pi \in \Aut(\AA)$ an automorphism. Any CPT-program that activates $x$ in its run on $\AA$ also activates $\pi(x)$.
\end{proposition}	
We omit the proof because this fact is well-known and follows simply from the fact that the construction steps of the h.f.\ sets are logically definable. As a consequence of this proposition, CPT-definable objects are closed under their $\Aut(\AA)$-\emph{orbits}. The orbit of a set $x \in \HF(A)$ is the set $\{\pi(x) \mid \pi \in \Aut(\AA)\}$ of all its images under $\Aut(\AA)$. Since CPT-definable objects must also obey a polynomial size bound, any object whose orbit size is super-polynomial in $|A|$ cannot be activated by any CPT-program on input $\AA$. This consideration is also reflected in the circuit properties from Theorem \ref{thm:mainInformal}.\\

An important fact from group theory that we sometimes need in this context is the Orbit-Stabiliser Theorem. Applied to h.f.\ sets over structures, it reads as follows.
\begin{proposition}
	\label{prop:orbitStabiliser}
	Let $\AA$ be a structure, $x \in \HF(A)$ a hereditarily finite set over $\AA$.
	Then $|\Orbit(x)| = \frac{|\Aut(\AA)|}{|\Stab(x)|}$.
\end{proposition}	
 Here, $\Stab(x) = \{ \pi \in \Aut(\AA) \mid \pi(x) = x \}$ denotes the subgroup of the automorphism group that fixes $x$. More generally, if $H$ is a subgroup of $G$ (denoted $H \leq G$), then the \emph{index} of $H$ in $G$ is $[G: H] = |G|/|H|$. This is equal to the number of cosets of $H$ in $G$, and if $H$ is the stabiliser of some object $x$, then $[G: H]$ is exactly the orbit size of $x$. 
 
 \paragraph*{Linear algebra}
 We denote by $\bbF_2$ the finite field with two elements $\{0,1\}$, and by $\bbF_2^n$ the set of $n$-tuples over $\bbF_2$, viewed as an $n$-dimensional vector space. We will also be dealing with spaces indexed by some finite (unordered) set $J$, in which case we write $\bbF_2^J$ for the $|J|$-dimensional vector space whose coordinates are the elements of $J$. Similarly, the rows and columns of a matrix can be indexed with such finite unordered sets, so a matrix $M \in \bbF_2^{I \times J}$ describes a linear transformation $M: \bbF_2^J \lra \bbF_2^I$. The \emph{image} of $M$ is $\Im(M) = \{ M \cdot \mathbf{v} \mid \mathbf{v} \in \bbF_2^J \}$, and the \emph{kernel} is $\Ker(M) = \{\mathbf{v} \in \bbF_2^J \mid M \cdot \mathbf{v} = \mathbf{0}\}$. The \emph{Rank Theorem} states that $|J| = \rk(M) + \dim(\Ker(M))$. The \emph{rank} $\rk(M)$ denotes the dimension of $\Im(M)$. It is equal both to the dimension of the space spanned by the column vectors of $M$ and the dimension of the space spanned by the row vectors.\\

\section{Unordered Cai-Fürer-Immerman graphs}
\label{sec:CFI}
Fix an undirected (and unordered) connected graph $G = (V,E)$ as the \emph{base graph} for the CFI-construction. We turn $G$ into a CFI-graph by replacing the edges with certain edge-gadgets and the vertices with vertex-gadgets. There are two types of vertex-gadgets, called odd and even. To construct a concrete CFI-graph over $G$, we have to fix a set $S \subseteq V$ of vertices which are replaced by the \emph{odd} gadget. The vertices in $V \setminus S$ will be turned into the \emph{even} gadget. Following the notation in \cite{dawar2008}, we denote the resulting CFI-graph by $\GG^S$. The precise definition is as follows:
Let $\widehat{E} := \{ e_0, e_1 \mid e \in E \}$.
These are the vertices that will form the edge-gadgets of $\GG^S$, so there are two vertices per edge-gadget. To define the vertices in vertex-gadgets, we let, for each $v \in V$, 
\begin{align*}
	v^*_S := \begin{cases} 
		\{  v^X \mid X \subseteq E(v), |X| \text{ even }   \} & \text{ if } v \notin S\\
		\{  v^X \mid  X \subseteq E(v), |X| \text{ odd }  \} & \text{ if } v \in S
	\end{cases}
\end{align*}
Here, $E(v) \subseteq E$ are the edges incident to $v$ in $G$. The vertices in $v^*_S$ form the vertex-gadget of $v$. In total, we let
\[
\widehat{V}_S := \bigcup_{v \in V} v^*_S.
\]
Then the vertex-set of $\GG^S$ is $V(\GG^S) := \widehat{V}_S \cup \widehat{E}$. The edges of the CFI-graph are given by
\[
E(\GG^S) := \{ \{ v^X, e_i  \} \mid v^X \in \widehat{V}_S, e_i \in \widehat{E}, |X \cap \{e\}| = i \} \cup \{ \{e_0,e_1\} \mid e \in E  \}.
\]
In other words, for every $v \in V$, we connect each $v^X \in v^*_S$ with the edge-gadgets of all edges $e \in E(v)$ in such a way that $v^X$ is connected with $e_0$ if $e \notin X$, and otherwise with $e_1$. Also, we connect $e_0$ and $e_1$ to ensure that no automorphism of $\GG^S$ can tear apart the edge-gadgets. Our CFI-graphs are unordered, so the only relation of the structure $\GG^S$ is the edge relation $E$.\\

Below are the gadgets $v^*_S, w^*_S$ for two vertices $v,w \in V$, and the gadget for the edge $e \in E$ connecting them. In this example, we have $v \notin S, w \in S$, and $E(v) = \{e,f,g\}, E(w) = \{e,h,i\}$. Only the edge $e$ is drawn. Notice that $v^*_S$ and $w^*_S$ look the same when we only consider their connections to the $e$-gadget, even though one gadget is even and the other is odd.
	\begin{figure}[H]
		\centering
		\begin{tikzpicture}[dot/.style={draw,circle,minimum size=1.5mm,inner sep=0pt,outer sep=0pt,fill=blue},circ/.style={draw,circle,minimum size=2.5mm,inner sep=0pt, fill=red},
			circY/.style={draw,circle,minimum size=2.5mm,inner sep=0pt, fill=yellow}]
			\node[dot,label=below:$e_0$] (e0) at (0,0) {};
			\node[dot,label=$e_1$] (e1) at (0,1) {};
			
			\draw[radius=1,fill=red, opacity=0.2] (-3.75,0.5) circle;
			\draw[radius=1,fill=yellow, opacity=0.2] (3.75,0.5) circle;

			\node[circ,label=$v^\emptyset$] (v0) at (-3.75,-0.2) {};
			\node[circ,label=$v^{\{e,f\}}$] (vef) at (-3.75,1.2) {};
			\node[circ,label=$v^{\{e,g\}}$] (veg) at (-4.5,0.5) {};
			\node[circ,label=$v^{\{f,g\}}$] (vfg) at (-3,0.5) {};
			
			\node[circY,label=$w^{\{e\}}$] (we) at (3.75,1.2) {};
			\node[circY,label=$w^{\{h\}}$] (wh) at (3.75,-0.2) {};
			\node[circY,label=$w^{\{i\}}$] (wi) at (4.5,0.5) {};
			\node[circY,label=$w^{\{e,h,i\}}$] (wehi) at (3,0.5) {};
			
			\draw (v0) -- (e0);
			\draw (vef) -- (e1);
			\draw (veg) -- (e1);
			\draw (vfg) -- (e0);
			
			\draw (we) -- (e1);
			\draw (wh) -- (e0);
			\draw (wi) -- (e0);
			\draw (wehi) -- (e1);
			
			\draw (e0) -- (e1);
			
		\end{tikzpicture}
		\caption{Gadgets $v^*_S, w^*_S$, connected by the gadget for the edge $e$.}
	\end{figure}
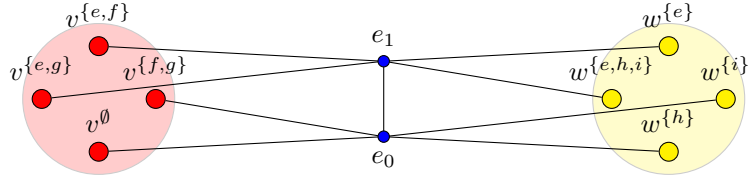
The \emph{CFI-query} asks for the parity of $|S|$, given a CFI-graph $\GG^S$. This is essentially the same question as the graph isomorphism problem for CFI-graphs: 
\begin{theorem}[\cite{caifurimm92} \cite{dawar2008}]
	For two given CFI-graphs over the same base graph, it holds
	\[
	\GG^S \cong \GG^R \text{ if and only if } |S| \equiv |R| \mod 2.
	\]
\end{theorem}	
Alternatively, deciding the parity of $|S|$ can be phrased as a linear equation system over $\bbF_2$ in the variables $\widehat{E}$ (see \cite{atseriasBulatovDawar}). Since the reduction to a linear equation system is easily computable from the given CFI-graph $\GG^S$, and linear equation systems can be efficiently solved using, for example, Gaussian elimination, the CFI-query is decidable in polynomial time.\\

For logics that lack the ability to create higher-order objects, such as bounded-variable counting logic and hence FPC, it is provably impossible to distinguish non-isomorphic CFI-graphs, provided that the treewidth of the base graphs is super-constant:
\begin{theorem}[\cite{caifurimm92} \cite{atseriasBulatovDawar}]
	\label{thm:CFI_cfiTheorem}
	Let $G = (V,E)$ be an undirected connected graph with treewidth $t$. Then for any two sets $S, S' \subseteq V$, it holds 
	$
	\GG^S \equiv_{\Cc^t} \GG^{S'}, 
	$
	even if $\GG^S \not\cong \GG^{S'}$.
\end{theorem}	
This holds because Duplicator has a winning strategy in the bijective $t$-pebble game on $\GG^S$ and $\GG^{S'}$. Intuitively, the difference between $\GG^S$ and $\GG^{S'}$ manifests itself in one single edge whose gadget is twisted, and the aim of Duplicator is to move this twist around in such a way that it is never exposed by the $t$ pebbles. This can be achieved by playing similarly as the robber in the cops and robbers game which witnesses the treewidth to be at least $t$.

\subsection{Automorphisms of unordered CFI-graphs}
\label{sec:CFI_automorphisms}
For a CFI-graph $\GG^S$ over an \emph{unordered} base graph $G = (V,E)$, two different kinds of automorphisms play a role: 
Firstly, there are what we call ``CFI-automorphisms'' or -isomorphisms. These are induced by swapping $e_0$ and $e_1$ in some edge-gadgets (this is called ``flipping the edge''). Secondly, there are the automorphisms of the underlying graph $G$ itself.\\

To speak about the CFI-isomorphisms, we use the terminology from \cite{dawar2008}: For a given base graph $G$, we consider not only a concrete CFI-instance with odd and even vertex gadgets, but we can also construct the ``full'' CFI-graph $\GG$, in which every vertex gadget is both even and odd. Formally, for $v \in V$, let $v^* := v^*_\emptyset \cup v^*_{\{v\}} = \{  v^X \mid X \subseteq E(v) \}$, and $\widehat{V} := \bigcup_{v \in V} v^*$.
The vertex-set of $\GG$ is $\widehat{V} \cup \widehat{E}$, and the edge-set is
\[
E(\GG) := \{ \{ v^X, e_i  \} \mid v^X \in \widehat{V}, e_i \in \widehat{E}, |X \cap \{e\}| = i \} \cup \{ \{e_0,e_1\} \mid e \in E  \}. 
\]
Every CFI-instance $\GG^S$ is an induced subgraph of $\GG$.\\
For each edge $e = \{v,w\} \in E$, let $\rho_e$ denote the automorphism of $\GG$ induced by flipping the edge $e$. Formally, $\rho_e(e_0) = e_1, \rho_e(e_1) = e_0$, and $\rho_e(v^X) = v^{X \triangle \{e\}}$, $\rho_e(w^X) = w^{X \triangle \{e\}}$ for all $v^X, w^X \in v^* \cup w^*$. All other vertices in $\widehat{V}$ are fixed by $\rho_e$. One can check that this is indeed an automorphism of $\GG$; furthermore, $\rho_e$ is an \emph{isomorphism} from any CFI-instance $\GG^S$ to $\GG^{S \triangle \{v,w\}}$ (see also \cite{dawar2008}).
It is easy to see that these edge-flip automorphisms commute, so for $F = \{e^1,...,e^m\} \subseteq E$ we may write $\rho_F$ for $\rho_{e^1} \circ \rho_{e^2} \circ ... \circ \rho_{e^m}$. So in total, for every $F \subseteq E$, $\rho_F$ is an automorphism of $\GG$. For any edge-set $F \subseteq E$, and $v \in V$ let $\deg_F(v) := | E(v) \cap F|$, i.e.\ the number of incident edges that are in $F$. We have $\rho_F(\GG^S) = \GG^{S \triangle T}$, where $T = \{v \in V \mid \deg_F(v) \text{ is odd } \}.$
In particular, if every $v \in V$ is incident to an even number of edges in $F$ (so $F$ is the symmetric difference over a set of cycles in $G$), then $\rho_F$ is also an automorphism of $\GG^S$, not only of $\GG$. 
To sum up, we have the following groups of CFI-automorphisms of $\GG$ and $\GG^S$:
\[
\Aut_{\CFI}(\GG) := \{\rho_F \mid F \subseteq E \}.
\]
This group is isomorphic to the Boolean vector space $\bbF_2^E$: Each $F \subseteq E$ is identified with its characteristic vector $\chi(F) \in \bbF_2^E$. It holds $\rho_F \circ \rho_{F'} = \rho_{F \triangle F'}$, and this corresponds to the vector $\chi(F) + \chi(F') \in \bbF_2^E$.\\
As already said, for a CFI-instance $\GG^S$, i.e.\ an induced subgraph of $\GG$, we have
\[
\Aut_{\CFI}(\GG^S) := \{ \rho_F \in \Aut_{\CFI}(\GG) \mid  \deg_F(v) \text{ is even for every } v \in V \}.
\]
This group is isomorphic to a subspace of $\bbF_2^E$.
In addition to the CFI-automorphisms, we also have to consider $\Aut(G) \leq \Sym(V)$, i.e.\ the automorphism group of the unordered underlying graph; this is different from the typical scenario studied in the literature, where $G$ is ordered and so the automorphisms of $\GG^S$ are just given by the edge-flips.
In total, the automorphism group of the full CFI-graph $\GG$ is isomorphic to the following semi-direct product: 
\[
\Aut(\GG) \cong \Aut_{\CFI}(\GG) \rtimes \Aut(G)  = \{ (\rho_F, \pi) \mid \rho_F \in \Aut_{\CFI}(\GG), \pi \in \Aut(G)  \}.
\]
The action of a pair $(\rho_F, \pi)$ on $V(\GG)$ is determined by its action on $\widehat{E}$: Let $e_i \in \widehat{E}$ with $i \in \{0,1\}$ and $e = \{v,w\} \in E$. Then $(\rho_F,\pi)(e_i) = f_j$, where $f = \{\pi(v), \pi(w)\}$, and $j = i + |F \cap \{e\}| \mod 2$. This action on $\widehat{E}$ extends to an automorphism of $\GG$ in a unique way. The automorphism group $\Aut(\GG^S)$ of a concrete CFI-instance is a subgroup of this.

\subsection{Symmetries and supports of hereditarily finite sets over CFI-structures}
\label{sec:CFIsupports}

Let $\GG^S$ be a CFI-graph over $G = (V,E)$ and $x \in \HF(\widehat{E})$. 
We only consider objects in $\HF(\widehat{E})$ instead of $\HF(\widehat{E} \cup \widehat{V})$ because this is easier and atoms $v^X \in \widehat{V}$ can be uniquely described by their neighbourhood in $\GG^S$, which is again a set over the atoms $\widehat{E}$. Therefore, we can pretend that any CPT-algorithm for the CFI-query only uses h.f.\ sets over $\widehat{E}$.\\
The automorphism group $\Aut(\GG^S)$, as well as the edge-flip-group $\Aut_{\CFI}(\GG)$, and the automorphisms of the base graph $\Aut(G)$ act on $\widehat{E}$ and therefore also on $\HF(\widehat{E})$: For example, let $\pi \in \Aut(G)$, and $x \in \HF(\widehat{E})$. Then $\pi x = \{ \pi y \mid y \in x  \}$. If $x$ is an atom $e_i$, with $i \in \{0,1\}$ and $e \in E$, then $\pi x = \pi(e)_i$.\\
An automorphism $\pi \in \Aut(G)$ \emph{stabilises} an object $x \in \HF(\widehat{E})$, if $\pi x = x$. More precisely, this means that $\pi$, which acts on the atoms of $x$, extends to some automorphism $\sigma$ of the DAG-structure $(\tc(x), \in)$, such that for every atom $e_i \in \tc(x)$, $\sigma(e_i) = \pi(e)_i$. As already said, $\Aut(\GG^S)$ is composed of edge flips and automorphisms of the base graph. We separate the effect of these two subgroups on the elements of $\HF(\widehat{E})$ and consider the following orbits and stabilisers for $x \in \HF(\widehat{E})$. Since the group of edge flips, $\Aut_{\CFI}(\GG)$, is isomorphic to the Boolean vector space $\bbF_2^E$, we often identify an automorphism $\rho_F \in \Aut_{\CFI}(\GG)$ with its characteristic vector $\chi(F) \in \bbF_2^E$.
\begin{align*}
	\Orb_E(x) &:= \{ \rho_F(x) \mid \rho_F  \in \Aut_{\CFI}(\GG) \}.\\
	\Stab_E(x) &:= \{ \chi(F) \mid \rho_F \in \Aut_{\CFI}(\GG), \rho_F(x) = x \} \leq \bbF_2^E.\\
	\Orb_{\CFI}(x) &:= \{ \rho_F(x) \mid \rho_F  \in \Aut_{\CFI}(\GG^S) \}.\\
	\Stab_{\CFI}(x) &:= \{ \chi(F) \mid \rho_F \in \Aut_{\CFI}(\GG^S), \rho_F(x) = x \} \leq \bbF_2^E.\\
	\Orb_{\GG^S}(x) &:= \{(\rho_F,\pi)(x) \mid (\rho_F,\pi) \in \Aut(\GG^S) \}.\\
	\Stab_{\GG^S}(x) &:= \{ (\rho_F,\pi) \in \Aut(\GG^S) \mid (\rho_F,\pi)(x) = x \}.\\
	\Orb_{G}(x) &:= \{(\rho_\emptyset,\pi)(x) \mid \pi \in \Aut(G) \}.\\
	\Stab_{G}(x) &:= \{ \pi \in \Aut(G) \mid (\rho_{\emptyset},\pi)(x) = x \}.\\
	\maxOrb_E(x) &:= \max_{y \in \tc(x)} |\Orb_E(y)|.\\
	\maxOrb_\CFI(x) &:= \max_{y \in \tc(x)} |\Orb_\CFI(y)|.
\end{align*}
It should be emphasised that $\Stab_E(x)$ is a subspace of $\bbF_2^E$, so it makes sense to speak about its dimension and to apply linear transformations to it. 
At this point, we observe for future reference that all objects in the same $\Aut_{\CFI}(\GG)$-orbit have the same $\Aut_{\CFI}(\GG)$-stabiliser because the group is Abelian:
\begin{lemma}
	\label{lem:XOR_stabiliserEqual}
	Let $x, x' \in \HF(\widehat{E})$ such that $x' = \rho_F(x)$, for some $\rho_F \in \Aut_{\CFI}(\GG)$. Then $\Stab_E(x) = \Stab_E(x')$.
\end{lemma}
\begin{proof}
	We have $\Stab_E(x') = \{ \chi(F) + \alpha + \chi(F) \mid \alpha \in \Stab_E(x) \}$. Since $\chi(F) = \chi(F)^{-1}$ and $\bbF_2^E$ is Abelian, we have $\Stab_E(x') = \Stab_E(x)$. 
\end{proof}

We also observe that the two different $\maxOrb$-parameters of a h.f.\ set can be bounded in terms of the size and orbit size of $x$.
\begin{lemma}
	\label{lem:XOR_maxOrbBound}
	Let $x \in \HF(\widehat{E})$. Then $\maxOrb_E(x) \leq |\Orb_E(x)| \cdot |\tc(x)|$ and\\ $\maxOrb_\CFI(x) \leq |\Orb_\CFI(x)| \cdot |\tc(x)|$.
\end{lemma}
\begin{proof}
	Let $y \in \tc(x)$ be the set where $\maxOrb_E(x)$ is attained, i.e.\ $|\Orb_E(y)| = \maxOrb_E(x)$. Let $Y := \Orb_E(y) \cap \tc(x)$. Clearly, for any $\rho_F \in  \Aut_{\CFI}(\GG)$, $\rho_F(Y) \subseteq \tc(\rho_F(x))$, and $\rho_F(x) \in \Orb_E(x)$. Hence:
	\[
	|\Orb_E(y)| \leq   \Big| \bigcup_{\rho_F\in \Aut_{\CFI}(\GG)} \rho_F(Y) \Big| \leq |\Orb_E(x)| \cdot |\tc(x)|.
	\]
	Similarly, the statement for $\maxOrb_{\CFI}(x)$ is proven.
\end{proof}

In \cite{dawar2008}, the term \emph{super-symmetry} was introduced for h.f.\ sets which are fixed by all automorphisms in $\Aut_{\CFI}(\GG)$. The notion of super-symmetry that is needed for our Theorem \ref{thm:mainInformal} can be relaxed a bit, in the sense that the orbit size w.r.t.\ $\Aut_{\CFI}(\GG)$ does not need to be exactly one.
\begin{definition}[Super-symmetric objects]
	\label{def:XOR_superSymmetricObjects}
	Fix a family of CFI-graphs $(\GG^S_n)_{n \in \bbN}$ and a $\mu_n \in \HF(\widehat{E}_n)$ for every $n$. The objects $\mu_n$ are \emph{super-symmetric} if 
	$
	|\Orb_E(\mu_n)| \leq \text{\upshape poly}(|\GG^S_n|).
	$
\end{definition}	

If a set is super-symmetric and CPT-definable, then we have a handle on its parameter $\maxOrb_{E}$:
\begin{lemma}
	\label{lem:XOR_maxOrbBoundSuperSymmetric}
	If $\mu_n$ is \emph{super-symmetric} and \emph{CPT-definable} in $\GG^S_n$, then $\maxOrb_{E}(\mu_n)$ is polynomially bounded in $|\GG^S_n|$. 
\end{lemma}	 
\begin{proof}
	By super-symmetry, $|\Orb_E(\mu_n)|$ is polynomially bounded. By CPT-definability, $|\tc(\mu_n)|$ is polynomially bounded. Hence the statement follows with Lemma \ref{lem:XOR_maxOrbBound}.
\end{proof}	

So super-symmetric objects in this sense satisfy the same orbit bound with respect to the bigger symmetry group $\bbF_2^E \cong \Aut_{\CFI}(\GG)$ as all CPT-definable objects naturally do with respect to the automorphism group of the input structure. 

\paragraph*{Supports for CFI-automorphisms}
Generally, a \emph{support} of a permutation group $\Gamma \leq \Sym(A)$ is a subset $S \subseteq A$ such that the pointwise stabiliser of $S$ in $\Sym(A)$ is a subgroup of $\Gamma$. A support of a h.f.\ set is a support of its stabiliser group. For subgroups of $\Aut_{\CFI}(\GG)$, we will use a different notion, that we call \emph{CFI-support}. The reason why we need a specific type of support for these groups is because otherwise, the group $\Aut_{\CFI}(\GG)$ does not admit unique minimum supports.
\begin{definition}[CFI-support]
	\label{def:CFI_support}
	A \emph{CFI-support} of an object $x \in \HF(\widehat{E})$ is a subset $S \subseteq E$ such that every $\rho_F \in \Aut_{\CFI}(\GG)$ with $F \cap S = \emptyset$ fixes $x$. 
\end{definition}	
There is always a unique minimal CFI-support:
\begin{lemma}
	\label{lem:XOR_minimalCFIsupports}
	Let $\GG^S$ be a CFI-instance and $x \in \HF(\widehat{E})$. Let $A_1,A_2 \subseteq E$ be CFI-supports of $x$. Then $A_1 \cap A_2$ is also a CFI-support of $x$. 	
\end{lemma}
\begin{proof}
	Assume $A_1 \cap A_2$ was not a CFI-support of $x$. Then there is $F \subseteq E$ disjoint from $A_1 \cap A_2$ such that $\rho_F(x) \neq x$. Let $F_1 := F \cap A_1$ and $F_2 := F \cap A_2$. These sets are both non-empty, because: If $F$ did not intersect $A_1$, then $\rho_F(x) = x$ because $A_1$ is a CFI-support for $x$. Similarly for $A_2$. Also, by assumption, $F_1$ and $F_2$ are disjoint from $A_1 \cap A_2$, and therefore, also $F_1 \cap F_2 = \emptyset$. Furthermore, $F' := F \setminus (F_1 \cup F_2)$ is disjoint from $A_1 \cup A_2$ and therefore, $\rho_{F'}$ fixes $x$. It follows that flipping the edges in $F_1 \cup F_2$ moves $x$, because by assumption, $\rho_F(x) \neq x$. But this is a contradiction because $\rho_{F_1}(x) = x$ (since $F_1$ is disjoint from the support $A_2$), and analogously, $\rho_{F_2}(x) = x$.   
\end{proof}
This justifies the following definition:
\begin{definition}[Minimal CFI-support]
	\label{def:XOR_minCFIsupport}
	\noindent For $x \in \HF(\widehat{E})$,
	$\sup_{\CFI}(x) \subseteq E$ denotes
	the unique minimal subset of $E$ that is a CFI-support of $x$.
\end{definition}

As explained above, for general CPT-definable objects, the orbit size can only be bounded when the ``true'' automorphism group $\Aut_{\CFI}(\GG^S)$ is considered. Only when the object is additionally super-symmetric, also the $\Aut_{\CFI}(\GG)$-orbit size is polynomial.
However, depending on the structure of the base graph, we can sometimes bound the $\Aut_{\CFI}(\GG)$-orbit size as well, even if the object is not necessarily super-symmetric:
\begin{lemma}
	\label{lem:XOR_supNotBridge}
	Let $\mu \in \HF(\widehat{E})$ be a h.f.\ set over $\GG^S$. Let $k$ be the number of connected components in the graph $G - \sup_{\CFI}(\mu)$ (i.e.\ the base graph after removing the edges in the minimum CFI-support).
	Then $|\Orb_{E}(\mu)| \leq 2^{k^2}\cdot |\Orb_{\GG^S}(\mu)|$.
\end{lemma}
\begin{proof}
	Let $A = \sup_{\CFI}(\mu) \subseteq E$ be the smallest CFI-support of $\mu$. Then for every $\rho_F \in \Aut_{\CFI}(\GG)$ with $F \cap A = \emptyset$, it holds that $\rho_F(\mu) = \mu$. Thus, the effect of an edge-flip $\rho_F$ on $\mu$ depends only on $F \cap A$. So we have:
	\[
	|\Orb_{E}(\mu)| \leq 2^{|A|},
	\]
	because there are $2^{|A|}$ ways how any $F \subseteq E$ can intersect the support $S$, and if $F \cap A = F' \cap A$, then also $\rho_F(\mu) = \rho_{F'}(\mu)$.\\
	Now we compute a lower bound on  $|\Orb_{\GG^S}(\mu)|$ by analysing how many subsets of $A$ can occur as the intersection $F \cap A$ for an automorphism $\rho_F \in \Aut_{\CFI}(\GG^S)$. In contrast to the edge-flips in $\Aut_{\CFI}(\GG)$, these are the edge-flips along cycles in $G$. Let $X_1,...,X_k \subseteq V(G)$ denote the vertex-sets of the connected components in the graph $G - A$. We partition the edge-set $A$ into at most $k^2$ many subsets, according to the components that the edges connect. So for each pair $(i,j) \in [k]^2$, let $\Aa_{(i,j)} \subseteq A$ denote those edges in $A$ that run between the components $X_i$ and $X_j$. Now it can be seen that for every pair $i \neq j$, for every $B \subseteq \Aa_{(i,j)}$ of even cardinality, there exists some symmetric difference of cycles $C_B$ in $G$ whose intersection with $A$ is exactly $B$. This is because any two edges $e,e' \in \Aa_{(i,j)}$ lie on a cycle through the components $X_i$ and $X_j$. For $\Aa_{(i,i)} \subseteq A$, \emph{every} subset $B \subseteq \Aa_{(i,i)}$ can be generated by the symmetric difference of some cycles because the endpoints of every $e \in \Aa_{{i,i}}$ are in the same connected component (but for simplicity, we pretend that also in this case, only the even subsets of $\Aa_{(i,i)}$ can be hit by the symmetric difference of some cycles). Summing up these considerations, we have:
	\begin{align*}
		&|\{  B \subseteq A \mid \text{ there exists a } \rho_F \in \Aut_{\CFI}(\GG^S) \text{ such that } F \cap A = B \}|\\
		&\geq \prod_{(i,j) \in [k]^2} 2^{|\Aa_{(i,j)}|-1} = 2^{|A|-k^2}.
	\end{align*}
	Let $s$ denote the number of $B \subseteq A$ in the above set such that flipping $B$ stabilises $\mu$. Then by the Orbit-Stabiliser Theorem, we have $|\Orb_{\GG^S}(\mu)| \geq 2^{|A|-k^2}/s$ and $|\Orb_E(\mu)| \leq 2^{|A|}/s$. Putting these two inequalities together, we get the desired bound $|\Orb_{E}(\mu)| \leq 2^{k^2}\cdot |\Orb_{\GG^S}(\mu)|$.
\end{proof}

This lemma essentially says that it does not make a difference whether we consider orbit-sizes with respect to the group $\Aut_{\CFI}(\GG)$ of all edge flips or the group $\Aut_{\CFI}(\GG^S)$ of cycle edge flips, as long as the CFI-support of an object separates the base graph only into a small number of components:

\begin{corollary}
	\label{cor:XOR_relationOrbitSizesEdgeFlips}
	Fix a family $(\GG_n^S)_{n \in \bbN}$ of CFI-structures. Let $\mu_n \in \HF(\widehat{E}_n)$. If the number of connected components in $G_n - \sup_{\CFI}(\mu_n)$ is at most $\Oo(\sqrt{\log |\GG^S_n|})$, then $|\Orb_{E}(\mu_n)| \leq \text{\upshape poly}(|\GG^S_n|) \cdot |\Orb_{\GG^S}(\mu)|$.
\end{corollary}

\section{CFI-symmetric hereditarily finite sets and algorithms}
\label{sec:CFIsymmetricSets}
The CFI-query is definable in CPT on instances that arise from linearly ordered base graphs, base graphs that come with a preorder with colour classes of logarithmic size, and base graphs of linear degree \cite{dawar2008, pakusaSchalthoeferSelman}. All these CPT-algorithms depend on the construction of a particular super-symmetric h.f.\ set $\mu \in\HF(\widehat{E})$ that encodes the parity of $|S|$, given an instance $\GG^S$. We isolate another property of these h.f.\ sets, besides super-symmetry, which is responsible for their small orbit size and suitability for encoding parities. We call this \emph{CFI-symmetry}. Intuitively, a set $\mu \in \HF(\widehat{E})$ is CFI-symmetric if its ``building blocks'' behave similarly as CFI-gadgets in CFI-graphs, in the sense that they are ``flipped'' whenever an even number of ``incident gadgets'' is flipped. These building blocks are the \emph{connected components} of sets. To define these, let a CFI-graph $\GG^S$ and a set $\mu \in \HF(\widehat{E})$ be fixed, and let $\sim_E$ be the following equivalence relation on the elements $x \in \tc(\mu)$: For $x,x' \in \tc(\mu)$, we write $x \sim_E x'$ iff there exists an edge-flip $\rho_F \in \Aut_{\CFI}(\GG)$ such that $x' = \rho_F(x)$. 
The $\sim_E$-equivalence class in $\tc(\mu)$ of an object $x \in \tc(\mu)$ is denoted $[x]_{\sim_E}$ or $[x]$.
The relation $\sim_E$ induces a partition $\Cc(x)$ on each $x \in \tc(\mu)$, namely $\Cc(x) := \{([z]_{\sim_E} \cap x) \mid z \in x\}$. In \cite{dawar2008}, the elements of $\Cc(x)$ are called the \emph{connected components} of $x$. Now in a \emph{CFI-symmetric} object, each connected component $\gamma \in \Cc(x)$, for each $x \in \tc(\mu)$, behaves like a CFI-gadget. That is, the component has exactly two images under $\Aut_{\CFI}(\GG)$: It can either be flipped or stabilised, and which of these two is the case, depends on the parity of flipped components in the elements of $\gamma$. Before we formalise this, we consider an example of a small ``parity-tracking'' h.f.\ set that is constructed similarly as in the algorithms from \cite{dawar2008} and \cite{pakusaSchalthoeferSelman}.
\begin{example}
	\label{ex:XOR_CFIsymmetricSet}
	Here is an example h.f.\ set $\mu_{\{e,f,g\}} \in \HF( \widehat{E} )$ with $E = \{e,f,g\}$. It tracks the parity of edge-flips for the edges $e,f,g$. For better readability, the set is printed in a structured form, so the sets $\mu_{\{f,g\}}$ and $\widetilde{\mu}_{\{f,g\}}$ are shown in the level below.
	\begin{figure}[H]
		\begin{tikzpicture}
			\node (mu) at (0,0) {$\mu_{\{e,f,g\}} = $};
			\node (muSet) [right = 0.1cm of mu] {$\Big\{ \  \{  \mu_{\{f,g\}} , e_0  \}, \{  \widetilde{\mu}_{\{f,g\}}  , e_1 \}  \ \Big\}$};
			\node (mufg) [below  = 0.4cm of muSet ] {$\{\{ f_0, g_0 \}, \{ f_1,g_1\} \}$  };
			\node (muTildefg) [below right = 0.4cm and 0.1cm of muSet] { $\{ \{ f_0, g_1 \}, \{ f_1,g_0\} \}$ };
			\node[draw=none] (muSetSouthWest) at (2,-0.2) {};
			\node[draw=none] (muSetSouthEast) at (4.5,-0.2) {};
			\draw (mufg) -- (muSetSouthWest);
			\draw (muTildefg) -- (muSetSouthEast);
		\end{tikzpicture}	
	\end{figure}	
	\noindent
	Each of the $\mu$-objects has only one connected component that consists of two sets which are related by $\sim_E$. For example, the two elements of $\mu_{\{e,f,g\}}$ are mapped to each other whenever an even number of edges is flipped.      These two elements of $\mu_{\{e,f,g\}}$ themselves have two connected components: Clearly, $e_0$ and $\mu_{\{f,g\}}$ cannot be mapped to each other by any edge-flip. The same goes for example for $f_0$ and $g_0$. They form distinct components of the set $\{f_0,g_0\}$, while $\{ \{f_0, g_0\}, \{ f_1, g_1 \} \}$ again only has one component that is stabilised if and only if an even number of edges in $\{f,g\}$ is flipped. This pattern of alternation between sets with two components and sets with one component is typical of the parity-tracking objects constructed by the known CFI-algorithms.\\
	Now we can observe that the objects which behave analogously to CFI-gadgets are the \emph{connected components} inside the sets, not the sets in $\tc(\mu_{\{e,f,g\}})$ themselves. For example, the sets $\{f_0, g_0\}$ and $\{ \mu_{\{f,g\}}, e_0 \}$ cannot be ``flipped'' between two states, like a CFI-gadget. Their orbit with respect to edge-flips has size four. But whenever these sets occur as elements of another set, they occur together with a counterpart from their orbit, which ensures that its connected component inside the parent set again has the ``CFI-property'': It has orbit-size two and is ``flipped'' if and only if an even number of elements are flipped. Note that the number of flipped elements is always the same in every member of a connected component. For example, in the component $\{ \{f_0, g_0 \}, \{f_1, g_1 \}  \}$, it is clear that $f_0$ is flipped iff $f_1$ is flipped and $g_0$ is flipped iff $g_1$ is; so $\{f_0, g_0  \}$ and  $\{f_1, g_1 \} $ are always affected by the same number of flips, and the same is true for $\{  \mu_{\{f,g\}} , e_0  \}$ and $\{  \widetilde{\mu}_{\{f,g\}}  , e_1 \}$. Therefore, it makes sense to view the connected components inside each set as analogues of CFI-vertex-gadgets, and the elements of each/any member of a component as its ``incident edges'', whose flips affect the ``vertex-gadget''. 
\end{example}	
Generalising the example, the structural pattern of such parity-tracking objects can be formalised like this:
\begin{definition}[CFI-symmetric components and objects]
	\label{def:XOR_CFIsymmetric}
	Let $\mu \in \HF(\widehat{E})$, $x \in \tc(\mu)$, and $\gamma \subseteq x$ be a connected component of $x$. Then we say that $\gamma$ is \emph{CFI-symmetric} if $|\Orb_{E}(\gamma)| = 2$ and for each $\rho_F \in \Aut_{\CFI}(\GG)$, it holds $\rho_F(\gamma) = \gamma$ iff for each/any $y \in \gamma$, the number of flipped components of $y$, that is $|\{ \gamma' \in \Cc(y) \mid \rho_F(\gamma') \neq \gamma'   \}|$, is even.\\
	\\
	The set $\mu$ is \emph{CFI-symmetric} if the following two conditions are satisfied:
	\begin{enumerate}
		\item For each $\rho_F \in \Aut_{\CFI}(\GG)$, it holds $\rho_F(\mu) = \mu$ iff the number of flipped components of $\mu$, that is, $|\{ \gamma \in \Cc(\mu) \mid \rho_F(\gamma) \neq \gamma  \}|$, is even.
		\item For every $x \in \tc(\mu)$, every connected component $\gamma \in \Cc(x)$ is CFI-symmetric.
	\end{enumerate}	
\end{definition}

We will never deal with objects $\mu \in \HF(\widehat{E})$ in which only some, but not all connected components of sets in $\tc(\mu)$ are CFI-symmetric. Therefore, when we speak of ``flipped components of $y$'' in the above definition, and denote these as $\{ \gamma' \in \Cc(y) \mid \rho_F(\gamma') \neq \gamma' \}$, the component $\rho_F(\gamma') \neq \gamma'$ really is the ``flip'' of $\gamma'$, because the orbit of $\gamma'$ has size exactly two.\\

We still have to show that the formulation ``each/any'' in Definition \ref{def:XOR_CFIsymmetric} is indeed justified, as we already indicated in the example.
\begin{lemma}
	\label{lem:XOR_CFIsymmetricWellDefined}
	Let $\mu \in \HF(\widehat{E})$, $x \in \tc(\mu)$, and $\gamma \subseteq x$ be a connected component of $x$. For any two $y,y' \in \gamma$ and every $\rho_F \in \Aut_{\CFI}(\GG)$, it holds 
	\[
	|\{ \gamma' \in \Cc(y) \mid \rho_F(\gamma') \neq \gamma'   \}| = |\{ \gamma' \in \Cc(y') \mid \rho_F(\gamma') \neq \gamma'   \}|.
	\]
\end{lemma}	
\begin{proof}
	Fix $\rho_F \in \Aut_{\CFI}(\GG)$. Further, let $\rho \in \Aut_{\CFI}(\GG)$ be an automorphism such that $\rho(y) = y'$. This exists because $y \sim_E y'$. Thus, $\rho$ induces a bijection from $\Cc(y)$ to $\Cc(y')$, as it maps each connected component of $y$ to a connected component of $y'$.
	We show that for each component $\gamma' \in \Cc(y)$ it holds: $\rho_F(\gamma') = \gamma'$ iff $\rho_F(\rho(\gamma')) = \rho(\gamma')$. If $\rho_F(\gamma') = \gamma'$, then we have (because $\Aut_{\CFI}(\GG)$ is Abelian): $\rho_F(\rho(\gamma')) = \rho(\rho_F(\gamma')) = \rho(\gamma)$. Conversely, if $\rho_F(\rho(\gamma')) = \rho(\gamma')$, then $\gamma' = \rho^{-1}(\rho(\gamma')) = \rho^{-1}(\rho_F(\rho(\gamma'))) = \rho_F(\gamma')$, where the last equality is again by commutativity. The lemma follows because $|\Cc(y)| = |\Cc(y')|$.
\end{proof}

\begin{definition}[CFI-symmetric and super-symmetric algorithms]
	\label{def:XOR_CFIsymmetricAlgorithms}
	A $\CPT$-program $\Pi$ that decides the CFI-query on a class $\Kk$ of base graphs is called \emph{CFI-symmetric} if it activates a CFI-symmetric h.f.\ set $\mu \in \HF(\widehat{E})$ on every input $\GG^S$ over a base graph $G \in \Kk$ (and this set $\mu$ is necessary for deciding the CFI-query). Similarly, $\Pi$ is called \emph{super-symmetric} if it necessarily activates a super-symmetric set. When $\Pi$ is both CFI- and super-symmetric, then the relevant h.f.\ set it activates satisfies both properties at the same time.
\end{definition}	
The condition that $\mu$ is \emph{necessary} to decide the query is supposed to mean that $\Pi$ could not succeed without the activation of $\mu$. This makes sense in light of Theorem \ref{thm:DRR_supportLowerBoundCFI}, which says that the activation of a h.f.\ set with \emph{large support} is necessary to define the CFI-query in CPT. Every currently known CPT-algorithm for the CFI-query is CFI-symmetric in this sense. We do not explicitly prove this here because this would also require to include a survey on the algorithms from \cite{dawar2008} and \cite{pakusaSchalthoeferSelman} in this already long article. However, it is not too difficult to verify that the h.f.\ sets defined in those two papers exactly satisfy the design pattern that we call CFI-symmetry here.

\section{Translating hereditarily finite sets to XOR-circuits}
\label{sec:XOR_circuits}
An XOR-circuit is a connected directed acyclic graph $C = (V_C,E_C)$ with a unique designated root $r$. Its internal nodes are understood as XOR-gates and its leafs correspond to the input gates of the circuit. If $(g,h) \in E_C$, then the output of gate $h$ is an input of gate $g$. Every XOR-circuit computes the Boolean XOR-function over a subset of its input bits.\\
Such circuits are the combinatorial objects that we will use to capture the structure of the CFI-symmetric h.f.\ sets in $\HF(\widehat{E})$. When we consider these h.f.\ sets, we always view them as objects over a given CFI-structure $\GG^S$ on some base graph $G = (V,E)$. Defining them in CPT requires to preserve the symmetries of the input structure $\GG^S$, so in particular, the automorphisms $\Aut(G)$ of the base graph. This symmetry will be reflected in the symmetry of the corresponding XOR-circuit. Therefore, we have to formalise how the automorphisms of a graph $G$ act on XOR-circuits:\\

We say that an XOR-circuit $C$ is a circuit \emph{over a graph} $G = (V,E)$, if the input gates of $C$ are labelled with the edges in $E$. More precisely, let $L \subseteq V_C$ be the leafs of $C$. There is an injective labelling function $\ell : L \lra E$ that relates the input gates with edges of $G$. To speak about the semantics of the circuit, we introduce a set of formal propositional variables $\Vv(G) := \{X_e \mid e \in E\}$. For every input gate $g \in L$, the input bit of this gate is the value of the variable $X_{\ell(g)}$.\\  

Since every internal gate is an XOR-gate, the function computed by it is the XOR over a subset of $\Vv(G)$. For our purposes, this subset is the main interesting property of a gate, and we call it $\Xx(g)$.
Formally, if $g \in L$, then $\Xx(g) := \{\ell(g)\} \subseteq E$.
If $g$ is an internal gate, then 
$
\Xx(g) := \bigtriangleup_{h \in gE_C} \Xx(h),
$
that is, the symmetric difference over the $\Xx(h)$ for all children of $g$. In other words, $\Xx(g) \subseteq E$ is precisely the set of edges in $E$ such that $g$ computes the Boolean function $\bigoplus_{e \in \Xx(g)} X_e$. The function computed by the circuit $C$ is the XOR over $\Xx(r) \subseteq E$, where $r$ is the root of $C$. An alternative way to think about this is to say that for any gate $g$, $\Xx(g)$ is the set of input bits to which the function computed by $g$ is \emph{sensitive}, that is, flipping a single input bit of the circuit changes the value of $g$ if and only if the flipped edge is in $\Xx(g)$.

\subsection{Symmetries of circuits}
A circuit $C$ over a graph $G$ is subject to the action of the automorphism group $\Aut(G) \leq \Sym(V)$. Any $\pi \in \Aut(G)$ changes the labels of the input gates in $L$. So let $g \in L$ with $\ell(g) = e = \{u,v\} \in E$. Then $\pi(g)$ is an input gate with $\ell(\pi(g)) = \pi(e) = \{\pi(u), \pi(v)\} \in E$. This extends to subcircuits of $C$ and to $C$ itself, so $\pi(C)$ is just $C$ with the input labels modified accordingly.
We say that $\pi$ \emph{extends to an automorphism} of $C$ if there exists a bijection $\sigma : V_C \lra V_{C}$ that is an automorphism of the graph $(V_C,E_C)$ and satisfies for each input gate $g \in V_C$: $\ell(\sigma(g)) = \pi(\ell(g))$. We write
\[
\Stab_G(C)  = \{ \pi \in \Aut(G) \mid \pi \text{ extends to an automorphism of } C \} \leq \Aut(G),
\]
and: $\Orb_G(C) = \{ \pi(C) \mid \pi \in \Aut(G) \}$.

\subsection{The parameter fan-in dimension}
\label{sec:XOR_fanInDimension}
The XOR-circuits we will construct from CFI-symmetric h.f.\ sets will satisfy a certain fan-in bound on the gates. However, this bound will not be -- as it is more common -- on the number of incoming wires of a gate but rather, on the ``linear algebraic complexity of incoming information'', so to say.
The subsets of $E$ form a Boolean vector space together with the symmetric difference operation. This space is isomorphic to $\bbF_2^E$.\\

With each internal gate $g$ of an XOR-circuit $C = (V_C,E_C)$, we can associate a Boolean matrix $M(g) \in \bbF_2^{gE_C \times E}$, that we call the \emph{gate matrix}: The row at index $h \in gE_C$ is defined as the characteristic vector of $\Xx(h) \subseteq E$, transposed, i.e. $M(g)_{h-} = \chi(\Xx(h))^T$. Here and in what follows, we write $\chi$ for the bijection from $\Pp(E)$ to $\bbF_2^E$ that associates with each subset of $E$ its characteristic Boolean vector. If $g$ is an input gate, then we define $M(g) \in \bbF_2^{[1] \times E}$ as the one-row matrix whose only row is $\chi(\Xx(g)))^T = \chi(\{\ell(g)\})^T$.
\begin{definition}[Fan-in dimension]
	\label{def:XOR_fanInDimension}
	The \emph{fan-in dimension} of a gate $g$ is the dimension of the row-space of $M(g)$, or equivalently, $\rk(M(g))$.\\
	The fan-in dimension of $g$, restricted to the space $\Aut_{\CFI}(\GG^S)$ (also called the \emph{restricted fan-in dimension}) is
	\[
	\dim (M(g) \cdot \Aut_{\CFI}(\GG^S)) = \dim \{ M(g) \cdot \mathbf{v} \mid \mathbf{v} \in \Aut_{\CFI}(\GG^S) \}.
	\]
	The (restricted) fan-in dimension of the circuit $C$ is the maximum (restricted) fan-in dimension of any of the gates in $C$.
\end{definition}
Thus, the fan-in dimension of a gate $g$ is the dimension of the subspace of $\bbF_2^E$ that is spanned by the characteristic vectors $\chi(\Xx(h)) \in \bbF_2^E$, for all children $h$ of $g$. One interpretation of $\rk(M(g))$ is that it tells us how many different patterns of incoming bits can occur at gate $g$: When we consider all $2^{|E|}$ possible inputs of the circuit, the number of distinct binary strings in $\{0,1\}^{gE_C}$ that can arise as the values of the children of $g$ is $2^{\rk(M(g))}$.\\
Sometimes we will also need the restricted fan-in dimension.
This describes how many different input patterns of the gate can occur if we only allow circuit input vectors $\mathbf{x} \in \bbF_2^E$ where the $1$-entries in $\mathbf{x}$ form a set of cycles in the base graph $G$ (i.e.\ input vectors from $\Aut_{\CFI}(\GG^S)$).\\

These notions of fan-in dimension are unusual circuit parameters but as we will show, they nicely capture the orbit size of the original h.f.\ set with respect to the groups of edge flips $\Aut_{\CFI}(\GG)$ and $\Aut_{\CFI}(\GG^S)$. In total, the benefit of the circuit-representation of h.f.\ objects over CFI-graphs is that this simplifies the effect of the complicated automorphism group $\Aut(\GG^S) \leq \Aut_{\CFI}(\GG) \rtimes \Aut(G)$: In order to show that the CFI-query is not in CPT, we would ultimately like to prove that certain h.f.\ sets necessarily have super-polynomial orbits w.r.t.\ $\Aut(\GG^S)$. By translating these objects into circuits, we can express the restrictions imposed by $\Aut_{\CFI}(\GG)$ in terms of fan-in dimension, and are left with the task of analysing the orbit size w.r.t.\ $\Aut(G)$.

\subsection{The circuit construction}
The following theorem is the core of the proof of Theorem \ref{thm:mainInformal}. We will prove it first and then explain in the next section which further steps are needed to conclude Theorem \ref{thm:mainInformal} from it.
\begin{theorem}
	\label{thm:XOR_mainCircuitTheorem}
	Fix a family $(G_n)_{n \in \bbN}$ of base graphs. 
	For every $n \in \bbN$, let $\GG_n^S$ be a CFI-graph over $G_n = (V_n,E_n)$ and let $\mu_n \in \HF(\widehat{E}_n)$ be a \emph{CFI-symmetric} h.f.\ set that is $\CPT$-definable on input $\GG^S_n$ (by the same $\CPT$-program for the whole family of graphs). Then for every $n \in \bbN$, there exists an XOR-circuit $C(\mu_n) = (V_C, E_C)$ over $G_n$ which satisfies:
	\begin{enumerate}
		\item The size of the circuit, i.e.\ $|V_C|$, is polynomial in $|\GG^S_n|$.
		\item The orbit-size $|\Orb_G(C(\mu_n))|$ of the circuit is polynomial in $|\GG^S_n|$. 
		\item $C(\mu_n)$ is \emph{sensitive} to an edge $e \in E_n$ if and only if $e \in \sup_{\CFI}(\mu_n)$.
		\item The \emph{fan-in dimension} of $C(\mu)$ is $\Oo(\log(\maxOrb_E(\mu)))$. The fan-in dimension \emph{restricted to the space} $\Aut_{\CFI}(\GG_n^S)$ is $\Oo(\log(\maxOrb_\CFI(\mu)))$. 
	\end{enumerate}
\end{theorem}
We now provide the construction of the circuit and prove several lemmas from which it follows that the circuit has the desired properties.
We fix $\mu \in \HF(\widehat{E})$ and denote by $C(\mu) = (V_C,E_C)$ the corresponding XOR-circuit that we are going to define. The gates of the circuit are the $\sim_E$-equivalence classes of the objects in $\tc(\mu)$. Recall that $\sim_E$-equivalence is the same-orbit-relation with respect to the edge-flips $\Aut_{\CFI}(\GG)$. Whenever we write $[x]$ for an $x \in \tc(\mu)$, we formally mean $[x] = \{ \rho_F(x) \mid \rho_F \in \Aut_{\CFI}(\GG) \text{ such that } \rho_F(x) \in \tc(\mu)  \}$. The circuit $C(\mu)$ is defined as follows:
\begin{itemize}
	\item $V_C := \tc(\mu)_{\sim_E} =  \{ [x] \mid x \in \tc(\mu)  \}$.
	\item $E_C := \{ ([x], [y]) \mid \text{there exists } y' \in [y] \text{ such that } y' \in x  \}$.
	\item By definition, the leafs of $C(\mu)$ correspond to $\sim_E$-classes of atoms in $\tc(\mu)$. The set of atoms is $\widehat{E}$, so any leaf of $C$ has the form $[e_0]$, for some $e \in E$. We let $\ell([e_0]) := e$.
	\item The root $r$ of $C(\mu)$ is $[\mu]$.
\end{itemize} 
In other words, the circuit is just the DAG $(\tc(\mu),\in)$, with the $\sim_E$-equivalence factored out.\\

First of all, we have to check that the set of edges $E_C$ can indeed be defined in this way, i.e.\ that whether or not there is an $E_C$-edge between $[x]$ and $[y]$ is independent of the choice of the representative of $[x]$ in the definition. In the following lemma, let $\in^\mu$ denote the element relation on $\tc(\mu)$ within the h.f.\ set $\mu$.
\begin{lemma}
	\label{lem:XOR_inRelation}
	Let $[x], [y] \subseteq \tc(\mu)$ be two $\sim_E$ classes. If there exists $y' \in [y]$ such that $y' \in^{\mu} x$, then for every $x' \in [x]$ there is a $y' \in [y]$ such that $y' \in^{\mu} x'$.
\end{lemma}
\begin{proof}
	Let $y' \in [y]$ such that $y' \in^{\mu} x$. Now let $x' \in [x]$ be arbitrary, and let $\rho_F \in \Aut_{\CFI}(\GG)$ be such that $\rho_F(x) = x'$. Then $\rho_F(y') \in^\mu \rho_F(x)$  because the operation is applied element-wise. The set $\rho_F(x) \in [x]$ is an element of $\tc(\mu)$ because $[x]$ denotes the $\sim$-class inside $\tc(\mu)$.
	Therefore, we also have $\rho_F(y') \in \tc(\mu)$, and thus $\rho_F(y') \in [y]$. This proves the lemma.   
\end{proof}
\textbf{ Property 2} from Theorem \ref{thm:XOR_mainCircuitTheorem} states that the $\Aut(G)$-orbit of $C(\mu)$ is sufficiently small. We prove this by showing that this orbit cannot be larger than the $\Aut(G)$-orbit of the h.f.\ set $\mu$; and for $\mu$, we know that its orbit is polynomial in $|\GG^S|$, because it is CPT-definable by the assumptions of Theorem \ref{thm:XOR_mainCircuitTheorem}. 

\begin{lemma}
	\label{lem:XOR_stabiliserOfCircuit}
	Every $\pi \in \Stab_G(\mu) \leq \Sym(V)$ extends to an automorphism of the circuit $C(\mu)$, that is:
	$\Stab_G(\mu) \leq \Stab_G(C(\mu))$.
\end{lemma}
\begin{proof}
 Let $\pi \in \Stab_G(\mu) \leq \Aut(G)$. That is, $\pi$ extends to an automorphism $\sigma : \tc(\mu) \lra \tc(\mu)$ of teh DAG $(\tc(\mu),\in^\mu)$. We define $\sigma' : V_C \lra V_C$ by letting $\sigma'([x]) = [\sigma(x)]$. This is well-defined because $x \sim_E x'$ if and only if $\sigma(x) \sim_E \sigma(x')$ ($\sigma$ is an automorphism of $\mu$). Now we check that $\sigma'$ is an automorphism of $C(\mu)$ induced by $\pi$.\\
Clearly, $\sigma'$ is a bijection on $V_C$, i.e.\ on the set of $\sim_E$-classes of $\tc(\mu)$: It is surjective because $\sigma$ is, and then it is already a bijection because it maps $V_C$ to $V_C$.  Let $[e_0] \in V_C$ be an input gate. Then $\ell([e_0]) = e$. We have $\sigma(e_0) = \pi(e)_0$. So $\sigma'([e_0]) = [\pi(e)_0]$. Hence, $\ell(\sigma'([e_0])) = \pi(\ell([e_0]))$, as desired.\\
Now let $([x],[y]) \in E_C$. Then there exists a $y' \in [y]$ such that $y' \in^\mu x$. Then because $\sigma$ is an automorphism, it also holds $\sigma(y') \in^\mu \sigma(x)$. Therefore, $([\sigma(x)], [\sigma(y')]) \in E_C$. It holds $([\sigma(x)], [\sigma(y')]) =  (\sigma'[x], \sigma'[y']) = (\sigma'[x], \sigma'[y])$, so $(\sigma'[x], \sigma'[y]) \in E_C$. 
In total, this means that $\pi \in \Aut(G)$ extends to the automorphism $\sigma'$ of the circuit $C(\mu)$.
\end{proof}

\begin{corollary}
	\label{cor:XOR_orbitOfCircuit}
	$|\Orb_G(C(\mu))| \leq |\Orb_G(\mu)|.$
\end{corollary}	
\begin{proof}
	Follows from Lemma \ref{lem:XOR_stabiliserOfCircuit} together with the Orbit-Stabiliser Theorem, which says that $|\Orb_G(C(\mu))| = | \Aut(G)  | / |\Stab_G(C(\mu))|$ and $|\Orb_G(\mu)| = | \Aut(G)  | / |\Stab_G(\mu)|$.
\end{proof}

Next, we would like to analyse the \emph{fan-in dimension} of $C(\mu)$, and the connection between $C(\mu)$ and $\sup_{\CFI}(\mu)$. The key for this is to establish a connection between the stabilisers $\Stab_E(x)$, for all $x \in \tc(\mu)$, and the kernels of the corresponding \emph{gate matrices}. For the definition of these matrices, we refer back to Section \ref{sec:XOR_fanInDimension}.
We start with the following observation that relates the stabilisers of objects in $\tc(\mu)$ with the stabilisers of their elements. 
\begin{proposition}
	\label{prop:XOR_stabiliserOfSet}
	For each $x \in \tc(\mu)$, it holds
	$
	\Stab_E(x) = \bigcap_{y \in x}  \Stab_E( [y] \cap x  ) =  \bigcap_{\gamma \in \Cc(x)} \Stab_E(\gamma).
	$
\end{proposition}
This is true simply because $x$ is stabilised by $\rho_F \in \Aut_{\CFI}(\GG)$ iff every $\Aut_{\CFI}(\GG)$-orbit within $x$ is fixed setwise by $\rho_F$.
\begin{lemma}
	\label{lem:XOR_kernelLemma}
	For every gate $[x] \in V_C$, and its gate matrix $M[x] \in \bbF_2^{[x]E_C \times E}$, it holds:
	\[
	\Ker(M[x]) = \Stab_E(x) = \Stab_E(x') \text{ for every } x' \in [x].
	\]
	For every row $M[x]_{[y]-}$, for $[y]\in[x]E_C$, it holds:
	\[
	\Ker(M[x]_{[y]-}) = \Stab_E([y] \cap x) \tag{$\star$}
	\]
\end{lemma}
\begin{proof}
	It holds $\Stab_E(x) = \Stab_E(x')$, for every $x' \in [x]$ and also $\Stab_E([y] \cap x) = \Stab_E([y] \cap x')$, for every $x' \in [x]$ (by Lemma \ref{lem:XOR_stabiliserEqual}). 
	Therefore, ($\star$) does not depend on the choice of representatives.
	From ($\star$) it immediately follows that $\Ker(M[x]) = \Stab_E(x)$, due to Proposition \ref{prop:XOR_stabiliserOfSet} and the fact that $\Ker(M[x])$ is the intersection over the kernels of the rows of $M[x]$.
	We now prove ($\star$) via induction from the input gates to the root. If $[x] = [e_0]$ is an input gate, then $M[x]$ has just one row, which is $\chi(e)^T$. The kernel of $\chi(e)^T$ is the set of all vectors in $\bbF_2^E$ which are zero at index $e$. This is precisely $\Stab_E(e_0) = \Stab_E(e_1)$, as desired.
	Now suppose $[x]$ is an internal gate, i.e.\ $x$ is a non-atomic h.f.\ set in $\tc(\mu)$. Each row of $M[x] \in \bbF_2^{[x]E_C \times E}$ is the characteristic vector of $\Xx[y] \subseteq E$, for a $[y] \in [x]E_C$. We have
	\[
	\Xx[y] = \bigtriangleup_{[w] \in [y]E_C} \Xx[w]. 
	\]
	In matrix-vector notation, we can write this as:
	\[
	M[x]_{[y]-} = \chi(\Xx[y])^T = \sum_{[w] \in [y]E_C} (M[y]_{[w]-})^T = (1 \ 1 \ ... \ 1) \cdot M[y].
	\]
	Let $\gamma \in \Cc(x)$ be the connected component such that $\gamma = [y] \cap x$. The equation above means that	$\Ker(M[x]_{[y]-}) = \mathcal{E}_y$, where $\mathcal{E}_y$ denotes the set of all vectors in $\bbF_2^E$ whose image under $M[y]$ has even Hamming weight. Thus we have to show that $\mathcal{E}_y = \Stab_E(\gamma)$.
	Each row $M[y]_{[w]-}$ corresponds to a connected component $\gamma' \in \Cc(y)$ with $w \in \gamma'$.\\  
	By the induction hypothesis, we have for each row $M[y]_{[w]-}$ and each $\mathbf{v} \in \bbF_2^E$ that $M[y]_{[w]-} \cdot \mathbf{v} = 1$ iff $\mathbf{v} \notin \Stab_E([w] \cap y)$. So $M[y] \cdot \mathbf{v}$ has even Hamming weight iff $\rho_{\chi^{-1}(\mathbf{v})} \in \Aut_{\CFI}(\GG)$ flips an even number of connected components of $y$. This is true iff $\rho_{\chi^{-1}(\mathbf{v})}$ flips an even number of components in every $y' \in \gamma$ (due to Lemma \ref{lem:XOR_CFIsymmetricWellDefined}). By definition of CFI-symmetry (Definition \ref{def:XOR_CFIsymmetric}), this is the case iff $\mathbf{v} \in \Stab_E(\gamma)$, because $\mu$ is CFI-symmetric, and thus, $\gamma$ is a CFI-symmetric component. In total, we have shown that $\mathbf{v} \in \Ee_y$ iff $\mathbf{v} \in\Stab_E(\gamma)$. This proves ($\star$) for every row of $M[x]$.
\end{proof}

As a consequence of this correspondence between kernels and stabilisers, we can bound the fan-in dimension of $C(\mu)$. This proves \textbf{Property 4} from Theorem \ref{thm:XOR_mainCircuitTheorem}.
\begin{lemma}
	\label{lem:XOR_dimensionOfCircuit}
	The \emph{fan-in dimension} of $C(\mu)$ is $\log (\maxOrb_E(\mu))$.
\end{lemma}
\begin{proof}
	Let $x \in \tc(\mu)$. From the Orbit-Stabiliser Theorem and the fact that $|\Aut_{\CFI}(\GG)| = 2^{|E|}$, it follows that 
	\[
	\Orb_E(x) = \frac{2^{|E|}}{|\Stab_E(x)|} \leq \maxOrb_E(\mu).
	\]
	This means that
	$
	\log (\maxOrb_E(\mu)) \geq |E| - \dim \ \Stab_E(x).
	$
	By Lemma \ref{lem:XOR_kernelLemma}, $\Stab_E(x) = \Ker(M[x])$.\\
	With the Rank Theorem we get:
	$
	\rk(M[x]) = |E| - \dim \ \Stab_E(x) \leq \log (\maxOrb_E(\mu)).  
	$
	Since there is an object $x \in \tc(\mu)$ where $\maxOrb_E(\mu)$ is attained,  $\rk(M[x]) =\\ \log (\maxOrb_E(\mu))$ is indeed the maximum rank of any gate matrix of $C(\mu)$.
\end{proof}

\begin{lemma}
	\label{lem:XOR_restrictedDimensionOfCircuit}
	The \emph{fan-in dimension} of $C(\mu)$ with respect to the space $\Aut_{\CFI}(\GG^S)$ is $\log (\maxOrb_{\CFI}(\mu))$. That is, for every gate $[x]$ in $C(\mu)$, we have
	\[
	\dim(M[x] \cdot \Aut_{\CFI}(\GG^S)) \leq \log (\maxOrb_{\CFI}(\mu)).
	\]
\end{lemma}	
\begin{proof}
	Let $x \in \tc(\mu)$. With the Orbit-Stabiliser Theorem we get
	\[
	\Orb_\CFI(x) = \frac{| \Aut_{\CFI}(\GG^S) |}{|\Stab_\CFI(x)|} = 2^{\dim  \Aut_{\CFI}(\GG^S) - \dim \Stab_\CFI(x)} \leq \maxOrb_\CFI(\mu).
	\]
	Thus,
	\[
	\log \maxOrb_\CFI(\mu) \geq \dim \Aut_{\CFI}(\GG^S) - \dim \Stab_\CFI(x).
	\]
	Using Lemma \ref{lem:XOR_kernelLemma}, we get $\Stab_\CFI(x) \subseteq \Stab_E(x) \subseteq \Ker(M[x])$. Therefore, 
	\[
	\dim (M[x] \cdot \Aut_{\CFI}(\GG^S)) \leq \dim \Aut_{\CFI}(\GG^S) - \dim \Stab_\CFI(x) \leq \log \maxOrb_\CFI(\mu).
	\]
\end{proof}

\begin{proof}[Proof of Theorem \ref{thm:XOR_mainCircuitTheorem}]
	First of all, since $\mu$ is by assumption CPT-definable in the structure $\GG^S$, the size $|\tc(\mu)|$ and the orbit $|\Orb_{\Aut(\GG^S)}(\mu)|$ are polynomial in $|\GG^S|$. Therefore, \textbf{Property 1} from Theorem \ref{thm:XOR_mainCircuitTheorem} clearly holds for $C(\mu)$, because $|V_C| \leq |\tc(\mu)|$. \textbf{Property 2} follows from the bound on $|\Orb_{\Aut(\GG^S)}(\mu)|$ together with Corollary \ref{cor:XOR_orbitOfCircuit}, and the fact that $| \Orb_G(\mu)  | \leq |\Orb_{\Aut(\GG^S)}(\mu)|$.\\
	\textbf{Property 4 } is proven in Lemmas \ref{lem:XOR_dimensionOfCircuit} and \ref{lem:XOR_restrictedDimensionOfCircuit}. 
	Finally, \textbf{Property 3} can be seen as follows: Suppose $C(\mu)$ is sensitive to an edge $e \in E$. This means that $e \in \Xx(r)$, for the root $r = [\mu]$ of $C(\mu)$. This is the case iff $e \in \Xx([y])$ for an odd number of children $[y] \in [\mu]E_C$. This is the same as saying that the column $M[\mu]_{-e}$ has odd Hamming weight. By equation ($\star$) from Lemma \ref{lem:XOR_kernelLemma}, this holds if and only if $\chi(e) \notin \Stab_E([y] \cap x)$ for an odd number of children $[y] \in [\mu]E_C$. Since $\mu$ is CFI-symmetric, by Definition \ref{def:XOR_CFIsymmetric} this is the case if and only if $\rho_e(\mu) \neq \mu$. And this holds iff $e \in \sup_{\CFI}(\mu)$ (because $\sup_{\CFI}(\mu)$ is the smallest possible CFI-support of $\mu$).
\end{proof}

\subsection{Proving the main theorem}
So far, we have a translation of CFI-symmetric h.f.\ sets in $\HF(\widehat{E})$ into XOR-circuits with the properties mentioned in Theorem \ref{thm:XOR_mainCircuitTheorem}. What is missing in order to conclude Theorem \ref{thm:mainInformal} from this is to prove that any CPT-algorithm which is both super-symmetric and CFI-symmetric and decides the CFI-query must construct a h.f.\ set whose properties translate into the circuit properties from Theorem \ref{thm:mainInformal}. Fortunately, a result to this effect exists already. The following support lower bound for general CPT-programs deciding the CFI-query is due to Dawar, Richerby, and Rossman \cite{dawar2008}.
\begin{theorem}[implicit in the proof of Theorem 40 in \cite{dawar2008}]
	\label{thm:DRR_supportLowerBoundCFI}
	Let $(G_n)_{n \in \bbN}$ be a family of base graphs and let $\tw_n$ denote the treewidth of $G_n$. Let $\GG_n^S, \GG_n^{S'}$ denote two non-isomorphic CFI-structures over $G_n$. Let $f(n) \leq \tw_n$ be a function such that $\GG_n^{S}$ and $\GG_n^{S'}$ are $\Cc^{\tw_n}$-homogeneous for all tuples of length $\leq 2f(n)$.\\
Then any $\CPT$-program that distinguishes $\GG_n^S$ and $\GG_n^{S'}$ for all $n \in \bbN$ must activate on input $\GG_n^S$ a h.f.\ set $x$ whose smallest support has size at least $\Omega(f(n))$.
\end{theorem}	
A structure $\GG^S_n$ is $\Cc^{\tw_n}$-\emph{homogeneous} if whenever two tuples $\bar{a}$ and $\bar{b}$ have the same $\Cc^{\tw_n}$-type in $\GG^S_n$, then there is an automorphism of $\GG^S_n$ that maps $\bar{a}$ to $\bar{b}$. The $\Cc^{\tw_n}$-type of a tuple $\bar{a}$ in $\GG^S_n$ is the collection of all $\Cc^{\tw_n}$-formulas that are true in $(\GG^S_n,\bar{a})$. By closer inspection of the entire proof in \cite{dawar2008}, one can see that homogeneity is actually only required in the weaker sense that for all tuples of some \emph{bounded length}, the $\Cc^{\tw_n}$-type partition coincides with the orbit partition -- hence the explicit restriction in the above theorem.\\ 

The homogeneity condition is satisfied by certain ordered CFI-graphs, as stated in \cite{dawar2008} and proved in \cite{ggpp}, and as we will show, the unordered CFI-graphs over hypercubes, which we use for the lower bound in Theorem \ref{thm:mainLowerBound}, satisfy it as well. Therefore, the homogeneity condition is not really a restriction in the cases that are of interest for us, which is why we omitted it in Theorem \ref{thm:mainInformal}. What we also omitted is the fact that we have to relate two different notions of support. Theorem \ref{thm:XOR_mainCircuitTheorem} refers to the minimum CFI-support of the h.f.\ sets, whereas the support lower bound above refers to the minimum $\Aut(\GG^S)$-support. Therefore, in order to formulate Theorem \ref{thm:mainInformal} correctly with all details, we have to speak about the ratio between these two supports.
\begin{definition}[CFI-support gap]
	\label{def:XOR_supportGap}
	Let $G = (V,E)$ be a base graph and $\GG^S$ a CFI-graph over it. Let $\mu \in \HF(\widehat{E})$. Denote by $s(\mu)$ the size of the smallest $\Aut(\GG^S)$-support of $\mu$ (while $\sup_{\CFI}(\mu)$ still denotes the smallest CFI-support).\\
	Then we call the ratio
	\[
	\alpha(\mu) = \frac{s(\mu)}{|\sup_{\CFI}(\mu)|}
	\]
	the \emph{CFI-support gap} of $\mu$ (with respect to $\GG^S$).
\end{definition}

Then the detailed version of Theorem \ref{thm:mainInformal} reads as follows:
\begin{restatable}{theorem}{restateCircuitLowerBound}
\label{thm:XOR_lowerboundProgram}
Let $(G_n = (V_n,E_n))_{n \in \bbN}$ be a sequence of base graphs. Let $\GG^S_n$ be a CFI-graph over $G_n$, let $\tw_n$ denote the treewidth of $G_n$. Let $f(n) \in \Oo(\tw_n)$ be a function such that every $\GG^S_n$ is $\Cc^{\tw_n}$-homogeneous, for all tuples of length $\leq 2f(n)$.
Let $g(n)$ be a function such that the \emph{CFI-support-gap} for every $\mu \in \HF(\widehat{E}_n)$ with minimum support $s(\mu) \in \Omega(f(n))$ is bounded by $g(n)$.\\
\\
If there exists a \emph{CFI-symmetric} $\CPT$-program $\Pi$ that decides the CFI-query on all $\GG_n^S$, then for every $G_n = (V_n,E_n)$, there exists an XOR-circuit $C_n$ over $G_n$ that satisfies the following ``instantiated properties'' from Theorem \ref{thm:XOR_mainCircuitTheorem}:
\begin{enumerate}
	\item The number of gates in $C_n$ is polynomial in $|\GG^S_n|$.
	\item The orbit-size $|\Orb_{G_n}(C_n)|$ of the circuit is polynomial in $|\GG^S_n|$. 
	\item $C_n$ is \emph{sensitive} to $\Omega(f(n)/g(n))$ many edges in $E_n$.
	\item The \emph{fan-in dimension} of $C_n$, \emph{restricted to the space} $\Aut_{\CFI}(\GG_n^S)$, is $\Oo(\log |\GG_n^S|)$.
	\item If the program $\Pi$ is \emph{super-symmetric} in addition to being CFI-symmetric, or if the base graph $G_n$ decomposes into at most $\Oo(\sqrt{\log |\GG^S|})$ many components when any $f(n)/g(n)$ many edges are removed, then also the (unrestricted) \emph{fan-in dimension} of $C_n$ is $\Oo(\log |\GG_n^S|)$. 
\end{enumerate}	
\end{restatable}
\begin{proof}
Assume such a CPT-program $\Pi$ exists. Let $\mu_n \in \HF(\widehat{E}_n)$ denote the CFI-symmetric h.f.\ set with large support that $\Pi$ activates on input $\GG^S_n$. Then by Theorem \ref{thm:DRR_supportLowerBoundCFI}, the smallest $\Aut(\GG^S_n)$-support of the object $\mu_n$ has size $\Omega(f(n))$. Since the CFI-support gap of $\mu_n$ in $\GG^S_n$ is at most $g(n)$, the size of the smallest CFI-support of $\mu_n$ is at least: $\sup_{\CFI}(\mu_n) \in \Omega(f(n)/g(n))$. Theorem \ref{thm:XOR_mainCircuitTheorem} applied to $\mu_n$ yields the XOR-circuit $C_n$. Property 3 from Theorem \ref{thm:XOR_mainCircuitTheorem} in combination with the bound $\sup_{\CFI}(\mu_n) \in \Omega(f(n)/g(n))$ means that $C_n$ is sensitive to $\Omega(f(n)/g(n))$ many edges in $E_n$. 
Property 4 from Theorem \ref{thm:XOR_mainCircuitTheorem} bounds the fan-in dimension and the restricted fan-in dimension in terms of $\log(\maxOrb_E(\mu_n))$ and $\log(\maxOrb_\CFI((\mu_n))$, respectively. Lemma \ref{lem:XOR_maxOrbBound} states that $\maxOrb_E(\mu_n) \leq |\Orb_E(\mu_n)| \cdot | \tc(\mu_n)|$ and $\maxOrb_\CFI(\mu_n) \leq |\Orb_\CFI(\mu_n)| \cdot | \tc(\mu_n)|$. Because $\mu_n$ is defined by the CPT-program $\Pi$ on input $\GG^S_n$, both $|\tc(\mu_n)|$ and $|\Orb_{\CFI}(\mu_n)|$ are polynomially bounded in $|\GG_n^S|$ (the orbit is bounded because $\Aut_{\CFI}(\GG^S)$ is a subgroup of $\Aut(\GG^S_n)$). This yields a polynomial bound on $\maxOrb_\CFI(\mu_n)$. Together with the $\log(\maxOrb_\CFI((\mu_n))$-bound on the restricted fan-in dimension, this gives us Property 4 from this theorem.\\
Property 5 follows then with Lemma \ref{lem:XOR_maxOrbBoundSuperSymmetric} if $\Pi$ is super-symmetric, and with Corollary \ref{cor:XOR_relationOrbitSizesEdgeFlips} in case that the base graph $G_n$ splits into a bounded number of components when $\sup_{\CFI}(\mu_n)$ is removed from it.
\end{proof}

\section{XOR-circuits for more general hereditarily finite sets}
\label{sec:generalisedCircuitConstruction}
So far, we have shown that \emph{CFI-symmetric} h.f.\ sets over CFI-structures $\GG^S$ can be quite easily transformed into XOR-circuits by factoring out the orbits under the edge-flip-group $\Aut_{\CFI}(\GG)$. Importantly, this construction automatically translates the relevant properties of the h.f.\ set, such as support size and symmetry, into more or less natural circuit-properties. As a consequence, we can -- in principle -- limit the power of \emph{CFI-symmetric} algorithms for the CFI-query by proving appropriate lower bounds for certain families of polynomial size symmetric XOR-circuits. Even though all currently known choiceless algorithms for the CFI-query \emph{are} CFI-symmetric, and it is not clear that non-CFI-symmetric algorithms are really more powerful, it would be much nicer if the circuit-translation were so general that it could be used to separate all of CPT from P, and not only the CFI-symmetric algorithms. In this subsection we explore to what extent the circuit construction can be generalised in that direction. We will present a modification of the construction above, that uses additional gadgets, works without the restriction to CFI-symmetric sets, and has almost all properties from Theorem \ref{thm:XOR_mainCircuitTheorem}. By ``almost all'' we mean that the additional gadgets we have to introduce in the circuit are of unknown size. Hence, we cannot be sure that the constructed circuit is always of polynomial size. However, we can formulate a condition on the h.f.\ sets, which generalises that of CFI-symmetry and guarantees polynomial size of the circuit. This condition concerns Boolean vector spaces with a permutation group acting on the index set. If certain subspaces of $\bbF_2^E$, which appear as stabiliser groups of the connected components of the sets in $\mu$ possess a basis that is (almost) invariant under the permutation group (which will be a subgroup of $\Aut(G)$), then the circuit constructed from $\mu$ has polynomial size. Here is the result of this section:
\begin{theorem}
	\label{thm:XOR_generalCircuitTheorem}
	Fix a family $(G_n)_{n \in \bbN}$ of base graphs. 
	For every $n \in \bbN$, let $\GG_n^S$ be a CFI-graph over $G_n = (V_n,E_n)$ and let $\mu_n \in \HF(\widehat{E}_n)$ be a h.f.\ set that is $\CPT$-definable on input $\GG^S_n$ (by the same $\CPT$-program for the whole family of graphs). Then for every $n \in \bbN$ there exists an XOR-circuit $\widehat{C}(\mu_n) = (V_C, E_C)$ over the edges of $G_n$ which satisfies:
	\begin{enumerate}
		\item The orbit-size $|\Orb_G(\widehat{C}(\mu_n))|$ of the circuit is polynomial in $|\GG^S_n|$. 
		\item $\widehat{C}(\mu_n)$ is \emph{sensitive} to at least $\frac{|\sup_{\CFI}(\mu)|}{\log(\maxOrb_E(\mu))}$ many edges in $E$.
		\item The \emph{fan-in dimension} of $\widehat{C}(\mu)$ is $\Oo(\log(\maxOrb_E(\mu)))$. 
		\item If for every $x,y \in \tc(\mu)$ such that $[y] \cap x \neq \emptyset$, the space $\Stab_E([y] \cap x)$ has a \emph{symmetric basis} (see Definition \ref{def:XOR_symmetricBasis}), then the size $|V_C|$ is polynomial in $|\GG^S|$.
	\end{enumerate}
\end{theorem}

This theorem differs from Theorem \ref{thm:XOR_mainCircuitTheorem} for CFI-symmetric objects in two aspects. Firstly, the circuit $\widehat{C}(\mu)$ is not necessarily sensitive to all edges in $\sup_{\CFI}(\mu)$ but only to a logarithmic fraction of them. Secondly, we have no guarantees for the size of the circuit unless all spaces $\Stab_E([y] \cap x)$ admit a symmetric basis; we will introduce this concept formally in Section \ref{sec:XOR_sizeBound}. It should be noted that this fourth property mentioned in the theorem is -- as far as we know -- not an ``if and only if''. It may be that $\widehat{C}(\mu)$ has polynomial size even when the symmetric basis condition is not satisfied for $\mu$.\\

As a consequence, we have the following version of Theorem \ref{thm:XOR_lowerboundProgram} for non-CFI-symmetric CPT-programs that decide the CFI-query. The condition that the h.f.\ set with large support which is used to decide the CFI-query is CFI-symmetric is weakened to the symmetric-basis condition.
As we show later, in Lemma \ref{lem:XOR_CFIsymmetricHasSymmetricBasis}, every CFI-symmetric set also has a symmetric basis, and there are as well simple examples of non-CFI-symmetric sets with a symmetric basis (see Example \ref{ex:XOR_nonCFIsymmetricSet}). Thus, the symmetric basis condition is indeed a strict generalisation of CFI-symmetry.
\begin{theorem}
	\label{thm:XOR_lowerboundProgram2}
Let $(G_n = (V_n,E_n))_{n \in \bbN}$ be a sequence of base graphs. Let $\GG^S_n$ be a CFI-graph over $G_n$, let $\tw_n$ denote the treewidth of $G_n$, and let $f(n) \leq \tw_n$ be a function such that $\GG^S_n$ is $\Cc^{\tw_n}$-homogeneous for all tuples of length $\leq 2f(n)$.\\
Let $g(n)$ be a function such that the \emph{CFI-support-gap} for every $\mu \in \HF(\widehat{E}_n)$ with minimum support $s(\mu) \in \Omega(f(n))$ is bounded by $g(n)$.\\
\\
Let $\Pi$ be a $\CPT$-program that decides the CFI-query on all $\GG_n^S$ using a h.f.\ set $\mu_n \in \GG_n^S$ with sufficient support such that $\Stab_E([y] \cap x)$ has a \emph{symmetric basis} according to Definition \ref{def:XOR_symmetricBasis}, for all $x,y \in \tc(\mu_n)$ with $[y] \cap x \neq \emptyset$.\\
Assume additionally that $\Pi$ is \emph{super-symmetric} or that the base graph $G_n$ decomposes into at most $\Oo(\sqrt{\log |\GG^S|})$ many components when any $f(n)/g(n)$ many edges are removed.\\
Then for every $n \in \bbN$ there exists an XOR-circuit $C_n$ over $G_n$ that satisfies the following ``instantiated properties'' from Theorem \ref{thm:XOR_generalCircuitTheorem}:
\begin{enumerate}
	\item The number of gates in $C_n$ is polynomial in $|\GG^S_n|$.
	\item The orbit-size $|\Orb_{G_n}(C_n)|$ of the circuit is polynomial in $|\GG^S_n|$. 
	\item $C_n$ is \emph{sensitive} to $\Omega(f(n)/(g(n) \cdot \log | \GG_n^S|))$ many edges in $E_n$.
	\item The \emph{fan-in dimension} of $C_n$ is $\Oo(\log |\GG_n^S|)$.
\end{enumerate}	
\end{theorem}	
We omit the proof of this theorem because it follows from Theorem \ref{thm:XOR_generalCircuitTheorem} in the same way as Theorem \ref{thm:XOR_lowerboundProgram} follows from Theorem \ref{thm:XOR_mainCircuitTheorem}.\\

Now let us start with the proof of Theorem \ref{thm:XOR_generalCircuitTheorem}, which spans the rest of the section. It should be noted that Theorem \ref{thm:XOR_mainCircuitTheorem} is actually a special case of this, so we could have omitted the circuit construction for CFI-symmetric objects; however, the more general construction that we present now is not as natural as the one for CFI-symmetric objects and much harder to describe.\\
Fix again a base graph $G = (V,E)$, a CFI-graph $\GG^S$ over it, and an object $\mu \in \HF(\widehat{E})$. This time, $\mu$ need not be CFI-symmetric. In the previous subsection, we wrote $C(\mu)$ for the circuit obtained by factoring out the $\sim_E$ classes in $\tc(\mu)$. Now, we denote the constructed circuit by $\widehat{C}(\mu)$. Before we explain the construction, let us look at why $C(\mu)$ is not ``the circuit we want'' if $\mu$ is not CFI-symmetric. The only place where CFI-symmetry was required in the previous subsection is in the proof of Lemma \ref{lem:XOR_kernelLemma}, which relates the kernels of the gate matrices with the vector spaces $\Stab_E(x)$. This relationship is crucial because it leads to the connection between $\sup_{\CFI}(\mu)$ and the sensitivity of $C(\mu)$ to its input bits, and is also necessary to get a bound on the fan-in dimension of $C(\mu)$. Without such a bound, the construction would not be interesting because without fan-in restrictions, there always exist small symmetric XOR-circuits. Hence, we would like to ensure that the statement of Lemma \ref{lem:XOR_kernelLemma} still holds for $\widehat{C}(\mu)$, even if $\mu$ is not CFI-symmetric. Now take a look at the inductive proof of Lemma \ref{lem:XOR_kernelLemma} again. The key in this induction is that for any gate matrix $M[x]$ and any child $[y]$ of $[x]$, the row $M[x]_{[y]-}$ can be written as the product of another matrix and the child-gate-matrix $M[y]$: $M[x]_{[y]-} = (1 \ 1 \ ... \ 1) \cdot M[y]$. This equation holds because of the CFI-symmetry of $\mu$. Now in the general case, a similar equation will hold, namely: $M[x]_{[y]-} = N \cdot M[y]$, for some matrix $N$ that has to be chosen depending on $\Stab_E([y] \cap x)$. So our plan for this section is as follows: We will first of all define these $N$-matrices, that essentially ``repair'' the proof of Lemma \ref{lem:XOR_kernelLemma} in the non-CFI-symmetric case. Based on these matrices and on $C(\mu)$, we will construct $\widehat{C}(\mu)$ by introducing gadgets that simulate the effect of the chosen $N$-matrices. Then this circuit will be exactly such that the proof of Lemma \ref{lem:XOR_kernelLemma} goes through again, even if $\mu$ is not CFI-symmetric. In a sense, we can view the gadgets as corrections for local violations of CFI-symmetry. 

\subsection{Definition of the matrices}

Actually, we will not only define the said $N$-matrices, but also, for every $[x] \in \{ [x] \mid x \in \tc(\mu)\}$, a Boolean matrix $M[x]$. We keep the notation $M[x]$ from the previous subsection, even though, strictly speaking, $M[x]$ will not be the gate matrix of any gate in $\widehat{C}(\mu)$; rather, it will be a matrix that satisfies $\Ker(M[x]) = \Stab_E(x)$, and it will serve as a kind of construction specification for a gadget in $\widehat{C}(\mu)$. The $M[x]$-matrices that we are going to define depend on the $M$- and $N$-matrices of the connected components of $x$ in the object $\mu$. Therefore, the construction of these matrices proceeds inductively from the atoms of $\mu$ to the more deeply nested sets. At this point, our main objective is to build the matrices in such a way that their kernels correspond to the stabiliser spaces of the sets they belong to. Furthermore, the matrices should satisfy certain symmetry requirements with respect to the action of $\Stab_G(\mu) \leq \Aut(G)$. Only after the construction of the matrices, we will use them to construct the circuit $\widehat{C}(\mu)$ in such a way that the gate matrices of $\widehat{C}(\mu)$ are as desired. So the procedure in this subsection is the other way round as in the previous one, where the circuit came first, and we then analysed its gate matrices. Now we are specifying the matrices first, and then build the circuit so that the statement of Lemma \ref{lem:XOR_kernelLemma} holds for $\widehat{C}(\mu)$ by construction.
Along with the matrices, we will provide certain group homomorphisms that ensure symmetry.
To speak about symmetry of matrices, we introduce notation to express the effect of row- and column-permutations:\\

Let $M \in \bbF_2^{I \times J}$ be any Boolean matrix, and $\sigma : I \lra I', \pi : J \lra J'$ be bijections. Then $(\sigma,\pi)(M) \in \bbF_2^{I' \times J'}$ is the matrix with $(\sigma,\pi)(M)_{\sigma i, \pi j} = M_{i,j}$ for each $(i,j) \in I \times J$. In particular, if $\sigma \in \Sym(I), \pi \in \Sym(J)$, then $(\sigma,\pi)(M)$ is the matrix that arises from $M$ when the respective row- and column-permutations are applied.
\begin{proposition}
	\label{prop:XOR_matrixMultiplicationAndPermutation}
	Let $\mathbf{v}, \mathbf{w} \in \bbF_2^I$ be two Boolean vectors and $\pi \in \Sym(I)$ a permutation of its entries. Then
	\[
	\mathbf{v}^T \cdot \mathbf{w} = \pi(\mathbf{v}^T) \cdot \pi(\mathbf{w}). 
	\]
	Therefore, if $M \in \bbF_2^{I \times J}$ is a matrix and $\mathbf{w} \in \bbF_2^J$ is a vector, then
	\[
	\sigma(M \cdot \mathbf{w}) = (\sigma, \pi)M \cdot \pi(\mathbf{w}),
	\]
	for every $\sigma \in \Sym(I), \pi \in \Sym(J)$.
\end{proposition}	
This proposition follows immediately from the definition of the scalar product of vectors because $\mathbf{v}^T \cdot \mathbf{w}$ and $\pi(\mathbf{v}^T) \cdot \pi(\mathbf{w})$ are the same sum of products, just summed in a different order. The statement about matrix-vector-multiplication then follows because this is just the scalar product of every row vector with $\mathbf{w}$. The proposition will sometimes be used without explicit reference in this section.\\

In the rest of this section, we will often speak about the following orbits and stabilisers. They differ from the ones from the previous section in so far as they concern the subgroup $\Stab_G(\mu) \leq \Aut(G)$, instead of $\Aut(G)$ itself. Thus, we override the notation from the previous section. Let $x \in \tc(\mu)$.
\begin{align*}
	\Orb_G([x]) &:= \{ \pi([x]) \mid \pi \in \Stab_G(\mu) \leq \Aut(G) \}.\\
	\Stab_G([x]) &:= \{ \pi \in \Stab_G(\mu) \mid \pi([x]) = [x]   \}.
\end{align*}
Similarly, for $x,y \in \tc(\mu)$, $\Stab_G([y] \cap x)$ refers to the stabiliser of the set $[y] \cap x$ in the group $\Stab_G(\mu)$.\\

Any orbit of a $\sim$-class is a set of $\sim$-equivalence classes:
\begin{lemma}
	\label{lem:XOR_orbitsOfSimClasses}
	For any $\sim_E$-class $[x] \subseteq \tc(\mu)$, and any $\pi \in \Stab_G(\mu)$, $\pi([x])$ is also a $\sim_E$-class of $\tc(\mu)$.
\end{lemma}	
\begin{proof}
	It is not hard to check that for any $\pi \in \Aut(G)$, and any $\rho_F \in \Aut_{\CFI}(\GG)$ it holds $\rho_{\pi F} = \pi \circ \rho_F \circ \pi^{-1}$.
	Let $x',x'' \in [x]$ and let $\rho_F \in \Aut_{\CFI}(\GG)$ be such that $\rho_F(x') = x''$.  Then $\rho_{\pi F}(\pi x') = \pi x''$ by the above equation. Thus, $\pi(x') \sim_E \pi(x'')$. Similarly one can show that if $x' \not\sim_E x''$, then $\pi(x') \not\sim_E \pi(x'')$. Therefore, $\pi([x])$ is again a $\sim_E$-class. The fact that $\pi([x]) \subseteq \tc(\mu)$ follows because $\pi$ extends to an automorphism of the h.f.\ set $\mu$. 
\end{proof}

\begin{corollary}
	\label{cor:XOR_stabYXinStabY}
	Let $x,y \in \tc(\mu)$ such that $[y] \cap x \neq \emptyset$. Then
	\[
	\Stab_G([y] \cap x) \leq \Stab_G([y]).
	\]
\end{corollary}
\begin{proof}
	The group $\Stab_G([y] \cap x)$ maps the set $[y] \cap x \subseteq [y]$ to itself, so it does not move this subset of the $\sim$-class $[y]$ into another $\sim$-class. Then by Lemma \ref{lem:XOR_orbitsOfSimClasses}, it must map the whole class $[y]$ to itself because the image of $[y]$ must again be a $\sim$-class.
\end{proof}

\begin{lemma}
	\label{lem:XOR_stabiliserEqualForComponents}
	For any $x,x',y \in \tc(\mu)$ such that $x \sim_E x'$ and $[y] \cap x \neq \emptyset$, it holds
	\[
	\Stab_E([y] \cap x) = \Stab_E([y] \cap x').
	\]	
\end{lemma}	
\begin{proof}
	The sets $[y] \cap x$ and $[y] \cap x'$ are related via an automorphism in $\Aut_{\CFI}(\GG)$. Therefore,
	Lemma \ref{lem:XOR_stabiliserEqual} applied to the set $[y] \cap x$ yields the desired statement.
\end{proof}	
Now we come to the inductive definition of the aforementioned $M$- and $N$-matrices. Here is the precise list of objects that we are going to define:\\
\\
\begin{enumerate}[(a)]
	\item For every $[x] \in \{[x] \mid x \in \tc(\mu) \}$:
	\begin{enumerate}[(i)]
		\item An index-set $I_{[x]}$.
		\item A matrix $M[x] \in \bbF_2^{I_{[x]} \times E}$ with the property that $\Ker(M[x]) = \Stab_E(x)$ (note that $\Stab_E(x) = \Stab_E(x')$ for every $x' \in [x]$ by Lemma \ref{lem:XOR_stabiliserEqual}).
	\end{enumerate}	
	Let $\Cc[x] := \{ [y] \mid [y] \cap x \neq \emptyset \}$ (note that this does not depend on the representative of $[x]$ -- see Lemma \ref{lem:XOR_inRelation}) denote the $\sim$-classes of the connected components of $x$.
	\item For every $[x] \in \{[x] \mid x \in \tc(\mu) \}$, where $x$ is a set,
	and every $[y] \in \Cc[x]$:
	\begin{enumerate}[(i)]
		\item An index-set $J_{[x][y]}$.
		\item A matrix $N[x][y] \in \bbF_2^{J_{[x][y]} \times I_{[y]}}$ with the property that $\Ker(N[x][y] \cdot M[y]) = \Stab_E([y] \cap x)$ (by Lemma \ref{lem:XOR_stabiliserEqualForComponents}, $\Stab_E([y] \cap x)$ is independent of the choice of representative of $[x]$).
	\end{enumerate}	
	\item For every orbit $\Omega_{[x]} := \Orb_G([x])$, let $I_{\Omega_{[x]}} := \biguplus_{[x'] \in \Omega_{[x]}} I_{[x']}$. For every $\Omega_{[x]}$, we provide a group homomorphism $g_{[x]} : \Stab_G(\mu) \lra \Sym(I_{\Omega_{[x]}})$ such that for each $[x'] \in \Omega_{[x]}$ and each $\pi \in \Stab_G(\mu)$, it holds $g_{[x]}(\pi)(I_{[x']}) = I_{\pi[x']}$. $M\pi[x'] = (g_{[x]},\pi)M[x']$ for each $[x'] \in \Omega_{[x]}$ and $\pi \in \Stab_G(\mu)$.
	\item For every orbit $\Omega_{[x]}$, let 
	\[
	J_{\Omega_{[x]}} := \biguplus_{\stackrel{[x'] \in \Omega_{[x]}}{[y'] \in \Cc[x']}} J_{[x'][y']}.
	\]
	For every $\Omega_{[x]}$, we provide a group homomorphism $h_{[x]} : \Stab_G(\mu) \lra \Sym(J_{\Omega_{[x]}})$ such that for each $[x'] \in \Omega_{[x]}$ and each $[y'] \in \Cc[x']$, it holds $h_{[x]}(\pi)(J_{[x'][y']}) = J_{\pi[x']\pi[y']}$. Furthermore, $N\pi[x']\pi[y']= (h_{[x]}(\pi),g_{[x]}(\pi))N[x'][y']$ for each $[x'] \in \Omega_{[x]}$, $[y'] \in \Cc[x']$, and $\pi \in \Stab_G(\mu)$.
\end{enumerate}	

The role of the group homomorphisms is to ensure -- when we build the circuit from these matrices -- that every $\pi \in \Stab_G(\mu)$ that acts on the input gates indeed extends to an automorphism of $\widehat{C}(\mu)$. Before we actually construct anything, we have to verify that it is indeed possible to satisfy the symmetry conditions witnessed by the group homomorphisms and the conditions on the kernels of the matrices simultaneously. In other words, we have to show that the stabiliser spaces, which are supposed to be equal to the respective kernels, are mapped to each other by the permutations in $\Stab_G(\mu)$:
\begin{lemma}
	\label{lem:XOR_actionOfAutomorphismsOnStabiliserSpaces}
	Let $x \in \tc(\mu)$ and $\pi \in \Stab_G(\mu)$. Then 
	\[
	\Stab_E(\pi x) = \pi(\Stab_E(x)) = \{  \pi \mathbf{v} \mid \mathbf{v} \in \Stab_E(x)  \},
	\]
	where $\pi \in \Sym(E)$ acts on vectors in $\bbF_2^E$ by permuting their entries. Furthermore, for every $[y] \in \Cc[x]$, we have
	\[
	\Stab_E(\pi[y] \cap \pi x) = \pi(\Stab_E([y] \cap x)).
	\] 
\end{lemma}	 
\begin{proof}
	For any $\pi \in \Aut(G)$, and any $\rho_F \in \Aut_{\CFI}(\GG)$, it holds $\rho_{\pi F} = \pi \circ \rho_F \circ \pi^{-1}$. Thus, $\rho_{\pi F} \in \Stab_E(\pi x)$ if and only if $(\pi \circ \rho_F \circ \pi^{-1})(\pi x) = \pi x$. This holds if and only if $(\pi \circ \rho_F)(x) = \pi x$, which is the case iff $\rho_F(x) = x$. This proves the first part of the lemma since $\chi(\pi F) = \pi(\chi(F))$. The second part can be shown in the same way because $\pi[y] \cap \pi x = \pi([y] \cap x)$. This last equation holds since the action of $\pi$ on $\HF(\widehat{E})$ is a bijection from $\HF(\widehat{E})$ to itself, so $\pi[y] \cap \pi x \subseteq \pi([y] \cap x)$ (this would not necessarily be true if $\pi$ were not injective on $\HF(\widehat{E})$). 
\end{proof}	

\subsection*{Inductive construction}
\textbf{Base case:}\\
Let $x = e_i$, for $e \in E$ and $i \in \{0,1\}$, be an atom in $\tc(\mu)$. Then we set
\[
M[x] := \chi(e)^T.
\]
Formally, we define the row index set as $I_{[x]} := \{[x]\}$, but any singleton set that is distinct from all other index sets will do.\\
Now for every orbit $\Omega_{[x]}$, where $x$ is an atom in $\tc(\mu)$, we define the homomorphism $g_{[x]} : \Stab_G(\mu) \lra \Sym(I_{\Omega_{[x]}})$ by letting $g_{[x]}(\pi)([x']) := [\pi x']$ for every $\pi \in \Stab_G(\mu)$, $[x'] \in I_{\Omega_{[x]}}$ (note that by definition of the index-sets $I_{[x]}$, $I_{\Omega_{[x]}}$ is equal to the orbit $\Omega_{[x]}$).\\
\\
\textbf{Inductive step:}\\
We deal with the items from the above list in the order (b), (d), (a), (c).
Let $x \in \tc(\mu)$ be a non-atomic object, that is, a set. Assume that for every $[y] \in \Cc[x]$ and every $[y] \in \Cc[x']$, for every $[x'] \in \Omega_{[x]}$, the respective matrix $M[y]$ with index set $I_{[y]}$ has been constructed. Thus we also assume that for any such $[y]$, the homomorphism $g_{[y]}$ corresponding to $\Omega_{[y]}$ has been defined. We fix a $y$ such that $[y] \in \Cc[x]$. For this fixed pair $[x], [y]$ we will now construct the matrix $N[x][y]$. Then we will close it under the action of $\Stab_G(\mu)$. That is, given this matrix $N[x][y]$, we will symmetrically define $N[x'][y']$ for all $[y'] \in \Omega_{[y]}$, and all $[x'] \in \Omega_{[x]}$ such that $[y'] \in \Cc[x']$.\\
After this, there may still exist some components $[y'] \in \Cc[x]$ for which $N[x][y']$ has not been defined. In that case, we fix such a $[y'] \in \Cc[x]$, define the corresponding matrix $N[x][y']$ explicitly, and define the matrices for all $\Stab_G(\mu)$-images of $[x]$ and $[y']$ symmetrically, and so on. Hence, we first have to describe how to define the respective initial matrix from which we obtain the other ones by symmetry.\\

\paragraph*{Definition of the N-matrices}
Let $y \in \tc(\mu)$ be such that $[y] \in \Cc[x]$. 
We assume that $M[y]$, $I_{[y]}$ and $g_{[y]}$ have been constructed. 
The matrix $N[x][y]$ is defined as the smallest Boolean matrix that satisfies the following two conditions:
\begin{enumerate}[(i)]
	\item
	$
	\Ker(N[x][y]) \cap \Im(M[y]) = M[y] \cdot \Stab_E([y] \cap x) = \{ (M[y] \cdot \mathbf{v}) \mid \mathbf{v} \in \Stab_E([y] \cap x)  \}.    
	$\\
	(For a matrix $M \in \bbF_2^{I \times J}$, $\Im(M)$ denotes the space that is the image of $\bbF_2^J$ under $M$).
	\item There exists a homomorphism $h$ from $g_{[y]}(\Stab_G([y] \cap x)) \leq \Sym(I_{[y]})$ into the symmetric group on the row index set of $N[x][y]$ such that for every $\sigma \in g_{[y]}(\Stab_G([y] \cap x))$, it holds $(h(\sigma), \sigma)(N[x][y]) = N[x][y]$.
\end{enumerate}	
In the second property, we abused notation and wrote $g_{[y]}(\Stab_G([y] \cap x))$ for a subgroup of $\Sym(I_{[y]})$, even though $g_{[y]}(\Stab_G([y] \cap x))$ is formally a subgroup of $\Sym(I_{\Omega_{[y]}})$. However, we know from Corollary \ref{cor:XOR_stabYXinStabY} that $\Stab_G([y] \cap x) \leq \Stab_G([y])$, so $g_{[y]}(\Stab_G([y] \cap x))$ indeed maps the row index set of $M[y]$, that is, $I_{[y]}$, to itself (see property (c) of $g_{[y]}$ that holds by the induction hypothesis).\\

By ``smallest'' matrix we mean one that satisfies (i) and (ii) and has the least number of rows. If there are multiple such matrices with the same minimal number of rows, we choose an arbitrary one of them for $N[x][y]$.\\

Let $m$ be the number of rows of $N[x][y]$. We define the row index set $J_{[x][y]}$ of $N[x][y]$ as an $m$-element set that is disjoint from all other index sets constructed so far. Formally, this can be achieved by letting $J_{[x][y]} := \{ (i,[x],[y]) \mid i \in [m]  \}$.\\
We have to show that there always exists a matrix that satisfies (a) and (b). A matrix satisfying (a) can be found with methods from linear algebra:
\begin{lemma}
	\label{lem:XOR_kernelMatrixExists}
	Let $\Gamma \leq \Delta \leq \bbF_2^I$ be Boolean vector spaces. Let $d$ be the dimension of $\Gamma$ and $k = \dim \Delta - d$ be the codimension of $\Gamma$ in $\Delta$. There exists a matrix $N \in \bbF_2^{[k] \times I}$ such that $\Ker(N) \cap \Delta = \Gamma$.
\end{lemma}	
\begin{proof}
	Each of the $k$ rows of $N$ can be obtained as the solution to a linear equation system. For $i \in [k]$, $j \in I$, let $n_{ij} := N_{ij}$ denote the sought entry in row $i$ and column $j$. Fix a basis $\Bb_{\Gamma}$ of $\Gamma$, and an extension of that basis $\Bb \supseteq \Bb_{\Gamma}$ such that $\Bb = \Bb_{\Gamma} \uplus \{ \mathbf{w}_1,...,\mathbf{w}_k \}$ is a basis of $\Delta$. For each $i \in [k]$, we define an equation system $A_i \cdot \mathbf{x} = \mathbf{b}_i$ whose unique solution vector is the desired row $(n_{i1}, n_{i2},...,n_{i|I|})$ of $N$. The system has $\dim \Delta$ many equations, where each equation is associated with a basis vector in $\Bb$. For every basis vector $\mathbf{v} \in \Bb \setminus \{ \mathbf{w}_i\}$, we have the equation 
	\[
	\sum_{j \in I}\mathbf{v}(j)\cdot \mathbf{x}(j) = 0
	\]
	in the system $A_i \cdot \mathbf{x} = \mathbf{b}_i$. For the basis vector $\mathbf{w}_i \in \Bb$, we have the equation 
	\[
	\sum_{j \in I}\mathbf{w}_i(j)\cdot \mathbf{x}(j) = 1
	\]
	in $A_i \cdot \mathbf{x} = \mathbf{b}_i$. In this way, we define $k$ equation systems, one for each $i \in [k]$. In fact, the coefficient matrix $A_i$ is the same for all of them. Its rows are the vectors in $\Bb$ (transposed). The vector $\mathbf{b}_i$ has a $1$-entry in the row containing $\mathbf{w}_i$, and is zero otherwise. The rank and the number of rows of every $A_i$ is $\dim \Delta$ because $\Bb$ is a basis of $\Delta$. Hence, each of the equation systems has a unique solution. If we define each entry $n_{ij}$ of $N$ to be the $j$-th entry of the solution vector to $A_i \cdot \mathbf{x} = \mathbf{b}_i$, then indeed, $\Ker(N) \cap \Delta = \Gamma$, by definition of the equation systems.
\end{proof}

This shows that a matrix satisfying condition (a) always exists. The matrix can be closed under the action of $g_{[y]}(\Stab_G([y] \cap x))$ so that it also satisfies condition (b). This requires that the vector space that we want as the kernel of $N[x][y]$ is invariant under that permutation group:    
\begin{lemma}
	\label{lem:XOR_imageOfStabiliserPermutationInvariant}
	The space $M[y] \cdot \Stab_E([y] \cap x) \leq \bbF_2^{I_{[y]}}$ is invariant under the action of the permutation group $g_{[y]}(\Stab_G([y] \cap x)) \leq \Sym(I_{[y]})$ on the entries of its vectors. That is, for every $\mathbf{v} \in M[y] \cdot \Stab_E([y] \cap x)$ and $\pi \in g_{[y]}(\Stab_G([y] \cap x))$, it holds $\pi(\mathbf{v}) \in M[y] \cdot \Stab_E([y] \cap x)$. 
\end{lemma}	
\begin{proof}
	Let $\mathbf{v} \in M[y] \cdot \Stab_E([y] \cap x)$ and $\pi \in g_{[y]}(\Stab_G([y] \cap x))$. We can write $\mathbf{v} = M[y] \cdot \mathbf{w}$ for some $\mathbf{w} \in \Stab_E([y] \cap x)$. Fix a $\sigma \in g^{-1}_{[y]}(\pi)$, i.e.\ $\sigma \in \Stab_G([y] \cap x)$. By Proposition \ref{prop:XOR_matrixMultiplicationAndPermutation} it holds:
	\[
	(\pi, \sigma)(M[y]) \cdot \sigma(\mathbf{w}) = \pi(\mathbf{v}).
	\]
	From the inductive hypothesis we have that $(\pi, \sigma)(M[y]) = M[y]$ since $\sigma \in \Stab_G([y] \cap x) \leq \Stab_G([y])$ (see item (c) in the enumeration above). The fact that $\Stab_G([y] \cap x) \leq \Stab_G([y])$ is shown in Corollary \ref{cor:XOR_stabYXinStabY}. We conclude: $M[y] \cdot \sigma(\mathbf{w}) = \pi(\mathbf{v})$. If $\sigma(\mathbf{w}) \in \Stab_E([y] \cap x)$, then we are done and have that $\pi(\mathbf{v}) \in M[y] \cdot \Stab_E([y] \cap x)$, as desired. To show that $\sigma(\mathbf{w}) \in \Stab_E([y] \cap x)$, we apply Lemma \ref{lem:XOR_actionOfAutomorphismsOnStabiliserSpaces}: Since $\mathbf{w} \in \Stab_E([y] \cap x)$, we have $\sigma \mathbf{w} \in \sigma (\Stab_E([y] \cap x) )= \Stab_E(\sigma[y] \cap \sigma x)$. Finally, as mentioned in the proof of Lemma \ref{lem:XOR_actionOfAutomorphismsOnStabiliserSpaces}, we have $\sigma[y] \cap \sigma x = \sigma([y] \cap x)$, and it holds $\sigma([y] \cap x) = [y] \cap x$, because $\sigma \in \Stab_G([y] \cap x)$. 
\end{proof}	

Knowing this, we can see that it is indeed possible to satisfy both conditions (i) and (ii) at the same time.
\begin{lemma}
	\label{lem:XOR_existenceOfNMatrixWithConditionsAB}
	Let $[y] \in \Cc[x]$. There exists a Boolean matrix $N$ that satisfies conditions (i) and (ii) mentioned above, i.e.:
	\begin{enumerate}[(i)]
		\item
		$
		\Ker(N) \cap \Im(M[y]) = M[y] \cdot \Stab_E([y] \cap x).
		$
		\item There exists a homomorphism $h$ from $g_{[y]}(\Stab_G([y] \cap x)) \leq \Sym(I_{[y]})$ into the symmetric group on the row index set of $N$ such that for every $\sigma \in g_{[y]}(\Stab_G([y] \cap x))$, it holds $(h(\sigma), \sigma)(N) = N$.
	\end{enumerate}
\end{lemma}	
\begin{proof}
	Lemma \ref{lem:XOR_kernelMatrixExists} applied to $\Delta = \Im(M[y])$ and $\Gamma =  M[y] \cdot \Stab_E([y] \cap x)$ gives us a matrix $N' \in \bbF_2^{[k] \times I_{[y]}}$ that satisfies condition (i); here, $k = \dim (\Im \ M[y]) - \dim( M[y] \cdot \Stab_E([y] \cap x) )$. We can close $N'$ under the action of $g_{[y]}(\Stab_G([y] \cap x)) \leq \Sym(I_{[y]})$ so that condition (ii) is also satisfied: For each row $N'_{i-}$ of $N'$, let $\Orb(N'_{i-}) := \{ \pi(N'_{i-}) \mid \pi \in g_{[y]}(\Stab_G([y] \cap x)) \}$. Here, $\pi \in \Sym(I_{[y]}) $ permutes the columns, i.e.\ the entries of the respective row $N'_{i-}$. Now let $N$ be the Boolean matrix whose set of rows is the disjoint union $\biguplus_{i \in [k]} \Orb(N'_{i-})$. Clearly, there is a homomorphism $h$ from $g_{[y]}(\Stab_G([y] \cap x))$ into the symmetric group on the rows of $N$ (more precisely into $\Pi_{i \in [k]} \Sym(\Orb(N'_{i-}))$). This homomorphism is just the group action of $g_{[y]}(\Stab_G([y] \cap x))$ on the rows of $N$ (separately on the orbits $\Orb(N'_{i-})$).\\
	It remains to show that this symmetry-closed matrix $N$ still satisfies condition (i), i.e.\ that $\Ker(N) \cap \Im(M[y]) = M[y] \cdot \Stab_E([y] \cap x)$. We have $\Ker(N) \cap \Im(M[y]) \subseteq M[y] \cdot \Stab_E([y] \cap x)$ because for any vector $\mathbf{v} \notin M[y] \cdot \Stab_E([y] \cap x)$, either $\mathbf{v} \notin \Im(M[y])$, or if $\mathbf{v} \in \Im(M[y])$, then $N' \cdot \mathbf{v} \neq \mathbf{0}$ because $\Ker(N') \cap \Im(M[y]) = M[y] \cdot \Stab_E([y] \cap x)$. Since $N'$ is a submatrix of $N$, we also have $N \cdot \mathbf{v} \neq \mathbf{0}$, so $\mathbf{v} \notin \Ker(N)$. Therefore, $\Ker(N) \cap \Im(M[y]) \subseteq M[y] \cdot \Stab_E([y] \cap x)$. It remains to show: $M[y] \cdot \Stab_E([y] \cap x) \subseteq \Ker(N) \cap \Im(M[y])$. So let $\mathbf{v} \in M[y] \cdot \Stab_E([y] \cap x)$. Then $N' \cdot \mathbf{v} = \mathbf{0}$. We have to prove that for every row $N'_{i-}$ and every $\pi \in g_{[y]}(\Stab_G([y] \cap x))$, we have $\pi(N'_{i-}) \cdot \mathbf{v} = 0$. It holds (see Proposition \ref{prop:XOR_matrixMultiplicationAndPermutation}):
	\[
	\pi(N'_{i-}) \cdot \mathbf{v} = \pi^{-1}(\pi(N'_{i-})) \cdot \pi^{-1}\mathbf{v} = N'_{i-} \cdot \pi^{-1}\mathbf{v} = 0.
	\]
	The final equality holds because $\Ker(N') = M[y] \cdot \Stab_E([y] \cap x)$, and $\pi^{-1}\mathbf{v} \in M[y] \cdot \Stab_E([y] \cap x)$ by Lemma \ref{lem:XOR_imageOfStabiliserPermutationInvariant}. Since each row of $N$ is of the form $\pi(N'_{i-})$ for some row $i$ of $N'$ and $\pi \in g_{[y]}(\Stab_G([y] \cap x))$, we have shown that $M[y] \cdot \Stab_E([y] \cap x) \subseteq \Ker(N) \cap \Im(M[y])$.
\end{proof}	
This lemma shows that there exists a Boolean matrix satisfying conditions (i) and (ii), so it is indeed possible to pick a smallest one for $N[x][y]$. The trouble is that we do not know a priori how small it is. Therefore, the construction of $\widehat{C}(\mu)$ that we are describing does not come with a guaranteed size bound. Later on in Section \ref{sec:XOR_sizeBound} we will get back to the choice of $N[x][y]$ and show how we can bound its size in case that $\Stab_E([y] \cap x)$ has a symmetric basis. The idea will be that the size of the closure of $N'$ under the action of $g_{[y]}(\Stab_G([y] \cap x))$ can be bounded then, because the symmetries of the basis of $\Stab_E([y] \cap x)$ ``propagate'' through the equation systems $A_i \cdot \mathbf{x} = \mathbf{b}_i$ that are used to define the rows of $N'$.\\
\\
Since at least one matrix satisfying (i) and (ii) exists, $N[x][y]$ can indeed be defined as the smallest one. Remember that this definition was for a fixed $[y] \in \Cc[x]$. Now let $[y'] \in \Omega_{[y]}$, and $[x'] \in \Omega_{[x]}$ such that $[y'] \in \Cc[x']$ and such that there is a $\pi \in \Stab_G(\mu)$ with $\pi([y] \cap x) = [y'] \cap x'$ (it may be that $[x'] = [x]$). Set $\pi_{[x'][y']} := \pi$, so we can refer back to this particular permutation in the future.
We set $J_{[x'][y']} :=  \{ (i,[x'],[y']) \mid i \in [m]  \}$. Here, $m$ still denotes the number of rows of the previously defined $N[x][y]$. 
Let $N[x'][y']$ be the Boolean matrix in $\bbF_2^{J_{[x'][y']} \times I_{[y']}}$ such that $(N[x'][y'])_{(k,[x'],[y']),g_{[y]}(\pi_{[x'][y']})(i)} = (N[x][y])_{(k,[x],[y]),i}$ for every $i \in I_{[y]}$ and $k \in [m]$. We will usually write $N([x'][y'])_{k,-}$ instead of $(N[x'][y'])_{(k,[x'],[y']),-}$ to denote the $k$-th row of the matrix.
In this way, we define the matrices $N[x'][y']$ for all $[y'] \in \Omega_{[y]}$. We will call the $\sim$-class $[y] \in \Cc[x]$, that we arbitrarily chose as the first one in its orbit to define $N[x][y]$ with Lemma \ref{lem:XOR_existenceOfNMatrixWithConditionsAB}, the \emph{primer} of the orbit $\Omega_{[y]}$. We proceed to pick a new primer $[y] \in \Cc[x]$ for which $N[x][y]$ has not been defined so far, and repeat the construction for $[y]$ and its orbit $\Omega_{[y]}$. This is done until $N[x][y]$ is defined for every $[y] \in \Cc[x]$ and for every $[y'] \in \Cc[x']$, for every $[x'] \in \Omega_{[x]}$.\\

We show that the defined matrices have the desired properties: 
\begin{lemma}
	\label{lem:XOR_constructedKernelCorrect}
	Let $[x'] \in \Omega_{[x]}$ and $[y'] \in \Cc[x']$. Then 
	\[
	\Ker(N[x'][y'] \cdot M[y']) = \Stab_E([y'] \cap x').
	\]
\end{lemma}	
\begin{proof}
	Let $[y] \in \Cc[x]$ be the primer of the orbit $\Omega_{[y']}$ that was used in the matrix construction. Then $[y'] = \pi_{[x'][y']} [y]$ and $[x'] = \pi_{[x'][y']} [x]$. If $[y'] = [y]$, then $\pi_{[x'][y']}$ is the identity permutation in $\Sym(V)$. For ease of notation, we write $\pi :=  \pi_{[x'][y']}$ in the following. By Lemma \ref{lem:XOR_actionOfAutomorphismsOnStabiliserSpaces}, we have
	\[
	\Stab_E([y'] \cap x') =  \Stab_E(\pi[y] \cap \pi x) = \pi(\Stab_E([y] \cap x)).
	\]
	By definition of $N[x][y]$, and because $\Ker(M[y]) = \Stab_E(y) \leq \Stab_E([y] \cap x)$, we have $\Ker(N[x][y] \cdot M[y]) = \Stab_E([y] \cap x)$. It holds $M[y'] = (g_{[y]}(\pi),\pi)M[y]$ (item (c) of the inductive hypothesis). So for any vector $\mathbf{v} \in \bbF_2^E$, we have $M[y'] \cdot \pi(\mathbf{v}) = g_{[y]}(\pi)( M[y] \cdot \mathbf{v})$. Here, $g_{[y]}(\pi)$ acts on a vector $\mathbf{w} \in \bbF_2^{I_{[y]}}$ by mapping it to a vector $\mathbf{w}' \in \bbF_2^{I_{[y']}}$ with $\mathbf{w}'(g_{[y]}(\pi)(i)) = \mathbf{w}(i)$ for every $i \in I_{[y]}$. By definition, we have for the $k$-th row of $N[x'][y']$: $N([x'][y'])_{k,-} = g_{[y]}(\pi)(N([x][y])_{k,-} )$. In total, this means that for every $\mathbf{v} \in \bbF_2^E$, it holds:
	\[
	N([x'][y'])_{k,-} \cdot M[y'] \cdot \pi(\mathbf{v}) = N([x][y])_{k,-} \cdot M[y] \cdot \mathbf{v}. 
	\]
	Therefore, $\Ker(N[x'][y'] \cdot M[y']) = \pi(\Ker(N[x][y] \cdot M[y])) = \pi(\Stab_E([y] \cap x)) = \Stab_E([y'] \cap x')$. 
\end{proof}

The proofs of the next lemmas are given in the appendix. They concern the symmetries of the constructed $N$-matrices, and essentially follow directly from the construction of the first $N$-matrix and the fact that the other $N$-matrices are symmetric to it. The formal proofs involve tedious calculations, though.

\begin{restatable}{lemma}{restateNMatricesPermuteWithStabiliser}
	\label{lem:XOR_NmatricesPermuteWithStabiliser}
	Let $[x'] \in \Omega_{[x]}$ and $[y'] \in \Cc[x']$. There exists a homomorphism $h'$ from $g_{[y']}(\Stab_G([y'] \cap x')) \leq \Sym(I_{[y']})$ into the symmetric group on the row index set of $N[x'][y']$ such that for every $\sigma \in g_{[y']}(\Stab_G([y'] \cap x'))$, it holds $(h'(\sigma), \sigma)(N[x'][y']) = N[x'][y']$.
\end{restatable}
\noindent
\textit{Proof sketch.}
Let $[y] \in \Cc[x]$ be the primer of $\Omega_{[y']}$. Then for $N[x][y]$, the lemma holds by construction because we explicitly closed the rows of $N[x][y]$ under these symmetries. For $N[x'][y']$, the result follows by symmetry of the construction.
\hfill \qed \\
\\
Lemmas \ref{lem:XOR_constructedKernelCorrect} and \ref{lem:XOR_NmatricesPermuteWithStabiliser} assert that all the constructed matrices $N[x'][y']$ satisfy the properties (i) and (ii) that $N[x][y]$ has by construction.\\
Now let 
\[
J_{\Omega_{[x]}} := \bigcup_{\stackrel{[x'] \in \Omega_{[x]}}{[y'] \in \Cc[x']}} J_{[x'][y']}.
\]
We provide a group homomorphism $h_{[x]} : \Stab_G(\mu) \lra \Sym(J_{\Omega_{[x]}})$ such that for each $[x'] \in \Omega_{[x]}$ and each $[y'] \in \Cc[x']$, it holds $h_{[x]}(\pi)(J_{[x'][y']}) = J_{\pi[x']\pi[y']}$. Furthermore, we want that $(N\pi[x']\pi[y']) = (h_{[x]}(\pi),g_{[x]}(\pi))N[x'][y']$ for each $[x'] \in \Omega_{[x]}$, $[y'] \in \Cc[x']$, and $\pi \in \Stab_G(\mu)$.
For each triple $(k,[x'],[y']) \in J_{\Omega_{[x]}}$, we set:
\[
h_{[x]}(\pi)(k,[x'],[y']) := (\ell, \pi[x'],\pi[y']),
\]
where $\ell$ is defined as follows: Let 
\[
t := |\{ i \in \{1,2,...,k-1 \} \mid (N[x'][y'])_{i,-} = (N[x'][y'])_{k,-}  \}|.
\]
That is, for the $k$-th row of $N[x'][y']$, there are $t$ rows identical to it with a smaller index. Then $\ell$ is defined such that
\[
|\{ i \in \{1,2,...,\ell-1 \} \mid (N\pi[x']\pi[y'])_{i,-} = g_{[y']}(\pi)((N[x'][y'])_{k,-})  \}| = t,
\]
and such that $(N\pi[x']\pi[y'])_{\ell,-} = g_{[y']}(\pi)((N[x'][y'])_{k,-})$. In other words, $\ell$ is the $(t+1)$st row of $N\pi[x']\pi[y']$ which is equal to the $k$-th row of $N[x][y]$, up to a permutation of the columns given by $g_{[y']}(\pi)$.\\
We have to argue that $h_{[x]}(\pi)(k,[x'],[y'])$ is indeed well-defined:
\begin{restatable}{lemma}{restateHwellDefined}
	\label{lem:XOR_hWellDefined}
	Let $[x'] \in \Omega_{[x]}, [y'] \in \Cc[x']$. Let $\pi \in \Stab_G(\mu)$. Then for every $(k,[x'],[y']) \in J_{[x'][y']}$, the number of rows of $N[x'][y']$ which are equal to the $k$-th row is the same as the number of rows of $N\pi[x']\pi[y']$ that are equal to the $k$-th row of $N[x'][y']$, up to application of $g_{[y']}$ to the columns. Formally:
	\begin{align*}
		&|\{ (i,[x'],[y']) \in J_{[x'][y']} \mid (N[x'][y'])_{(i,[x'],[y']),-} = (N[x'][y'])_{k,-} \}| \\
		= &| \{ (i,\pi[x'],\pi[y']) \in J_{\pi[x']\pi[y']} \mid (N\pi[x']\pi[y'])_{(i,\pi[x'],\pi[y']),-} = g_{[y']}((N[x'][y'])_{k,-}) \}   |   
	\end{align*}
\end{restatable}	
\noindent
\textit{Proof sketch.} Follows again from the fact that $N[x'][y']$ and $N\pi[x']\pi[y']$ are by construction symmetric to each other. \hfill \qed\\

This lemma shows that we can indeed define $h_{[x]}$ as we did. Now one can verify that $h_{[x]}$ is a group homomorphism with the desired properties. Again, we prove this in the appendix; it follows from the definition of $h_{[x]}$ and Lemma \ref{lem:XOR_hWellDefined}.
\begin{lemma}
	\label{lem:XOR_matrixHomomorphismHcorrect}
	The mapping $h_{[x]} : \Stab_G(\mu) \lra \Sym(J_{\Omega_{[x]}})$ is a group homomorphism.
For every $[x'] \in \Omega_{[x]}$, $[y'] \in \Cc[x']$ and each $\pi \in \Stab_G(\mu)$, it holds $h_{[x]}(\pi)(J_{[x'][y']}) = J_{\pi[x']\pi[y']}$. Furthermore, $N\pi[x']\pi[y']= (h_{[x]}(\pi),g_{[y']}(\pi))N[x'][y']$ for each $[x'] \in \Omega_{[x]}$, $[y'] \in \Cc[x']$, and $\pi \in \Stab_G(\mu)$.
\end{lemma}	
This finishes the construction and correctness proof of the matrices $N[x'][y']$ for all $[x'] \in \Omega_{[x]}$ and $[y'] \in \Cc[x']$ and of the associated homomorphism $h_{[x]} : \Stab_G(\mu) \lra J_{\Omega_{[x]}}$. Items (b) and (d) of the inductive step are thus covered. To complete the inductive step we still have to define the matrices $M[x']$ for all $[x'] \in \Omega_{[x]}$.\\

\paragraph*{Definition of the M-matrices}
Let $[x'] \in \Omega_{[x]}$. Informally, $M[x']$ is obtained by collecting all the rows of the matrices $(N[x'][y'] \cdot M[y'])$, for all $[y'] \in \Cc[x']$, and putting them together as the rows of $M[x']$. Formally, let 
\[
I_{[x']} := \biguplus_{[y'] \in \Cc[x']} J_{[x'][y']}. 
\] 
Note that by construction, the sets $J_{[x'][y']}, J_{[x'][y'']}$ are pairwise disjoint if $[y'] \neq [y'']$. Then the rows of $M[x'] \in \bbF_2^{I_{[x']} \times E}$ are defined as follows: For $[y'] \in \Cc[x']$ and $(i,[x'],[y']) \in J_{[x'][y']}$, we let
\[
M[x']_{(i,[x'],[y']),-} := (N[x'][y'] \cdot M[y'])_{(i,[x'],[y']),-}.
\]
This matrix has the desired kernel:
\begin{lemma}
	\label{lem:XOR_kernelMmatrixCorrect}
	For every $[x'] \in \Omega_{[x]}$, the matrix $M[x']$ defined as above satisfies:
	\[
	\Ker(M[x']) = \Stab_E(x').
	\]
\end{lemma}	
\begin{proof}
	By definition of $M[x']$, a vector $\mathbf{v} \in \bbF_2^E$ is in $\Ker(M[x'])$ if and only if 
	$\mathbf{v} \in \Ker(N[x'][y'] \cdot M[y'])$ for all $[y'] \in \Cc[x']$. By Lemma \ref{lem:XOR_constructedKernelCorrect}, this is the case iff $\mathbf{v} \in \Stab_E([y'] \cap x')$, for all $[y'] \in \Cc[x']$. That is to say, 
	\[
	\mathbf{v} \in \bigcap_{[y'] \in \Cc[x']} \Stab_E([y'] \cap x').
	\]
	This is equivalent to $\mathbf{v} \in \Stab_E(x')$ because $x' = \bigcup_{[y'] \in \Cc[x']} ([y'] \cap x')$ (see Proposition \ref{prop:XOR_stabiliserOfSet}).
\end{proof}	

Finally, we have to provide the homomorphism $g_{[x]}: \Stab_G(\mu) \lra \Sym(I_{\Omega_{[x]}})$, where $I_{\Omega_{[x]}} = \biguplus_{[x'] \in \Omega_{[x]}} I_{[x']}$.
Note that $I_{\Omega_{[x]}} = J_{\Omega_{[x]}}$ by definition of the index sets $I_{[x']}$.   
Therefore we can simply set $g_{[x]} := h_{[x]}$. This homomorphism indeed satisfies the desired properties:
\begin{lemma}
	\label{lem:XOR_homomorphismGcorrect}
	For each $[x'] \in \Omega_{[x]}$ and each $\pi \in \Stab_G(\mu)$, it holds $g_{[x]}(\pi)(I_{[x']}) = I_{\pi[x']}$. Furthermore, $M\pi[x'] = (g_{[x]}(\pi),\pi)M[x']$.
\end{lemma}
\begin{proof}
	Let $\pi \in \Stab_G(\mu)$ and $[x'] \in \Omega_{[x]}$. Due to Lemma \ref{lem:XOR_orbitsOfSimClasses}, $\pi[x']$ is again a $\sim$-class in $\tc(\mu)$. Because $\pi$ extends to an automorphism of $\mu$, it also holds 
\[
\pi(\Cc[x']) = \{ \pi[y'] \mid [y'] \in \Cc[x']  \} = \Cc(\pi[x']).
\]
Hence, by Lemma \ref{lem:XOR_matrixHomomorphismHcorrect}, $h_{[x]}(\pi)$ maps the set $I_{[x']} = \biguplus_{[y'] \in \Cc[x']} J_{[x'][y']}$ to the set $I_{\pi[x']} = \biguplus_{\pi[y'] \in \Cc\pi[x']} J_{\pi[x']\pi[y']}$. 
It remains to show $(M\pi[x'])_{g_{[x]}(\pi)(i),\pi(e)} = M[x']_{i,e}$ for each $i = (j,[x'],[y']) \in I_{[x']}$, and $e \in E$.
By definition of the $M$-matrices, we have: 
\begin{align*}
	(M\pi[x'])_{g_{[x]}(\pi)(i),\pi(e)} = (M\pi[x'])_{(j',\pi[x'],\pi[y']),\pi(e)} &=  (N\pi[x']\pi[y'] \cdot M\pi[y'])_{(j',\pi[x'],\pi[y']),\pi(e)}\\
	&= (N[x'][y'] \cdot M[y'])_{(j,[x'],[y']),e}\\ 
	&= M[x']_{i,e}.
\end{align*}	
Here, $j'$ is such that $(j',\pi[x'],\pi[y']) = g_{[x]}(\pi)(j,[x'],[y']) = h_{[x]}(\pi)(j,[x'],[y'])$. The final equality holds for the following reason: $(N\pi[x']\pi[y'] \cdot M\pi[y'])_{(j',\pi[x'],\pi[y']),\pi(e)}$ is the product of row $j'$ of $N\pi[x']\pi[y']$ with column $\pi(e)$ of $M\pi[y']$. The lemma we are currently proving already holds for $M[y']$ by induction hypothesis, so $(M\pi[y'])_{-\pi(e)} = g_{[y']}(\pi)(M[y']_{-e})$. Also, we have $(N\pi[x']\pi[y'])_{j',-} = g_{[y']}(\pi)((N[x'][y'])_{j,-})$ by Lemma \ref{lem:XOR_matrixHomomorphismHcorrect}. Then Proposition \ref{prop:XOR_matrixMultiplicationAndPermutation} tells us that $(N\pi[x']\pi[y'])_{j',-} \cdot (M\pi[y'])_{-,\pi(e)} = (N[x'][y'])_{j,-}\cdot (M[y'])_{-e}$ because these vector products are really the same sums of products, where the summands are just reordered by $g_{[y']}(\pi)$.
\end{proof}	

\subsection{Construction of the circuit}
We construct the circuit $\widehat{C}(\mu) = (V_C,E_C)$ from the matrices that we have defined in the previous section.
The set of gates is 
\[
V_C := \biguplus_{\stackrel{x \in \tc(\mu)}{(i,[x],[y]) \in I_{[x]}} } \mathbf{g}_{i,[x],[y]}.
\]
So every row of any of the $M[x]$-matrices with index $(i,[x],[y])$ will correspond to a gate $\mathbf{g}_{i,[x],[y]}$. The construction will ensure that 
\[
\chi(\Xx(\mathbf{g}_{i,[x],[y]}))= M[x]_{(i,[x],[y]),-}^T.
\]
Thus, the gate will compute the XOR over precisely the input gates labelled with edges that have a $1$-entry in the row $M[x]_{(i,[x],[y]),-}$.\\
\\
For every atom $x \in \tc(\mu)$, we have $I_{[x]} = \{ [x] \}$. We define the corresponding gate $\mathbf{g}_{[x]}$ as an input gate of $\widehat{C}(\mu)$ with $\ell(\mathbf{g}_{[x]}) = e$, where $e \in E$ is the edge such that $[x] \subseteq \{e_0,e_1\}$.\\
\\
If $x \in \tc(\mu)$ is not an atom, then for every $(i,[x],[y]) \in I_{[x]}$, $\mathbf{g}_{i,[x],[y]}$ is an internal gate that has incoming wires from exactly those gates $\mathbf{g}_{j,[y],[z]}$ such that the matrix $N[x][y]$ has a $1$-entry in row $(i,[x],[y])$ and column $(j,[y],[z])$. Note that $(j,[y],[z]) \in I_{[y]}$, so $[z] \in \Cc[y]$. Formally, the set of children of $\mathbf{g}_{i,[x],[y]}$ is:
\[
\mathbf{g}_{i,[x],[y]}E_C := \{ \mathbf{g}_{j,[y],[z]}  \mid (j,[y],[z]) \in I_{[y]}, (N[x][y])_{(i,[x],[y]),(j,[y],[z])} = 1 \}.
\] 
It remains to specify the root of $\widehat{C}(\mu)$. Consider the matrix $M[\mu]$. Let $(i,[\mu],[y])$ be a row of $M[\mu]$ with a maximum number of one-entries. We define the root $r$ to be the gate $\mathbf{g}_{i,[\mu],[y]}$.\\

The figure below shows the matrices and XOR-gates for an object $x$ consisting of two connected components $[y] \cap x$ and $[y'] \cap x$. The matrices $N[x][y]$ and $N[x][y']$ have just one row each in this example, and we are assuming that $\Cc[y] = \{[z]\}$ and $\Cc[y'] = \{[z']\}$, and that $M[y] = N[y][z] \cdot M[z]$ has three rows and $M[y'] = N[y'][z'] \cdot M[z']$ has two rows.
\begin{figure}[H]
	\centering
	\begin{tikzpicture}
				
	\draw[radius=1.5,fill=yellow, opacity=0.3] (1.3+7,-6) circle;
	\draw[radius=1.5,fill=red, opacity=0.3] (1.5+10.5,-6) circle;
	
	\draw[radius=0.5,fill=yellow, opacity=0.3] (0.5,-6) circle;
	\draw[radius=0.5,fill=red, opacity=0.3] (5.5,-6) circle;
	
	\draw[radius=0.3,fill=yellow, opacity=0.3] (1.85,1.4) circle;
	\draw[radius=0.3,fill=red, opacity=0.3] (2.9,1.4) circle;

	\node at (2,2) {$x = \{\underbrace{ y_1, y_2}, \underbrace{y'_1, y'_2}  \}$};
	\node at (1.85,1.4) {$[y]$};
	\node at (2.9,1.4) {$[y']$};
	
	\node (Mx) at (2,0) {$M[x] = \left[   \begin{array}{c} (N[x][y] \cdot M[y]) \\ (N[x][y']\cdot M[y'])            \end{array}  \right] $};
	\node (My) at (0.5,-2*3) {$M[y]$};	
	\node (MyPrime) at (5.5,-2*3) {$M[y']$};		
	
	\node (Nxy) at (1,-2*1.5) {$N[x][y] = \left[ \begin{array}{ccc} 1 & 1& 1  \end{array} \right]$};	
	\node (NxyPrime) at (5,-2*1.5) {$N[x][y'] = \left[ \begin{array}{cc} 1 & 0 \end{array}  \right]$};		
	
	\node (xyGate) at (1.5+7,0.2) {$\oplus \mathbf{g}_{1,[x],[y]}$};
	\node (xyPrimeGate) at (1.5+9,-0.3) {$\oplus \mathbf{g}_{1,[x],[y']}$};
	
	\node (yGate1) at (1.5+6,-2*3) {$\oplus \mathbf{g}_{1,[y],[z]}$};
	\node  (yGate2)at (1.5+7,-2*2.5) {$\oplus \mathbf{g}_{2,[y],[z]}$};
	\node  (yGate3) at (1.5+8,-2*3) {$\oplus \mathbf{g}_{3,[y],[z]}$};
	\node (yPrimeGate) at (1.5+11,-2*2.5) {$\oplus \mathbf{g}_{1,[y'],[z']}$};
	\node at (13,-2*3) {$\oplus \mathbf{g}_{2,[y'],[z']}$};
	
	\draw[->] (xyGate.south)-- (yGate1);
	\draw[->] (xyGate.south)-- (yGate2);
	\draw[->] (xyGate.south)-- (yGate3);
	
	\draw[->] (xyPrimeGate.south)-- (yPrimeGate.west);
	
	\draw[dashed] (Mx.south)-- (Nxy);
	\draw[dashed] (Mx.south)-- (NxyPrime);
	\draw[dashed]  (Nxy.south)-- (My);
	\draw[dashed] (NxyPrime.south)-- (MyPrime);

	\end{tikzpicture}
	\caption{Example showing how the XOR-gates are connected according to the $N$-matrices.}
\end{figure}
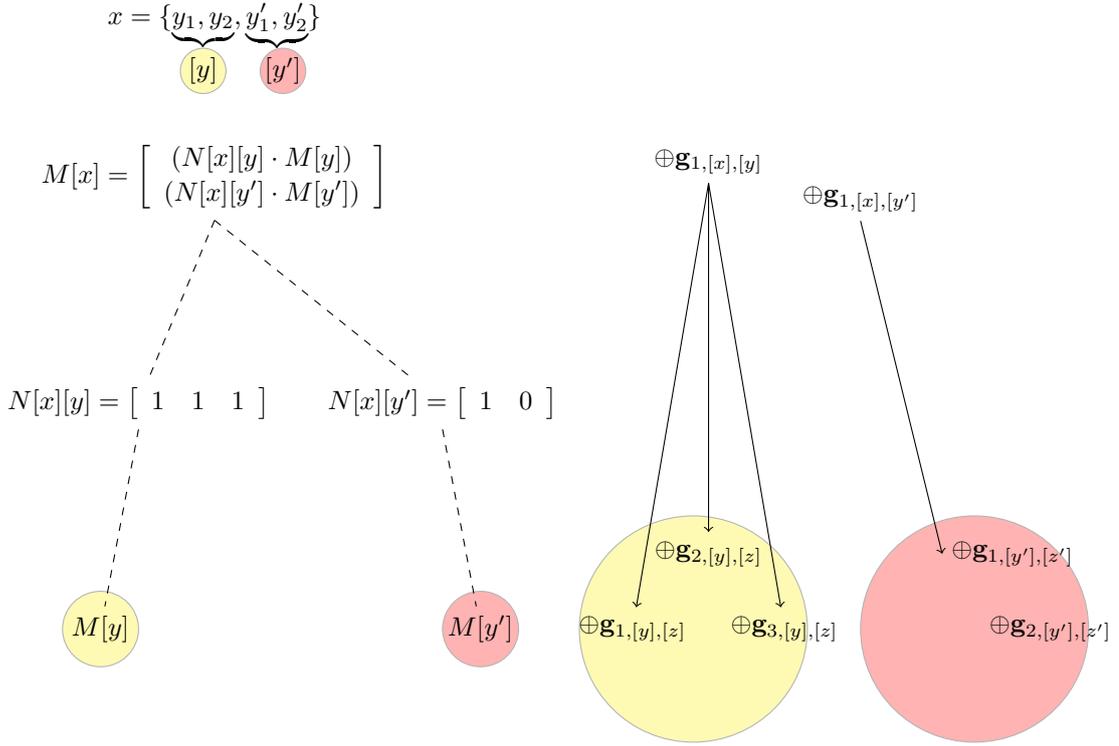	
Now we prove that $\widehat{C}(\mu)$ has the desired properties.
\begin{lemma}
	\label{lem:XOR_CHatIsSymmetric}
	Every $\pi \in \Stab_G(\mu)$ extends to an automorphism of the circuit $\widehat{C}(\mu)$, that is:
	\[
	\Stab_G(\mu) \leq \Stab_G(\widehat{C}(\mu)).
	\]
\end{lemma}	
\begin{proof}
	Let $\pi \in \Stab_G(\mu)$. We claim that the following mapping $\sigma : V_C \lra V_C$ is an automorphism of $\widehat{C}(\mu)$ that $\pi$ extends to. We let
	\[
	\sigma(\mathbf{g}_{i,[x],[y]}) := \mathbf{g}_{g_{[x]}(\pi)(i,[x],[y])}, 
	\]
	where $g_{[x]} : \Stab_G(\mu) \lra \Sym(I_{\Omega_{[x]}})$ is the group homomorphism for $\Omega_{[x]}$ from the construction of the matrices in the previous subsection. We have to show three things about $\sigma$. Firstly, that $\sigma \in \Sym(V_C)$. Secondly, that $\sigma$ maps wires to wires and non-wires to non-wires of $\widehat{C}(\mu)$. Finally, that for every input gate $\mathbf{g}_{[x]}$ it holds: $\ell(\sigma(\mathbf{g}_{[x]})) = \pi(\ell(\mathbf{g}_{[x]}))$. This last statement actually follows directly from the definition of $g_{[x]}$ for atoms $x \in \tc(\mu)$: If $[x] \subseteq \{e_0,e_1\}$, then $\pi(\ell(\mathbf{g}_{[x]})) = \pi(e)$, and $g_{[x]}(\pi)([x]) = \pi[x] \subseteq \{\pi(e_0), \pi(e_1)\}$, so $\ell(\sigma(\mathbf{g}_{[x]})) = \pi(e)$.\\
	The fact that $\sigma \in \Sym(V_C)$ follows because for each orbit $\Omega_{[x]}$, we have that $g_{[x]}(\pi) \in \Sym(I_{\Omega_{[x]}})$, and the set of gates $V_C$ can be partitioned into these orbits so that each part has the form $\{ \mathbf{g}_{i,[x'],[y]} \mid (i,[x'],[y]) \in I_{\Omega_{[x]}} \}$, for some orbit $\Omega_{[x]}$.\\
	It remains to prove that $\sigma$ preserves the wire structure of the circuit. For any two gates $\mathbf{g}_{i,[x],[y]}, \mathbf{g}_{j,[x'],[y']}$, we have 
	\[
	(\mathbf{g}_{i,[x],[y]}, \mathbf{g}_{j,[x'],[y']}) \in E_C \text{ if and only if } [y] = [x'] \text{ and }(N[x][y])_{(i,[x],[y]),(j,[y],[y'])} = 1.
	\]
	The latter equation holds if and only if 
	\[
	(N\pi[x]\pi[y])_{h_{[x]}(\pi)(i,[x],[y]),g_{[y]}(\pi)(j,[y],[y'])}  = 1.  
	\]
	This is true by Lemma \ref{lem:XOR_matrixHomomorphismHcorrect}. If $[x'] = [y]$, then we have $\sigma(\mathbf{g}_{j,[x'],[y']}) = \mathbf{g}_{g_{[y]}(\pi)(j,[y],[y'])}$. Furthermore, it holds $\sigma(\mathbf{g}_{i,[x],[y]}) = \mathbf{g}_{g_{[x]}(\pi)(i,[x],[y])}$. By the definition in the previous section, $g_{[x]} = h_{[x]}$, and it holds that $g_{[x]}(\pi)(i,[x],[y])$ is of the form $(\ell, \pi[x],\pi[y])$, and $g_{[y]}(\pi)(j,[y],[y'])$ is of the form $(\ell', \pi[y],\pi[y'])$, so these are indeed row- and column-indices of the matrix $N\pi[x]\pi[y]$. So altogether, the above equation is equivalent to:
	$
	(\sigma \mathbf{g}_{i,[x],[y]}, \sigma \mathbf{g}_{j,[x'],[y']}) \in E_C.
	$
\end{proof}	
As in the case of the circuit $C(\mu)$, this lemma implies that the orbit-size of $\widehat{C}(\mu)$ cannot be greater than the orbit-size of $\mu$ (here, $\Orb_G(\mu)$ again refers to the $\Aut(G)$-orbit and not to the $\Stab_G(\mu)$-orbit, which would be pointless).
\begin{corollary}
	\label{cor:XOR_orbitOfCircuitHat}
	\[
	|\Orb_G(\widehat{C}(\mu))| \leq |\Orb_G(\mu)|.
	\]
\end{corollary}	
\begin{proof}
	Analogous to the proof of Corollary \ref{cor:XOR_orbitOfCircuit}.
\end{proof}

\begin{lemma}
	\label{lem:XOR_gateSensitivityExactlyRowsOfMatrix}
	For every internal gate $\mathbf{g} := \mathbf{g}_{i,[x],[y]} \in V_C$ it holds:
	\[
	\Xx(\mathbf{g}) = \chi^{-1}(M[x]_{(i,[x],[y]),-}^T).
	\]
\end{lemma}	
\begin{proof}
	By induction. For each input gate $\mathbf{g}_{[x]}$ with $[x] \subseteq \{e_0,e_1\}$ we have $\Xx(\mathbf{g}_{[x]}) = \{ e\}$.\\
	Now let $\mathbf{g} = \mathbf{g}_{i,[x],[y]}$ be an internal gate. By definition of $\widehat{C}(\mu)$, it holds
	\[
	\mathbf{g}E_C = \{ \mathbf{g}_{j,[y],[z]}  \mid (j,[y],[z]) \in I_{[y]}, (N[x][y])_{(i,[x],[y]),(j,[y],[z])} = 1 \}.
	\]
	Thus,
	\begin{align*}
		\Xx(\mathbf{g}) &= \bigtriangleup_{ \stackrel{(j,[y],[z]) \in I_{[y]},} {(N[x][y])_{(i,[x],[y]),(j,[y],[z])} = 1 }} \Xx(\mathbf{g}_{j,[y],[z]})\\
		&= \chi^{-1}\Big( \sum_{ \stackrel{(j,[y],[z]) \in I_{[y]},} {(N[x][y])_{(i,[x],[y]),(j,[y],[z])} = 1 }} \chi(\Xx(\mathbf{g}_{j,[y],[z]})) \  \  \mod 2 \Big)\\
		&= \chi^{-1}(((N[x][y])_{(i,[x],[y]),-} \cdot M[y])^T).
	\end{align*}
	The last step uses the induction hypothesis and the fact that $I_{[y]}$ is the row-index-set of $M[y]$. By the definition of the matrix $M[x]$ in the previous subsection, we have $(N[x][y])_{(i,[x],[y]),-} \cdot M[y] = M[x]_{(i,[x],[y]),-}$. This finishes the proof.
\end{proof}

\begin{lemma}
	\label{lem:XOR_dimensionOfCircuitHat}
	The \emph{fan-in dimension} of $\widehat{C}(\mu)$ is at most $\log (\maxOrb_E(\mu))$.
\end{lemma}
\begin{proof}
	Consider any gate $\mathbf{g} := \mathbf{g}_{i,[x],[y]} \in V_C$. By definition of $\widehat{C}(\mu)$, the children of $\mathbf{g}$ are 
	\[
	\mathbf{g}E_C = \{ \mathbf{g}_{j,[y],[z]}  \mid (j,[y],[z]) \in I_{[y]}, (N[x][y])_{(i,[x],[y]),(j,[y],[z])} = 1 \}.
	\]
	By Lemma \ref{lem:XOR_gateSensitivityExactlyRowsOfMatrix}, we have for each of these children:
	\[
	\Xx(\mathbf{g}_{j,[y],[z]}) = \chi^{-1}(M[y]_{(j,[y],[z]),-}).
	\]
	Thus, the gate matrix $M\mathbf{g} \in \bbF_2^{\mathbf{g}E_C \times E}$ is a submatrix of $M[y]$: It consists precisely of those rows $(j,[y],[z])$ of $M[y]$ such that $(N[x][y])_{(i,[x],[y]),(j,[y],[z])} = 1$. Therefore, 
	$\rk(M\mathbf{g}) \leq \rk (M[y]).$
	Now we can argue as in the proof of Lemma \ref{lem:XOR_dimensionOfCircuit} in order to bound $\rk (M[y])$. Using the Orbit-Stabiliser Theorem, we obtain again:
	\[
	\log(\maxOrb_E(\mu)) \geq |E| - \dim \Stab_E(y).
	\] 
	By Lemma \ref{lem:XOR_kernelMmatrixCorrect}, we have $\Ker(M[y]) = \Stab_E(y)$. Therefore, with the Rank Theorem we obtain:
	\[
	\rk(M[y]) = |E| - \dim \ \Stab_E(y) \leq \log (\maxOrb_E(\mu)).  
	\]
	In total, we have $\rk(M\mathbf{g}) \leq  \log (\maxOrb_E(\mu))$. Since $\mathbf{g} \in V_C$ was arbitrary, this is a bound on the fan-in dimension of $\widehat{C}(\mu)$.  
\end{proof}

\begin{lemma}
	\label{lem:XOR_rootSensitivityCHat}
	For the root $r$ of $\widehat{C}(\mu)$, it holds
	\[
	| \Xx(r) | \geq \frac{|\sup_{\CFI}(\mu)|}{\log (\maxOrb_E(\mu))}.
	\]
\end{lemma}	
\begin{proof}
	By definition of $\widehat{C}(\mu)$, we have $r = \mathbf{g}_{i,[\mu],[y]}$ for some $(i,[\mu],[y]) \in I_{[\mu]}$ such that $M[\mu]_{(i,[\mu],[y]),-}$ is a row with a maximum number of $1$-entries. We have
	\[
	\Xx(r) = \chi^{-1}(M[\mu]_{(i,[\mu],[y]),-}^T)
	\]
	according to Lemma \ref{lem:XOR_gateSensitivityExactlyRowsOfMatrix}. Therefore, we have to show that $M[\mu]$ has a row with at least $\frac{|\sup_{\CFI}(\mu)|}{\log (\maxOrb_E(\mu))}$ many $1$-entries.\\
	\\
	\textbf{Claim:} For every $e \in \sup_{\CFI}(\mu)$ there is a row of $M[\mu]$ which is non-zero in column $e$.\\
	\textit{Proof of claim:} Since $\sup_{\CFI}(\mu)$ is the minimal CFI-support of $\mu$, it holds that $\chi(e) \notin \Stab_E(\mu)$. Otherwise there is a smaller support not containing $e$. Suppose for a contradiction that all rows of $M[\mu]$ are zero in column $e$. Then $\chi(e) \in \Ker(M[\mu])$. But this contradicts the fact that $\Ker(M[\mu]) = \Stab_E(\mu)$ (Lemma \ref{lem:XOR_kernelMmatrixCorrect}). This proves the claim.\\
	Now take a subset $\Bb$ of the rows of $M[\mu]$ that forms a basis of the row space of $M[\mu]$. By Lemma \ref{lem:XOR_dimensionOfCircuitHat}, we have $|\Bb| \leq \log (\maxOrb_E(\mu))$. For every $e \in E$ such that some row of $M[\mu]$ is non-zero in column $e$, there must also be a row in $\Bb$ that is non-zero in column $e$ (else $\Bb$ does not generate the whole row space). So by the claim and by the size bound on $\Bb$, there is a row in $\Bb$ with at least $\frac{|\sup_{\CFI}(\mu)|}{\log (\maxOrb_E(\mu))}$ many $1$-entries.
\end{proof}	

So far we have shown \textbf{Properties} \textbf{1}, \textbf{2} and \textbf{3} from Theorem \ref{thm:XOR_generalCircuitTheorem}. \textbf{Property 1} is proved by Corollary \ref{cor:XOR_orbitOfCircuitHat}, \textbf{Property 2} by Lemma \ref{lem:XOR_rootSensitivityCHat}, and \textbf{Property 3} by Lemma \ref{lem:XOR_dimensionOfCircuitHat}. 
It remains to estimate the size of the circuit under the assumption that every space $\Stab_E([y] \cap x)$ has a symmetric basis.

\subsection{Bounding the size of the circuit}
\label{sec:XOR_sizeBound}
\begin{lemma}
	\label{lem:XOR_circuitSize}
	Let $\text{Atoms}(\mu) \subseteq \tc(\mu)$ denote the set of atoms in $\tc(\mu)$, and $\text{Sets}(\mu) := \tc(\mu) \setminus \text{Atoms}(\mu)$.
	The size of $\widehat{C}(\mu)$ is
	\[
	|V_C| = |\{e \in E \mid e_0 \text{ or } e_1 \in \text{Atoms}(\mu)\}| + \sum_{\stackrel{([x],[y]), x \in \text{Sets}(\mu),}{[y] \in \Cc[x]}} |J_{[x][y]}|.
	\]
	In other words: The size of $\widehat{C}(\mu)$ is determined by the total number of rows of all $N[x][y]$-matrices.
\end{lemma}	
\begin{proof}
	By definition of $\widehat{C}(\mu)$, $|V_C| = \sum_{[x], x \in \tc(\mu)} | I_{[x]} |$. For each $x \in \tc(\mu) \cap \text{Sets}(\mu)$, we have $|I_{[x]}| = \sum_{[y] \in \Cc[x]} |J_{[x][y]}|$. Every set $J_{[x][y]}$ is exclusively associated with the pair $([x],[y])$, so 
	\[
	\sum_{[x], x \in \tc(\mu)} | I_{[x]} | = \sum_{\stackrel{([x],[y]), x \in \text{Sets}(\mu),}{[y] \in \Cc[x]}} |J_{[x][y]}|.
	\]
	For each atomic $x \in \tc(\mu)$, $|I_{[x]}| = 1$.
\end{proof}

Hence, in order to bound $|V_C|$, we have to bound the number of rows of each of the $N[x][y]$-matrices. To do so, we revisit the proofs of Lemmas \ref{lem:XOR_kernelMatrixExists} and \ref{lem:XOR_existenceOfNMatrixWithConditionsAB}. We will see that the matrix that is constructed in these lemmas can be chosen to have polynomial size if a symmetric basis for $\Stab_E([y] \cap x)$ exists. The first step in the construction of $N[x][y]$ in Lemma \ref{lem:XOR_kernelMatrixExists} is to solve a family of linear equation systems. We now show that the symmetries of such systems correspond to symmetries of their solutions.
\begin{lemma}
	\label{lem:XOR_symmetricSolutionLemma}
	Let $I,J$ be abstract index sets and $A,A' \in \bbF_2^{I \times J}, \mathbf{b}, \mathbf{b}' \in \bbF_2^I$ such that the linear equation systems $A \cdot \mathbf{x} = \mathbf{b}$ and $A' \cdot \mathbf{x} = \mathbf{b}'$ each have a unique solution $\mathbf{s}, \mathbf{s'}$, respectively. Let $\mathbf{G} \leq \Sym(J)$ be a permutation group acting on the columns of $A$. Let $\pi \in \mathbf{G}$ a permutation that maps the row-set of the extended coefficient matrix $(A|\mathbf{b})$ 
	\[
	R = \{ (A_{i-},\mathbf{b}_i) \mid i \in I \}
	\]
	to the row-set $R'$ of $(A'|\mathbf{b}')$. Here, the action of $\pi$ on $R$ is $\pi(A_{i-},\mathbf{b}_i) := (\pi(A_{i-}),\mathbf{b}_i)$. Then $\pi(\mathbf{s}) = \mathbf{s}'$.
\end{lemma}	
\begin{proof}
	We only have to show that $\pi(\mathbf{s})$ is a solution of $A' \cdot \mathbf{x} = \mathbf{b}'$. Then we have $\pi(\mathbf{s}) = \mathbf{s}'$ by uniqueness of the solution. 
	If $\pi \in \mathbf{G}$ maps $R$ to $R'$, then there is a permutation $\sigma \in \Sym(I)$ that is induced by the action of $\pi$, i.e.\ $\pi(A_{i-},\mathbf{b}_i) = (A'_{\sigma i,-}, \mathbf{b}'_{\sigma i})$ for all $i \in I$. For every $i \in I$, it holds $A_{i-} \cdot \mathbf{s} = \mathbf{b}_i$, because $\mathbf{s}$ is a solution to the equation system. Since (by Proposition \ref{prop:XOR_matrixMultiplicationAndPermutation}) $A_{i-} \cdot \mathbf{s} = \pi(A_{i-}) \cdot \pi(\mathbf{s}) = A'_{\sigma i -} \cdot  \pi(\mathbf{s})$, and $\mathbf{b}_i = \mathbf{b}'_{\sigma i}$, $\pi(\mathbf{s})$ is a solution to the equation $A'_{\sigma i -} \cdot \mathbf{x} = \mathbf{b}'_{\sigma i}$. Since every row of $(A'|\mathbf{b}')$ has such a preimage under $\pi$, it follows that $\pi(\mathbf{s})$ is a solution for every equation in $A' \cdot \mathbf{x} = \mathbf{b}'$. By assumption, the equation system has a unique solution. Therefore, $\pi(\mathbf{s}) = \mathbf{s}'$.
\end{proof}

The next lemma will become interesting once we are dealing with symmetric bases of vector spaces. It shows that the permutation invariance of a set of vectors (for example a symmetric basis) is preserved under linear maps and appropriate group homomorphisms. 
\begin{lemma}
	\label{lem:XOR_symmetricBasisPreservedByLinearMap}
	Let $M \in \bbF_2^{I \times J}$ and let $\mathbf{G} \leq \Sym(J)$ be a permutation group acting on the columns of $M$. Let $B \subseteq \bbF_2^J$ be a set of vectors and $S := \Stab_{\mathbf{G}}(B)$. Let $g : \mathbf{G} \lra \Sym(I)$ be a group homomorphism such that for all $\pi \in \mathbf{G}$, $(g(\pi),\pi)M = M$.\\
	Then $g(S)$ stabilises the set $M\cdot B = \{ M \cdot \mathbf{v} \mid \mathbf{v} \in B \}$.
\end{lemma}	
\begin{proof}
	Let $\pi \in S$ and $\mathbf{v} \in B$. Let $\mathbf{w} = M \cdot \mathbf{v}$. We show that $(g(\pi))^{-1}(\mathbf{w}) = M \cdot \pi(\mathbf{v})$. For each $i \in I$, we have 
	\[
	M_{i,-} \cdot \pi(\mathbf{v}) = \pi(M_{(g(\pi))(i),-}) \cdot  \pi(\mathbf{v}) = M_{(g(\pi))(i),-} \cdot \mathbf{v}.
	\] 
	The first equality holds because $(g(\pi),\pi)M = M$ by assumption and the second one is due to Proposition \ref{prop:XOR_matrixMultiplicationAndPermutation}. It follows that $g(\pi)(M \cdot \pi(\mathbf{v})) = M \cdot \mathbf{v} = \mathbf{w}$. So $(g(\pi))^{-1}(\mathbf{w}) = M \cdot \pi(\mathbf{v})$. We have $\pi(\mathbf{v}) \in B$ because $\pi \in S$. Therefore, $(g(\pi))^{-1}(\mathbf{w}) \in M \cdot B$. Because $\mathbf{w} = M \cdot \mathbf{v}$ was arbitrary, we know that $(g(\pi))^{-1} = g(\pi^{-1}) \in g(S)$ stabilises the set $M \cdot B$. Since $\pi \in S$ was also arbitrary, $g(S)$ stabilises $M \cdot B$.
\end{proof}	

Finally, we provide the exact definition of what we mean by a \emph{symmetric basis}. This definition is tailored to the spaces $\Stab_E([y] \cap x)$ that occur for the objects in $\tc(\mu)$. When we say ``symmetric basis'', we actually mean two bases: We require that both the basis of $\Stab_E([y] \cap x)$ as well as its extension to a  basis of the ambient space be symmetric. Symmetry is meant in the sense that the orbit must have polynomial size.
\begin{definition}
	\label{def:XOR_symmetricBasis}
	Let $x \in \tc(\mu)$ and $[y] \in \Cc[x]$. We say that the vector space $\Gamma := \Stab_E([y] \cap x) \leq \bbF_2^E$ has a \emph{symmetric basis} if there exist two bases
	\[
	\Bb_\Gamma \subseteq \Bb
	\]
	such that $\Bb_\Gamma$ is a basis of $\Gamma=\Stab_E([y] \cap x)$ and $\Bb$ is a basis of $\bbF_2^E$, and such that: The group
	\begin{align*}
		\Stab_{\Stab_G([y] \cap x)}&(\Bb_\Gamma) \cap \Stab_{\Stab_G([y] \cap x)}(\Bb) =\\
		&\{ \pi \in \Stab_G([y] \cap x) \mid \pi(\Bb_{\Gamma}) = \Bb_{\Gamma} \text{ and } \pi(\Bb) = \Bb \} \leq \Stab_G([y] \cap x)
	\end{align*}
	has index $\leq \text{\upshape poly}(|\GG^S|)$ in $\Stab_G([y] \cap x)$. 
\end{definition}
In the above definition, the polynomial $\text{poly}( \cdot )$ is of course meant to be fixed for the whole family of CFI-instances that we are considering in Theorem \ref{thm:XOR_generalCircuitTheorem}. 
Now let us continue with the main lemma that bounds the size of $N[x][y]$ assuming the existence of a symmetric basis.
\begin{lemma}
	\label{lem:XOR_NmatricesPolynomial}
	Let $x \in \tc(\mu)$ and $[y] \in \Cc[x]$. If $\Stab_E([y] \cap x)$ has a \emph{symmetric basis}, then the number of rows of $N[x][y]$ is polynomial in $|\GG^S|$.
\end{lemma}	
\begin{proof}
We show that there is a Boolean matrix $N$ with a polynomial number of rows which satisfies conditions (i) and (ii) from Lemma \ref{lem:XOR_existenceOfNMatrixWithConditionsAB}. This proves the lemma because $N[x][y]$ is defined as the smallest such matrix.\\
Consider the proof of Lemma \ref{lem:XOR_kernelMatrixExists}. When applied to $\Gamma = M[y] \cdot \Stab_E([y] \cap x)$ and $\Delta = M[y] \cdot \bbF_2^E = \Im(M[y])$, the proof shows that there is a matrix $N \in \bbF_2^{[k] \times I_{[y]}}$, for $k = \dim (\Im \ M[y]) - \dim (M[y] \cdot \Stab_E([y] \cap x))$, such that $\Ker(N) \cap \Im(M[y])= M[y] \cdot \Stab_E([y] \cap x)$. In the proof, $k$ linear equation systems $A_i \cdot \mathbf{x} = \mathbf{b}_i$ are defined, each one with a unique solution. Then, for $i \in [k]$, the $i$-th row of $N$ is defined as the unique solution to the equation system $A_i \cdot \mathbf{x} = \mathbf{b}_i$. The matrix $A_i$ is the same for every $i \in [k]$.
It depends on the choice of a basis for $M[y] \cdot \Stab_E([y] \cap x)$ and an extension to a basis of $\Im(M[y])$. More precisely, let $\{\mathbf{w}_1,...,\mathbf{w}_k\}$ be the vectors that extend the basis of $M[y] \cdot \Stab_E([y] \cap x)$ to a basis of $\Im(M[y])$. 
The rows of $A_i$ are the basis vectors of $\Im(M[y])$, and the vector $\mathbf{b}_i$ has exactly one $1$-entry in the row corresponding to $\mathbf{w}_i$. In this way, $\mathbf{b}_i$ is defined for every $i \in [k]$.\\
One can check that this proof of Lemma \ref{lem:XOR_kernelMatrixExists} still goes through if one uses a \emph{generating set} for the space $M[y] \cdot \Stab_E([y] \cap x)$ instead of a basis for the rows of $A_i$ -- as long as the extension $\{\mathbf{w}_1,...,\mathbf{w}_k\}$ to a basis of the full space $\Im(M[y])$ is a linearly independent set of vectors. This changes nothing and in particular, each equation system $A_i \cdot \mathbf{x} = \mathbf{b}_i$ will still have a unique solution because we have just added some redundant equations.\\

We choose appropriate bases now. Since $\Stab_E([y] \cap x)$ has a symmetric basis by assumption, there are bases $\Bb_\Gamma \subseteq \Bb$ of $\Stab_E([y] \cap x)$ and $\bbF_2^E$, respectively, such that $\Stab_{\Stab_G([y] \cap x)}(\Bb_\Gamma) \cap \Stab_{\Stab_G([y] \cap x)}(\Bb)$ has small index in $\Stab_G([y] \cap x)$. Let $\Bb'_\Gamma := M[y] \cdot \Bb_\Gamma \setminus \{\mathbf{0}\}$ and $\Bb' := M[y] \cdot \Bb \setminus \{\mathbf{0}\}$. Then $\Bb'_\Gamma$ is a generating set for $M[y] \cdot \Stab_E([y] \cap x)$, and $\Bb' \setminus \Bb'_\Gamma$ extends this generating set to a basis of $\Im(M[y])$. Importantly, $\Bb' \setminus \Bb'_\Gamma$ is a linearly independent set of vectors (while $\Bb'_\Gamma$ may be linearly dependent). This is because $\Ker(M[y]) = \Stab_E(y)$ (Lemma \ref{lem:XOR_kernelMmatrixCorrect}), and $\Stab_E(y) \leq \Stab_E([y] \cap x)$. Thus, if there were a subset $K \subseteq \Bb' \setminus \Bb'_\Gamma$ such that $\sum K = 0$, then the sum of the $M[y]$-preimages of the vectors in $K$ would be in $\Ker(M[y]) = \Stab_E(y)$. This cannot be the case because $\Bb_\Gamma$ is a basis for $\Stab_E([y] \cap x) \geq \Stab_E(y)$, so no linear combination of vectors in $\Bb \setminus \Bb_\Gamma$ can be in $\Stab_E(y)$.\\

Now apply Lemma \ref{lem:XOR_symmetricBasisPreservedByLinearMap} to the matrix $M = M[y] \in \bbF_2^{I_{[y]} \times E}$, $\mathbf{G} = \Stab_G([y] \cap x), B = \Bb_\Gamma$ and the homomorphism $g$ defined like this: $g : \Stab_G([y] \cap x) \lra \Sym(I_{[y]})$ maps any $\pi \in \Stab_G([y] \cap x)$ to the $I_{[y]}$-restriction of the permutation $g_{[y]}(\pi) \in \Sym(I_{\Omega_{[y]}})$. This is well-defined and $g(\pi) \in \Sym(I_{[y]})$ because $\Stab_G([y] \cap x) \leq \Stab_G([y])$ by Corollary \ref{cor:XOR_stabYXinStabY}, and $g_{[y]}(\Stab_G([y]))$ maps $I_{[y]}$ to itself by Lemma \ref{lem:XOR_homomorphismGcorrect}.\\
Now it follows with Lemma \ref{lem:XOR_symmetricBasisPreservedByLinearMap} that $g(\Stab_{\Stab_G([y] \cap x)}(\Bb_\Gamma)) \leq \Sym(I_{[y]})$ stabilises the set $M[y] \cdot \Bb_{\Gamma} \subseteq \bbF_2^{I_{[y]}}$. It holds $\Bb'_\Gamma = M[y] \cdot \Bb_{\Gamma}  \setminus \{\mathbf{0}\}$, and the zero-vector forms a singleton orbit with respect to permutations of the entries, so also $\Bb'_\Gamma$ is stabilised by $g(\Stab_{\Stab_G([y] \cap x)}(\Bb_\Gamma))$.
Similarly, by applying Lemma \ref{lem:XOR_symmetricBasisPreservedByLinearMap} to $B = \Bb$, we get that $g(\Stab_{\Stab_G([y] \cap x)}(\Bb))$ stabilises $\Bb'$.\\
Every permutation $\pi \in \Sym(I_{[y]})$ that stabilises the sets $\Bb'_\Gamma$ and $\Bb'$ induces a permutation $\sigma \in \Sym(\Bb')$, whose restriction to $\Bb'_\Gamma$ is a permutation in $\Sym(\Bb'_\Gamma)$. Since the rows of each coefficient matrix $A_i$ are the vectors in $\Bb'$, every $\pi \in \Sym(I_{[y]})$ that fixes $\Bb'$ setwise reorders the rows of $A_i$. The right hand side $\mathbf{b}_i$ has only one $1$-entry in the row corresponding to $\mathbf{w}_i \in \Bb' \setminus \Bb'_\Gamma$. So if $\pi$ also stabilises $\Bb'_\Gamma$, the action of $\pi$ on $(A_i|\mathbf{b}_i)$ moves the row vector $\pi^{-1}(\mathbf{w}_i) \in \Bb' \setminus \Bb'_\Gamma$ to the row where $\mathbf{b}_i$ has its $1$-entry. Up to a reordering of rows, this yields one of the other linear equation systems $\{ (A_i|\mathbf{b}_i) \mid i \in [k] \}$, because in some linear equation system, the equation with coefficient vector $\pi^{-1}(\mathbf{w}_i)$ has a $1$ on the right hand side. So any $\pi$ that stabilises both $\Bb'_\Gamma$ and $\Bb'$ induces a permutation on the set of linear equation systems $\{ (A_i|\mathbf{b}_i) \mid i \in [k] \}$ in the sense of Lemma \ref{lem:XOR_symmetricSolutionLemma} (with the action of column permutations on the row set of an equation system as defined there).
Let
\[
\Ss := g(\Stab_{\Stab_G([y] \cap x)}(\Bb_\Gamma)) \cap g(\Stab_{\Stab_G([y] \cap x)}(\Bb)).
\]
As we argued above, this group fixes both $\Bb'$ and $\Bb'_\Gamma$.  
Therefore, it induces a permutation on the equation systems and so  Lemma \ref{lem:XOR_symmetricSolutionLemma} tells us that $\Ss$ also induces a corresponding permutation on the set
\[
\{ \mathbf{s}_i \in \bbF_2^{I_{[y]}} \mid \mathbf{s}_i \text{ is the unique solution to } A_i \cdot \mathbf{x} = \mathbf{b}_i \}
\]
(where $\Ss$ acts on these vectors by permuting the entries). These solution vectors form exactly the rows of the matrix $N \in \bbF_2^{[k] \times I_{[y]}}$ that is being constructed in the proof of Lemma \ref{lem:XOR_kernelMatrixExists}. Therefore, the group $\Ss$ acting on the columns of $N$ induces corresponding permutations on the rows of $N$. In other words, for every $\pi \in \Ss$ there is a $\sigma \in \Sym_k$ such that $(\sigma, \pi)N = N$. Now we close the rows of $N$ under the action of $g(\Stab_G([y] \cap x))$, exactly like in the proof of Lemma \ref{lem:XOR_existenceOfNMatrixWithConditionsAB}. This yields a matrix satisfying the desired conditions (a) and (b). We now argue that for each row of $N$, only $\text{poly}(k+|\GG^S|)$ many rows are added to form the closure under $g_{[y]}(\Stab_G([y] \cap x))$.\\

To show this, we have to bound the size of
$\Orb(N_{i-}) = \{ \pi(N_{i-}) \mid \pi \in g(\Stab_G([y] \cap x)) \}$. Let $\Orb(N)$ denote the $g(\Stab_G([y] \cap x))$-orbit of the set of rows of $N$, and $\Stab(N) \leq g(\Stab_G([y] \cap x))$ the setwise stabiliser of the set of rows. It holds $|\Orb(N_{i-}) | \leq k \cdot |\Orb(N) |$ because any image of the row $N_{i-}$ is an element of at least one of the $k$-element row sets in $\Orb(N)$. Together with the Orbit-Stabiliser Theorem, we get
\[
|\Orb(N_{i-}) | \leq k \cdot |\Orb(N) | = k \cdot \frac{|g(\Stab_G([y] \cap x)) |}{| \Stab(N) |}.
\]
By what we argued above, we have
$\Ss \leq \Stab(N)$ and thus
\begin{align*}
	\frac{|g(\Stab_G([y] \cap x)) |}{| \Stab(N) |} &\leq \frac{|g(\Stab_G([y] \cap x)) |}{| \Ss |}\\ &\leq \frac{|g(\Stab_G([y] \cap x)) |}{| g(\Stab_{\Stab_G([y] \cap x)}(\Bb_\Gamma) \cap \Stab_{\Stab_G([y] \cap x)}(\Bb)) |}.  
\end{align*}
The last inequality holds because $g(\Stab_{\Stab_G([y] \cap x)}(\Bb_\Gamma) \cap \Stab_{\Stab_G([y] \cap x)}(\Bb)) \leq \Ss$.
It is not difficult to prove that the application of a group homomorphism can only decrease the index of $H$ in $G$, i.e.\ $[h(G) : h(H)] \leq [G : H]$. According to Definition \ref{def:XOR_symmetricBasis}, the index of $\Stab_{\Stab_G([y] \cap x)}(\Bb_\Gamma) \cap \Stab_{\Stab_G([y] \cap x)}(\Bb)$ in $\Stab_G([y] \cap x)$ is polynomially bounded in$ | \GG^S|$. As $g$ is a group homomorphism, this is also a bound for the index of the image under $g$, which is equal to the fraction above.
In total, we have shown:
\[
|\Orb(N_{i-}) | \leq k \cdot \text{poly}(| \GG^S|) \leq \text{poly}(k+| \GG^S|).
\]
Now this orbit bound applies to each of the $k$ rows of $N$, so when closing the rows of $N$ under the action of $g(\Stab_G([y] \cap x))$, we add at most $\text{poly}(k+| \GG^S|)$ many new rows to the matrix. Because $k  = \dim (\Im \ M[y]) - \dim (M[y] \cdot \Stab_E([y] \cap x)) \leq \dim (\Im \ M[y]) \leq |E|$ (this holds because $M[y]$ is a linear map defined on the $|E|$-dimensional space $\bbF_2^E$), $\text{poly}(k+|\GG^S|) = \text{poly}(|\GG^S|)$. The resulting matrix is a candidate for $N[x][y]$, so this shows that $N[x][y]$ has at most $\text{poly}(|\GG^S|)$ many rows.
\end{proof}

\begin{lemma}
	If for all $x \in \tc(\mu)$ and all $[y] \in \Cc[x]$, $\Stab_E([y] \cap x)$ has a symmetric basis, then $|V_C|$ has size polynomial in $|\GG^S|$.
\end{lemma}	
\begin{proof}
	This follows directly from Lemma \ref{lem:XOR_circuitSize} and Lemma \ref{lem:XOR_NmatricesPolynomial} (remember that $|J_{[x][y]}|$ is the number of rows of $N[x][y]$), and from the fact that $|\tc(\mu)|$ is polynomial in $|\GG^S|$ because $\mu$ is CPT-definable in $\GG^S$.
\end{proof}	
This proves \textbf{Property 4} from Theorem \ref{thm:XOR_generalCircuitTheorem}.

\subsection{Which vector spaces have symmetric bases?}
We have shown that the size of $\widehat{C}(\mu)$ can be polynomially bounded if for all $x \in \tc(\mu)$ and all $[y] \in \Cc[x]$, the stabiliser space $\Stab_E([y] \cap x)$ has a \emph{symmetric basis}. If this is not the case, then we do not know anything about the size of $\widehat{C}(\mu)$. There may be other ways to bound it but a priori we have to assume that it is super-polynomial then. This makes these symmetric XOR-circuits less useful for deriving lower bounds against CPT than in the CFI-symmetric case, unless, of course, we know that the object $\mu$ satisfies the symmetric basis property.\\

This leads to the question what the class of h.f.\ sets with the symmetric basis property looks like. We do not have a definitive answer but we can show: All CFI-symmetric objects admit symmetric bases, there are easy examples of objects with symmetric bases which are not CFI-symmetric, and there also exist objects which do not admit symmetric bases but may a priori be CPT-definable. This latter result shows that there is unfortunately little hope to generally prove that \emph{all} CPT-definable objects over CFI-graphs admit symmetric bases.\\

\begin{restatable}{lemma}{restateCFIsymmetricHasSymmetricBasis}
	\label{lem:XOR_CFIsymmetricHasSymmetricBasis}
	Let $\mu \in \HF(\widehat{E})$ be CFI-symmetric. Then $\mu$ satisfies the symmetric basis condition from Definition \ref{def:XOR_symmetricBasis}.
\end{restatable}	
\noindent
\textit{Proof sketch.} Let $x \in \tc(\mu)$ and $[y] \in \Cc[x]$. Let $\Gamma := \Stab_E([y] \cap x) \leq \bbF_2^E$. We have to define two bases $\Bb_\Gamma \subseteq \Bb$ of $\Gamma$ and of $\bbF_2^E$, respectively, such that the group $\Stab_{\Stab_G([y] \cap x)}(\Bb_\Gamma) \cap \Stab_{\Stab_G([y] \cap x)}(\Bb)$ has polynomial index in $\Stab_G([y] \cap x)$. Since $\mu$ is CFI-symmetric, by Definition \ref{def:XOR_CFIsymmetric}, the $\Aut_{\CFI}(\GG)$-orbit of $[y] \cap x$ has size exactly two. One can prove that then, the stabiliser space $\Stab_E([y] \cap x)$ is a direct sum of a subspace on some coordinate set $I \subseteq E$ containing exactly the vectors with even Hamming weight and the full Boolean space on coordinates $E \setminus I$. It is relatively easy to construct a basis for such a space whose orbit has only linear size: For the even subspace of $\bbF_2^I$ one can fix one coordinate $i \in I$ and take the basis $\{ \chi(i)+\chi(e) \mid e \in I \setminus \{i \} \}$, which is symmetric up to the choice of $i$. For the space $\bbF_2^{E \setminus I}$, we can simply take the canonical basis consisting of the unit vectors. Details are in the appendix. \hfill \qed\\

Thus, all CFI-symmetric h.f.\ sets satisfy the symmetric basis property from Definition \ref{def:XOR_symmetricBasis}. But are there any other objects that have symmetric bases? The answer is affirmative. To keep things simple, we do not give a fully specified example but only sketch how a family of \emph{non-CFI-symmetric} h.f.\ sets \emph{with symmetric bases} may look like. In the following, we always write $\widetilde{\bbF}_2^{I}$ for the subspace of $\bbF_2^I$ that consists of all vectors with even Hamming weight.
\begin{example}
	\label{ex:XOR_nonCFIsymmetricSet}
	For $n \in \bbN$, let $E_n$ denote an $n$-element set of base edges. We do not fix a specific family $(G_n = (V_n,E_n))_{n \in \bbN}$ of base graphs. Let $A_n \uplus B_n \uplus \{e\} = E_n$ be an arbitrary partition of the edge set such that one part is a singleton. Consider again Example \ref{ex:XOR_CFIsymmetricSet}. There, we defined the CFI-symmetric object $\mu_{\{e,f,g\}} = \{ \{ \mu_{\{f,g\}}, e_0  \},   \{ \widetilde{\mu}_{\{f,g\}}  , e_1 \} \}$. In this construction,  $\mu_{\{f,g\}}$ and its automorphic image $\widetilde{\mu}_{\{f,g\}}$ are sets that are stabilised by every $\rho \in \Aut_{\CFI}(\GG)$ that flips an even number of edges in $\{f,g\}$. Such objects can be defined more generally for any set of edges. This is done in the CFI-algorithms from \cite{dawar2008, pakusaSchalthoeferSelman}. So let $\mu_{B_n}, \widetilde{\mu}_{B_n}$ denote two sets that together form an $\Aut_{\CFI}(\GG_n)$-orbit and are stabilised by any $\rho \in \Aut_{\CFI}(\GG_n)$ flipping an even number of edges in $B_n$. In other words, $\Stab_E(\mu_{B_n}) = \bbF_2^{A_n \cup \{e\}} \oplus \widetilde{\bbF}_2^{B_n}$.
	Now for every $n \in \bbN$, let
	\[
	\mu_n := \{ \{ \mu_{B_n}, e_0    \} \} \in \HF(\widehat{E}_n).
	\]
	This object is similar to the one from Example \ref{ex:XOR_CFIsymmetricSet}, with the difference that the connected component of $\{ \mu_{B_n}, e_0    \}$ contains just this set itself, and therefore, the $\Aut_{\CFI}(\GG_n)$-orbit of this component has size four instead of two. Thus, $\mu_n$ is \emph{not} CFI-symmetric. However, it does satisfy the symmetric basis property (assuming that $\mu_{B_n}$ does -- which is possible since $\mu_{B_n}$ could e.g.\ be CFI-symmetric). Let $y = \{ \mu_{B_n}, e_0    \}$ and $x = \mu_n$. Then $[y] \cap x = y$, and hence $\Stab_E([y] \cap x) = \Stab_E(y) = \Stab_E(\{ \mu_{B_n}, e_0    \})$. It is not hard to construct a symmetric basis for this space. By the properties of $\mu_{B_n}$, we have 
	\[
	\Stab_E(y) \cong \bbF_2^{A_n} \oplus \widetilde{\bbF}_2^{B_n}.
	\]
	In other words, this space contains every vector that has even Hamming weight on $B_n$ and a zero entry at coordinate $e$. A basis for this can be defined as in the proof of Lemma \ref{lem:XOR_CFIsymmetricHasSymmetricBasis}: Fix some $f \in B_n$. Then include in the basis $\Bb_\Gamma$ every unit vector $\chi(g)$ for $g \in A_n$ and the vector $\chi(\{f,g\})$ for each $g \in B_n \setminus \{f \}$. Let $\Bb := \Bb_\Gamma \cup \{ \chi(e), \chi(f) \}$.
	We have not specified the base graphs exactly, so we have not made any assumptions on $\Aut(G)$.  Suppose now that $A_n$ and $B_n$ lie in different orbits of $\Stab_G([y] \cap x)$. 
	This makes sense because otherwise, $\mu_{B_n}$ would not necessarily be stabilised. Then $\Stab_{\Stab_G([y] \cap x)}(\Bb_\Gamma) \cap \Stab_{\Stab_G([y] \cap x)}(\Bb)$ is the pointwise stabiliser of $\{e,f\}$ in $\Stab_G([y] \cap x)$. This has index $\leq n^2$, which is polynomial.
\end{example}	

Objects with symmetric bases are therefore indeed a strict generalisation of CFI-symmetric objects. Nonetheless, there currently exists no choiceless algorithm for the CFI query that requires the construction of objects which go beyond the CFI-symmetric ones.\\ 

Finally, the most important question is whether there also exist objects that are neither CFI-symmetric nor have symmetric bases. 
Ideally, we would like the answer to be that \emph{every} CPT-definable object $\mu$ satisfies the symmetric basis condition. Then, super-polynomial size lower bounds against suitable circuits would actually separate CPT from \ptime \ because they would rule out \emph{any} CPT-algorithm for the CFI-query, not just special algorithms like the CFI-symmetric ones. We do not know if this ideal situation is in fact reality. However, we have an example that suggests it is not.\\

The line of thought is this: An obvious way to show that every CPT-definable object $\mu$ has the symmetric basis property would be to try and exploit the fact that for every $x \in \tc(\mu)$ and every $[y] \in \Cc[x]$, $\Stab_E([y] \cap x)$ must have a polynomial index in $\Aut_{\CFI}(\GG^S) \leq \bbF_2^E$ (otherwise, the orbit of $[y] \cap x \subseteq \tc(\mu)$ would be super-polynomial, so $\mu$ would not be CPT-definable). This is perhaps the most obvious consequence that follows from the CPT-definability of $\mu$. To simplify things a bit, let us assume that  $\Aut_{\CFI}(\GG^S) = \bbF_2^E$. Then in terms of vector spaces, $[\bbF_2^E : \Stab_E([y] \cap x)]$ being polynomial means that the codimension of $\Stab_E([y] \cap x)$ in $\bbF_2^E$, i.e.\ $|E| - \dim \ \Stab_E([y] \cap x)$, is logarithmic. What we also know by Lemma \ref{lem:XOR_actionOfAutomorphismsOnStabiliserSpaces} is that the space $\Stab_E([y] \cap x)$ is invariant under the action of the permutation group $\Stab_G([y] \cap x)$. This leads to the question if these two restrictions on $\Stab_E([y] \cap x)$ are sufficient to show that $\Stab_E([y] \cap x)$ necessarily has a symmetric basis in the sense of Definition \ref{def:XOR_symmetricBasis}? Unfortunately, the answer is no. There is a family of Boolean vector spaces together with permutation groups on their index sets such that the spaces are invariant under the permutations, have at most logarithmic codimension in the ambient space, and do not admit a symmetric basis. We construct such an example in Lemma \ref{lem:XOR_counterExample} below. From this it does not follow directly that there are actually families of CFI-graphs and CPT-definable h.f.\ sets over them which do not have the symmetric basis property. It just means that we cannot show the symmetric basis property for general CPT-definable sets with arguments that are only based on the obvious properties of vector spaces which can occur as $\Stab_E([y] \cap x)$ in CPT-definable objects.

\begin{restatable}{lemma}{restateCounterExample}
	\label{lem:XOR_counterExample}
	There exists a family of Boolean vector spaces $(\Gamma_n)_{n \in \bbN}$, a function $t(n) \in \Theta(n)$ with $\Gamma_n \leq \bbF_2^{t(n)}$, and a family of permutation groups $(\mathbf{G}_n)_{n \in \bbN}$ with $\mathbf{G}_n \leq \Sym_{t(n)}$ such that
	\begin{enumerate}
		\item $\Gamma_n$ is $\mathbf{G}_n$-invariant.
		\item The codimension of $\Gamma_n$ in $\bbF_2^{t(n)}$ is $\Oo(\log n)$.
		\item For any pair of bases $\Bb_\Gamma \subseteq \Bb$ such that $\Bb_\Gamma$ is a basis of $\Gamma_n$ and $\Bb$ is a basis of $\bbF_2^{t(n)}$, $[\mathbf{G}_n : \Stab_{\mathbf{G}_n}(\Bb)] \geq \Big(\frac{n}{(\log n)^2}\Big)^{\log n}$, which is super-polynomial in $n$.   
	\end{enumerate}	
\end{restatable}	
\noindent
\textit{Proof sketch.} We define $\Gamma_n$ as the direct sum of $\log n$ many even spaces $\widetilde{\bbF}_2^{P_i}$, for $i \in [\log n]$. The index-sets $P_i$ are pairwise disjoint and form a partition of $[n]$. The group $\mathbf{G}_n$ is defined as the setwise stabiliser of the partition $\{P_1,...,P_{\log n}\}$. Thus, the space $\Gamma_n$ is $\mathbf{G}_n$-invariant. The fact that the codimension of $\Gamma_n$ in $\bbF_2^n$ is $\Oo(\log n)$ is not difficult to see. The lower bound on the orbit size of any pair of bases $\Bb_\Gamma \subseteq \Bb$ holds because extending any $\Bb_\Gamma$ to a basis of $\bbF_2^E$ requires to choose at least one point in each part $P_i$, for $i \in [\log n]$. These are $\log n$ choices, each from a set of size $\frac{n}{\log n}$. The orbit size of this tuple of choices is super-polynomial in $n$. The detailed proof can be found in the appendix. \hfill \qed \\

We do not know if there actually exist CFI-graphs $(\GG_n^S)_{n \in \bbN}$ and h.f.\ sets $(\mu_n)_{n \in \bbN}$ over them in which $\Gamma_n \cong \Stab_E([y] \cap x)$ for some $x,y \in \tc(\mu_n)$, and $\mathbf{G}_n \cong \Stab_G([y] \cap x)$. This could a priori be the case. Anyway, we can conclude that the question whether a CPT-definable object $\mu$ over some CFI-instance $\GG^S$ admits symmetric bases for all relevant spaces $\Stab_E([y] \cap x)$ cannot be answered without using further information about $\GG^S$ and $\mu$: It seems that CPT-definability of the objects is not sufficient to infer the existence of the required symmetric bases (or this requires more sophisticated techniques than just using the logarithmic bound on the codimension). It should be remarked, though, that making further progress in this direction seems only useful once we have a strong enough lower bound for these circuits, which would separate the CFI-symmetric algorithms from $\ptime$. Then, as a next step, one could try to see in how far this generalises to all choiceless algorithms.

\section{Application to hypercube CFI-structures}
Our second main result is a lower bound against symmetric XOR-circuits over $n$-dimensional hypercubes. It shows that the circuits corresponding to CFI-symmetric h.f.\ sets over hypercube CFI-structures do not exist if we make the circuit parameters slightly more restrictive than in Theorem \ref{thm:XOR_lowerboundProgram} (so note that we are returning to the \emph{CFI-symmetric} setting now). First, however, we have to check that these hypercube CFI-structures  indeed satisfy the preconditions of Theorem \ref{thm:XOR_lowerboundProgram}.
The theorem mainly depends on three parameters of the base graphs: The \emph{treewidth} of the graph, the \emph{CFI-support gap} of the h.f.\ sets, and the fact that the CFI-graphs over the base graphs are $\Cc^{\tw_n}$-homogeneous.\\

The \emph{$n$-dimensional hypercube} $\H_n$ is the undirected graph with universe $\{0,1\}^n$ in which there is an edge between any two words with Hamming distance exactly one. We let $E_n$ denote this edge relation and $V_n = \{0,1\}^n$ the vertex set of $\H_n$. Its automorphism group is the semi-direct product $\bbF_2^n \rtimes \Sym_n$, where $\Sym_n$ acts on the positions of the binary words and $\bbF_2^n$ acts on $\{0,1\}^n$ via the bitwise XOR-operation  \cite{harary2000automorphism}. We will in the following pretend that $\Aut(\H_n) \cong \Sym_n$, i.e.\ we ignore teh action of $\bbF_2^n$. This makes things easier and besides, if $\Sym_n$-symmetric XOR-circuits with the necessary properties do not exist, then this is ``even more true'' for the larger symmetry group $\bbF_2^n \rtimes \Sym_n$. Precisely, the group action is given by $\pi(v_1...v_n) = v_{\pi^{-1}(1)}v_{\pi^{-1}(2)}...v_{\pi^{-1}(n)}$, for every word $v \in \{0,1\}^n$ and $\pi \in \Sym_n$. In the hypercube, this corresponds to applying the same permutation $\pi$ to the neighbourhood of every vertex; this preserves the graph structure of the hypercube.\\ 

In the following, when we speak about CFI-structures over hypercubes, we do not distinguish between isomorphic ones, so we only consider \emph{the} even and \emph{the} odd CFI-structure over $\H_n$ and denote them $\HH_n^0$ and $\HH_n^1$, respectively. The size $|\HH_n^i|$ of these CFI-structures is polynomial in $2^n = |\H_n|$, because the CFI-construction increases the size of the graph exponentially in the maximum degree. This maximum degree in $\H_n$ is $n$, so the size increase by a factor of $2^n$ is still polynomial in $|\H_n|$. Now let us check the relevant properties of the hypercubes and their CFI-structures.

\paragraph*{Treewidth of hypercubes}
\begin{lemma}[Theorem 5 in \cite{treewidthHypercubes}]
	\label{lem:CFI_treewidthHypercubes}
	The \emph{treewidth} of the $n$-dimensional hypercube $\H_n$ is a function in $\Theta(2^{n}/\sqrt{n})$. 
\end{lemma}	
This is close to being linear in $2^n = |\H_n|$, so it is sufficiently large to translate into a meaningful lower bound on the input sensitivity of the resulting circuits. In particular, it is super-constant, and thus, the hypercube CFI-query is undefinable in fixed-point logic with counting.

\paragraph*{Homogeneity of hypercubes}
Recall that a structure $\AA$ is $\Cc^k$-homogeneous for all tuples of some length $\ell \leq k$ if for all tuples $\bar{a}, \bar{b}$ of length $\leq \ell$ it holds: If $\bar{a}$ and $\bar{b}$ have the same $\Cc^k$-type in $\AA$, then there is an automorphism of $\AA$ that moves $\bar{a}$ to $\bar{b}$. 
The next lemma is a technical ingredient that we need for the homogeneity result for hypercube CFI-structures. It shows the homogeneity condition for specific tuples. The full proof is given in the appendix.\\  
We say that a tuple $\bar{\alpha}$ in $V(\HH_n^i)$ \emph{contains a star} if there is a centre $c \in V(\H_n)$ such that for each incident edge $e \in E_n(c)$, there is an entry of $\alpha$ in the edge gadget $e^*$.
\begin{restatable}{lemma}{restateMainHomogeneity}
	\label{lem:XOR2_MainHomogeneity}
	Let $\tw_n \in \Theta(2^n/\sqrt{n})$ denote the treewidth of $\H_n$. Let $\bar{\alpha}$ be a tuple in $V(\HH_n^i)$ that contains a star and has length at most $(\tw_n/n)-2$. Let $\gamma, \gamma' \in V(\HH_n^i)$ and let $\tp(\bar{\alpha}\gamma)$ denote the $\Cc^{\tw_n}$-type of this extended tuple.
	If $\tp(\bar{\alpha}\gamma) = \tp(\bar{\alpha}\gamma')$, then there is an automorphism $\rho \in \Aut(\HH_n^i)$ such that $\rho(\gamma) = \gamma'$ and $\rho(\bar{\alpha}) = \bar{\alpha}$.	
\end{restatable}
\noindent
\textit{Proof sketch.} The proof is mostly standard and works similarly as the homogeneity proofs in \cite{ggpp} and \cite{svenja}. Firstly, one can show that if $\tp(\bar{\alpha}\gamma) = \tp(\bar{\alpha}\gamma')$, then $\gamma$ and $\gamma'$ must be in the same gadget of the CFI-structure. This holds even though the gadgets are not identifiable by means of a preorder: Using the parameters $\bar{\alpha}$, the gadgets of $\gamma$ and $\gamma'$ can be defined with a constant number of variables. Thus, we know that $\gamma$ and $\gamma'$ are in the same, say, edge gadget $e^*$ (the case where they are in a vertex gadget is analogous). If $\gamma = \gamma'$, then there is nothing to show. Else, we need to find an automorphism in $\Aut_{\CFI}(\HH_n^i)$ that flips $e^*$ and does not move any element of $\bar{\alpha}$. This is possible if there exists a cycle in $\H_n$ passing through $e$ and through none of the edges in $\bar{\alpha}$. If such a cycle does not exist, then the tuple $\bar{\alpha}\gamma$ marks the boundary of some sufficiently small subgraph of $\HH_n^i$, which, importantly, has treewidth $< \tw_n$. Then one can show that $\gamma$ and $\gamma'$ are definable in $\Cc^{\tw_n}$ using the parameters $\bar{\alpha}$. This works by using a transfer of the Cops' winning strategy from the cops and robber game to a winning strategy for Spoiler in the bijective $\tw_n$-pebble game, as in \cite{atseriasBulatovDawar}. Thus, the desired cycle and hence automorphism must exist because otherwise, $\tp(\bar{\alpha}\gamma) \neq \tp(\bar{\alpha}\gamma')$. \hfill \qed \\

\begin{lemma}
	\label{lem:XOR2_HomogeneityOfHypercubes}
	The structures $\HH_n^0$ and $\HH_n^1$ are homogeneous in the following sense:\\
	Let $\tw_n$ denote the treewidth of $\H_n$.
	For any tuples $\bar{\alpha}, \bar{\alpha}'$ in $V(\HH_n^i)$ of length $|\bar{\alpha}| = |\bar{\alpha}'| \leq \tw_n/n-n-1$ it holds: If $\bar{\alpha}$ and $\bar{\alpha}'$ have the same $\Cc^{\tw_n}$-type in $\HH_n^i$, then there is an automorphism of $\HH_n^i$ that maps $\bar{\alpha}$ to $\bar{\alpha}'$.\\
\end{lemma}
\begin{proof}
	First of all, we show the following statement via induction on $|\bar{\alpha}|$.\\
	\textbf{Claim 1:}
	Let $\bar{s}$ be a tuple of length $n$ that contains a star. If $\tp(\bar{s}\bar{\alpha}) = \tp(\bar{s}\bar{\alpha}')$, then there is an automorphism of $\HH_n^i$ that maps the tuple $\bar{s}\bar{\alpha}$ to $\bar{s}\bar{\alpha}'$.\\
	\textit{Proof of claim:} In the base case, $|\bar{\alpha}| = 0$, there is nothing to show because the identity permutation is the desired automorphism then.\\
	For the inductive step, let $\bar{\alpha} = \bar{\beta}\gamma$ where $|\bar{\beta}| = |\bar{\alpha}|-1$, and similarly, write $\bar{\alpha}' = \bar{\beta}'\gamma'$. Since $\tp(\bar{s}\bar{\alpha}) = \tp(\bar{s}\bar{\alpha}')$, we also have $\tp(\bar{s}\bar{\beta}) = \tp(\bar{s}\bar{\beta}')$. Therefore, the induction hypothesis gives us an automorphism $\pi \in \Aut(\HH_n^i)$ such that $\pi(\bar{s}\bar{\beta}) = \bar{s}\bar{\beta}'$. Since automorphisms preserve types, we have $\tp(\bar{s}\bar{\beta}\gamma) = \tp(\pi(\bar{s}\bar{\beta}\gamma)) = \tp(\bar{s}\bar{\beta}'\pi(\gamma)) = \tp(\bar{s}\bar{\beta}'\gamma')$.
	The length of the tuples $\bar{s}\bar{\beta}'\pi(\gamma)$ and $\bar{s}\bar{\beta}'\gamma'$ is at most $\tw_n/n-1$, so we can apply Lemma \ref{lem:XOR2_MainHomogeneity} to them. This gives us another automorphism $\sigma$ such that $\sigma(\bar{s}\bar{\beta}'\pi(\gamma)) =  \bar{s}\bar{\beta}'\gamma'$. In total, $\sigma \circ \pi \in \Aut(\HH_n^i)$ is the desired automorphism that maps $\bar{s}\bar{\alpha}$ to $\bar{s}\bar{\alpha}'$.\\
	
	Now we will use Claim 1 to prove the lemma. Let $\bar{\alpha}$ and $\bar{\alpha}'$ be as in the lemma. Let $\bar{s}$ be a tuple of length $n$ that contains a star (and only a star). By Lemma 34 in \cite{dawar2008}, there exists a tuple $\bar{s}'$ such that $\tp(\bar{s}\bar{\alpha}) = \tp(\bar{s}'\bar{\alpha}')$.
	This holds because $(\HH_n^i,  \bar{\alpha}) \equiv_{\Cc^{\tw_n}} (\HH_n^i,\bar{\alpha}')$ (Theorem \ref{thm:CFI_cfiTheorem}), and so for any extension of the tuple $\bar{\alpha}$ (which has length at most $\tw_n - n$) by only $n$ elements, there is an extension of $\bar{\alpha}'$ that preserves the type.\\
	\\
	\textbf{Claim 2:} There exists an automorphism $\sigma \in \Aut(\HH_n^i)$ such that $\sigma(\bar{s}) = \bar{s}'$.\\
	\textit{Proof of claim:} Let $s_1,...,s_n \in E_n$ be the edges of the star that is covered by $\bar{s}$ and $s'_1,...,s'_n \in E_n$ be the edges of the star of $\bar{s}'$. There is an automorphism $\sigma' \in \Aut(\H_n)$ such that $s'_i = \sigma'(s_i)$, for every $i$. This is easy to see because we can map the centre $c$ of one star to the centre $c'$ of the other, and apply the right permutation to its incident edges. In $\HH_n^i$, the gadgets $c^*$ and $c'^*$ are either both even or both odd, in relation to the tuple $\bar{s}$.
	That is, if we pretend that the vertices $\bar{s}_i$ are the $1$-vertices in their respective edge-gadgets, then the vertex-gadgets $c^*$ and $c'^*$ have the same parity: This is because we can express in, say, $\Cc^{5}$ that every vertex in $c^*$ is connected with an even number of vertices in $\bar{s}$ (and $\bar{s}'$, respectively).\\ 
	We say that there is a mismatch between $\sigma'(\bar{s})$ and $\bar{s}'$ at position $i$ if $\sigma'(\bar{s})_i$ is the $0$-vertex in its edge gadget, and  $\bar{s}'_i$ the $1$-vertex, or vice versa. By what we just argued, the number of mismatches between $\sigma'(\bar{s})$ and $\bar{s}'$ is even. This can be corrected with an automorphism $\rho_F \in \Aut_{\CFI}(\HH_n^i)$ that flips edges along $\ell$ disjoint cycles originating in $c'$, where $\ell$ is half the number of mismatches. Then $\sigma = \rho_F \circ \sigma'$ is an automorphism that takes $\bar{s}$ to $\bar{s}'$. This proves the claim.\\
	\\
	With Claim 2, we get that $\tp(\bar{s}'\sigma(\bar{\alpha})) = \tp(\bar{s}'\bar{\alpha}')$ because automorphisms preserve types. The fact that $\bar{s}$ contains a star is easily definable in counting logic, so $\bar{s}'$ also contains a star. Therefore, the lemma now follows from Claim 1, which gives us an automorphism $\pi$ that maps $\sigma(\bar{\alpha})$ to $\bar{\alpha}'$. Then $\pi \circ \sigma$ is the automorphism whose existence is claimed in the lemma.
\end{proof}

In total, the hypercube CFI-structures satisfy the homogeneity condition required by Theorem \ref{thm:XOR_lowerboundProgram} for all tuples of length bounded by $\tw_n /n - n -1$, which is in $\Theta(\tw_n/n)$. 

\paragraph*{CFI-support gap of hypercube CFI-structures}
~\\
Recall from Definition \ref{def:XOR_supportGap} that the CFI-support gap of a h.f.\ set $\mu$ is $\alpha(\mu) = \frac{s(\mu)}{|\sup_{\CFI}(\mu)|}$ where $s(\mu)$ denotes the size of the smallest support; this depends on the structure. 
Let again $\HH_n^i$ denote the odd/even CFI-structure over the $n$-dimensional hypercube, and let $E_n$ be the edge set of that hypercube. Let $\HH_n$ denote the full CFI-graph over $\H_n$ (see Section \ref{sec:CFI}). We would like to prove an upper bound on the ratio $\alpha(\mu)$ over all h.f.\ sets $\mu \in \HF(\widehat{E}_n)$. Here, $s(\mu)$ denotes the size of the smallest $\Aut(\HH_n^i)$-support of $\mu$, and $\sup_{\CFI}(\mu)$ the size of the smallest CFI-support of $\mu$. Recall from Definition \ref{def:CFI_support} that a CFI-support of $\mu$ is a set of edges $S \subseteq E$ such that any edge-flip $\rho_F \in \Aut_{\CFI}(\HH_n)$ that fixes all edges in $S$ also fixes $\mu$. The edge-flips considered here are not necessarily automorphisms of $\HH_n^i$, so they include all combinations of flipped edges and not only cycles.
\begin{lemma}
	\label{lem:XOR2_fromSupportToCFIsupport}
	Let $n \in \bbN$ and $\mu \in \HF(\widehat{E}_n)$. Let $S \subseteq E$ be a smallest CFI-support of $\mu$. Then there exists an $\Aut(\HH_n^i)$-support of $\mu$ of size at most $|S|+n$.
\end{lemma}	
\begin{proof}
	Every $\rho_F \in \Aut_{\CFI}(\HH_n)$ such that $F \cap S = \emptyset$ fixes $\mu$. This holds in particular for every such $\rho_F \in \Aut_{\CFI}(\HH_n^i) \leq \Aut_{\CFI}(\HH_n)$. Thus, let $\bar{\alpha}$ be an arbitrary tuple in $\widehat{E}_n$ that contains exactly one vertex $e_i$ from every edge $e \in S$. So $|\bar{a}| = |S|$. Then any automorphism of the form $(\rho_F, \text{id}) \in \Aut(\HH_n^i)$ that fixes $\bar{\alpha}$ also fixes $\mu$. Now extend $\bar{\alpha}$ to a tuple $\bar{\beta}$ that contains a star. This is always possible such that $|\bar{\beta}| \leq |\bar{\alpha}|+n$. Now any automorphism in $\Aut(\HH_n^i)$ that fixes $\bar{\beta}$ must have $\text{id}$ as its second component, because any $(\rho_F, \pi) \in \Aut(\HH_n^i)$ with $\pi \neq \text{id}$ moves every star in the hypercube. So in total, every automorphism that fixes $\bar{\beta}$ fixes $\mu$, and the length of $\bar{\beta}$ is at most $|S| + n$.
\end{proof}

Let $\tw_n \in \Theta(2^n/\sqrt{n})$ denote the treewidth of the $n$-dimensional hypercube. Let $\alpha(n) := \max_{\stackrel{\mu \in \HF(\widehat{E}_n),}{s(\mu) \in \Omega(\tw_n/n)}} \alpha(\mu)$ be the maximum CFI-support gap that can occur for any object in $\HF(\widehat{E}_n)$ with minimum support size at least $\Omega(\tw_n/n)$.
\begin{corollary}
	\label{cor:XOR2_hypercubeSupportGapAsymptoticallyOne}
	There is a function $g(n) \in \Oo(1)$ which is an upper bound for $\alpha(n)$.
\end{corollary}	
\begin{proof}
	Let $\mu \in \HF(\widehat{E}_n)$ be an object whose minimum $\Aut(\HH_n^i)$-support size $s(\mu)$ is at least $\Omega(\tw_n/n)$. By Lemma \ref{lem:XOR2_fromSupportToCFIsupport}, its smallest CFI-support $\sup_{\CFI}(\mu)$ must have size at least $s(\mu) - n$. Thus, the CFI-support gap $\frac{s(\mu)}{|\sup_{\CFI}(\mu)|}$ is at most $\frac{s(\mu)}{s(\mu) -n} = 1 + \frac{n}{s(\mu) -n}$. Since $s(\mu) \in \Theta(2^n/n^{1.5})$, this expression is asymptotically equal to $1$.
\end{proof}	
Consequently, in the setting of Theorem \ref{thm:XOR_lowerboundProgram}, we can take a constant function $g(n)$ for the upper bound of the CFI support gap. This is convenient because it means that Theorem \ref{thm:XOR_lowerboundProgram} yields XOR-circuits that are sensitive to as many edges as possible, namely $\Omega(\tw_n/n)$ many. We are aiming to prove that the symmetric circuits given by the theorem cannot exist, so it is good that the support gap does not loosen the constraints on the circuits here.\\

In total, we can summarise our result for hypercube CFI-structures like this:
\begin{theorem}
	\label{thm:XOR_lowerboundProgramInstantiatedHypercubes}
	If there exists a \emph{super-symmetric} and \emph{CFI-symmetric} $\CPT$-program $\Pi$ that decides the CFI-query on the family of all hypercube CFI-structures $(\HH_n^i)_{n \in \bbN}$, then for every $\H_n = (V_n,E_n)$, there exists an XOR-circuit $C_n$ over $\H_n$ that satisfies:
	\begin{enumerate}
		\item The number of gates in $C_n$ is polynomial in $2^n$.
		\item The orbit-size $|\Orb_{\H_n}(C_n)|$ of the circuit is polynomial in $2^n$. 
		\item $C_n$ is \emph{sensitive} to $\Omega(2^n/n^{1.5})$ many edges in $E_n$.
		\item The (unrestricted) \emph{fan-in dimension} of $C_n$ is $\Oo(n)$. 
	\end{enumerate}	
\end{theorem}	
This is simply Theorem \ref{thm:XOR_lowerboundProgram}, instantiated with the hypercube CFI-structures. The bounds come from the fact that $|\HH_n^i| = \text{poly}(2^n)$ and because the homogeneity condition for $\HH_n^i$ holds for all tuples of length at most $\Oo(\tw_n/n) = \Oo(2^n/n^{1.5})$.

\section{Lower bounds for families of symmetric XOR-circuits over hypercubes }
\label{sec:chapterXOR2}
If we could successfully show that circuits with the properties from Theorem \ref{thm:XOR_lowerboundProgramInstantiatedHypercubes} do not exist, then this would imply that no super- and CFI-symmetric choiceless algorithm can solve the hypercube CFI-problem.
Unfortunately, we only manage this to a certain extent. We impose slightly stronger constraints on the circuits and then show that such circuit families over hypercubes indeed cannot exist. Concretely, we strengthen the symmetry condition on the circuits and assume that they are stabilised by \emph{all} automorphisms of the base graphs (i.e.\ $n$-dimensional hypercubes), so their orbit size is one. Theorem \ref{thm:XOR_lowerboundProgramInstantiatedHypercubes} states only that the orbit-size of the circuits has to be polynomial. Moreover, we impose the condition that the (orbit-wise) number of children and parents of every gate has to be logarithmically bounded. This may be related to the logarithmic bound on the fan-in dimension that we get from Theorem \ref{thm:XOR_lowerboundProgramInstantiatedHypercubes}, but it is probably a stronger restriction. For circuits with these properties over the $n$-dimensional hypercubes, we show that they are not sensitive to enough input gates and hence violate Property 3 from Theorem \ref{thm:XOR_lowerboundProgramInstantiatedHypercubes}.\\
Our lower bound is inspired by an ``almost right'' construction of circuits satisfying the properties from Theorem \ref{thm:XOR_lowerboundProgramInstantiatedHypercubes}. The most difficult part about constructing such circuits seems to be the condition that they should have polynomial orbit size with respect to the action of $\Sym_n$ on $\{0,1\}^n$. A first idea would be to use some tree with logarithmic degree whose leafs are labelled with the elements of $\{0,1\}^n$. This would satisfy all properties except (maybe) the orbit-size. Actually, we do not have a proof that tree-like circuits with the required orbit size do not exist, but we suspect that trees are not symmetric enough. Our result from this section also points in that direction, as we will explain later.\\

Instead of tree-like circuits, there is another more or less obvious construction idea, that could be considered the opposite of trees: In order to build a circuit that is guaranteed to be symmetric under the hypercube automorphisms, we can simply use the hypercube itself: Cut the hypercube in the middle, and use one half of the hypercube as the circuit. The output gate will then be, for example, the string $0^n$, and the input gates are labelled with the strings of Hamming weight $(n/2)$, which are located in the ``middle slice'' of the hypercube. This construction is visualised below.
\begin{figure}[H]
	\centering
	\begin{tikzpicture}[dot/.style={draw,circle,minimum size=1.5mm,inner sep=0pt,outer sep=0pt,fill=blue},circ/.style={draw,circle,minimum size=2.5mm,inner sep=0pt, fill=red},
		circY/.style={draw,circle,minimum size=2.5mm,inner sep=0pt, fill=yellow}]
		
		\node[dot,label=below:$1100$] (gate1100) at (-4,0) {};
		\node[dot,label=below:$1010$] (gate1010) at (-2,0) {};
		\node[dot,label=below:$1001$] (gate1001) at (0,0) {};		
		
		\node[dot,label=below:$0110$] (gate0110) at (2,0) {};
		\node[dot,label=below:$0101$] (gate0101) at (4,0) {};
		
		\node[dot,label=below:$0011$] (gate0011) at (6,0) {};

		\node[dot,fill=red,label=left:$1000$] (gate1000) at (-3+1,2) {};
		\node[dot,fill=red,label=left:$0100$] (gate0100) at (0,2) {};		
		\node[dot,fill=red,label=right:$0010$] (gate0010) at (2,2) {};
		\node[dot,fill=red,label=right:$0001$] (gate0001) at (4,2) {};		
		
		\node[dot,fill=red,label=above:$0000$] (gate0000) at (1.1,3.5) {};
		
		\draw[->, thick] (gate0000) edge (gate1000);
		\draw[->, thick] (gate0000) edge (gate0100);
		\draw[->, thick] (gate0000) edge (gate0010);		
		\draw[->, thick] (gate0000) edge (gate0001);
		
		\draw[->, thick] (gate1000) edge (gate1100);
		\draw[->, thick] (gate1000) edge (gate1010);
		\draw[->, thick] (gate1000) edge (gate1001);		
		\draw[->, thick] (gate0100) edge (gate0110);		
		\draw[->, thick] (gate0100) edge (gate0101);
		\draw[->, thick] (gate0100) edge (gate1100);
		\draw[->, thick] (gate0010) edge (gate1010);		
		\draw[->, thick] (gate0010) edge (gate0110);
		\draw[->, thick] (gate0010) edge (gate0011);
		\draw[->, thick] (gate0001) edge (gate1001);		
		\draw[->, thick] (gate0001) edge (gate0101);
		\draw[->, thick] (gate0001) edge (gate0011);		
		
	\end{tikzpicture}
	\caption{The $4$-dimensional hypercube cut in half. Red nodes are XOR-gates, blue nodes input gates.}
\end{figure}
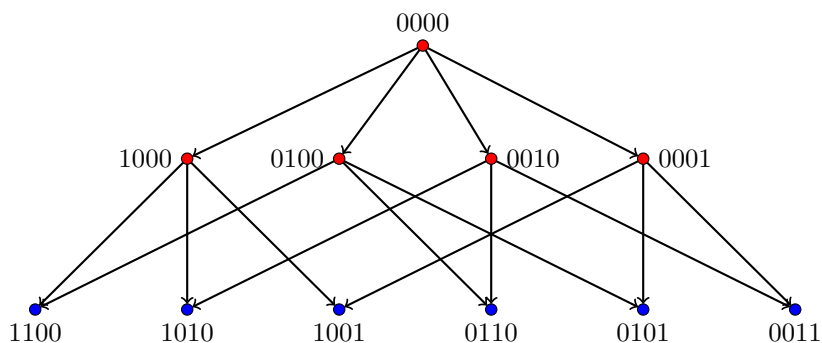
The circuit satisfies the correct size and orbit size bounds (namely, the orbit size of the circuit is one by construction), and also, it has degree $n$, which is logarithmic in $|\HH_n^i|$. Since the degree is an upper-bound for the fan-in dimension, the logarithmic fan-in bound is satisfied as well. However, we can observe that the circuit is actually not sensitive to \emph{any} input bit at all. Already in the second layer, counted from the input layer, all inputs cancel out (in the above example, this is in the root): Each input bit arrives at the root an even number of times. So we learn from this construction that it is possible to build highly symmetric circuits by making them very DAG-like, but the high number of distinct paths that any given input bit can take through the circuit can easily lead to the input bit cancelling itself at some point in the circuit. This happens whenever this number of paths is even. So roughly speaking, tree-like circuits and these ``halved-hypercube circuits'' are two extremes: With trees, the sensitivity condition (Property 3 in Theorem \ref{thm:XOR_lowerboundProgram}) is clearly satisfied, but the symmetry condition is not obvious. Conversely, in the halved hypercubes, the symmetry is satisfied by construction but the circuits are not sensitive to any input bit. The question is: What is in between these two extremes? We will now show that the input-cancellation effect from the halved-hypercube circuits actually occurs in a larger class of symmetric circuits. In some sense, this rules out all circuits that are too similar to the halved hypercube. Roughly speaking, these include all circuits over the $n$-dimensional hypercube that are stabilised by all permutations in $\Sym_n$ and have ``orbit-wise'' logarithmic in- and out-degree.\\

Fix a family $(C_n)_{n \in \bbN}$ of XOR-circuits such that the input gates of $C_n$ are labelled with edges of the $n$-dimensional hypercube $\H_n$. As in the previous section, the circuits are connected DAGs with a designated unique output gate, the root. 
We define 
\begin{align*}
	\Aut(C_n) := &\{  \sigma \in \Sym(V_{C_n}) \mid \sigma \text{ is an automorphism of the rooted DAG } (V_C,E_C,r) \\
	&\text{and there exists a } \pi \in \Sym_n \text{ s.t. } \ell(\sigma(g)) = \pi(\ell(g)) \text{ for every input gate } g  \}. 
\end{align*}
Note that in particular, every automorphism must fix the root of the circuit and must permute the leafs in a way that complies with a permutation in $\Sym_n$ acting on the labels of leafs.
We say that a permutation $\pi \in \Sym_n$ \emph{extends} to an automorphism $\sigma \in \Aut(C_n)$ if $\sigma$ maps the input gates $g$ such that $\ell(\sigma(g)) = \pi(\ell(g))$ is satisfied. It may be that the identity permutation in $\Sym_n$ extends to non-trivial circuit automorphisms in $C_n$. In this case, the circuit is not rigid and every permutation in $\Sym_n$ has multiple circuit automorphisms that it extends to. 
For a gate $g$ in $C_n$ and a parent $h$ of $g$, we let \[
\Orbit_{(g)}(h) := \{ \sigma(h) \mid \sigma \in \Aut(C_n), \sigma(g) = g  \}.
\]
Similarly, 
\[
\Orbit_{(h)}(g) := \{ \sigma(g) \mid \sigma \in \Aut(C_n), \sigma(h) = h  \}.
\]
So these are the orbits of $g$, $h$, respectively, with respect to those circuit automorphisms that fix the child $g$, or the parent $h$, respectively. Note that $\Orbit_{(g)}(h) \subseteq E_C g$, and $\Orbit_{(h)}(g) \subseteq hE_C$ because circuit automorphisms preserve the wires and if one endpoint of a wire is fixed, then the image of the other endpoint must still be connected with the fixed gate. In the rest of this section, we prove:
\begin{theorem}
	\label{thm:XOR2_symmetricCircuitsDoNotCompute}
	Let $(C_n)_{n \in \bbN}$ be a family of XOR-circuits over the $n$-dimensional hypercubes such that for all $n \in \bbN$ it holds:
	\begin{enumerate}
		\item The size $|V_{C_n}|$ is polynomial in $2^n$ (and thus polynomial in $|\HH_n^i|$).
		\item \emph{Every} permutation $\pi \in \Sym_n$ acting on $\{0,1\}^n$ extends to a circuit-automorphism of $C_n$. Thus, the $\Sym_n$-orbit of $C_n$ has size one.
		\item There exists a function $f(n) \in \Oo(n)$ such that for all large enough $n$, for every gate $g$ and every parent $h$ of $g$ in $C_n$, both $|\Orbit_{(g)}(h)|$ and $|\Orbit_{(h)}(g)|$ are at most $f(n)$.
	\end{enumerate}	 
	Then for any constant $\epsilon > 0$, it holds for all large enough $n$: The circuit $C_n$ can only be sensitive to an input gate $g$ if $\ell(g) = \{u,v\}$ is an edge such that the zero-one-split in the binary strings $u,v \in \{0,1\}^n$ is more imbalanced than $\epsilon n$ vs. $(1-\epsilon)n$. In other words, only if the number of $1$s or the number of $0$s in $u$ and $v$ is $< \epsilon n$, then this input gate contributes to the output of $C_n$.
\end{theorem}	
In particular, these circuits do not satisfy Property 3 from Theorem \ref{thm:XOR_lowerboundProgram} and are therefore ruled out:
\begin{corollary}
	\label{cor:XOR2_symmetricCircuitsDoNotSatisfySensitivity}
	Let $(C_n)_{n \in \bbN}$ be a circuit family as in Theorem \ref{thm:XOR2_symmetricCircuitsDoNotCompute}. Then for all large enough $n$, the circuit $C_n$ is sensitive to strictly less than 
	$o(2^n/n^{1.5})$ many inputs. 
\end{corollary}	
\begin{proof}
	According to Exercise 9.42 on page 492 of \cite{concreteMathematics}, it holds for any $\alpha \leq \frac{1}{2}$:
	\[
	\sum_{k \leq \alpha n} \binom{n}{k} = 2^{nH(\alpha)-\frac{1}{2}\log n + \Oo(1)},
	\]
	where $H(\alpha) = \alpha \log(\frac{1}{\alpha}) + (1-\alpha)\log(\frac{1}{1-\alpha})$. According to Theorem \ref{thm:XOR2_symmetricCircuitsDoNotCompute}, the only edges of $\H_n$ that $C_n$ can be sensitive to are between binary strings with less than $\epsilon n$ many one- or zero-entries, for any $\epsilon > 0$. The number of potential endpoints of such edges is twice the above sum, for $\alpha = \epsilon$. The degree of $\H_n$ is $n$, so in total, $C_n$ is sensitive to at most $2n \cdot \sum_{k \leq \epsilon n} \binom{n}{k}$ many edges of $\H_n$. We can calculate that this is in $o(2^n/n^{1.5})$, for any $0 < \epsilon < \frac{1}{2}$:
	\begin{align*}
		&\lim_{n\to\infty} \frac{2^n}{n^{1.5} \cdot 2n \cdot 2^{nH(\epsilon)-\frac{1}{2}\log n + \Oo(1)}} = \\
		&\lim_{n\to\infty} \frac{2^{(1-H(\epsilon))n +\frac{1}{2}\log n - \Oo(1)}}{2n \cdot n^{1.5}} = \infty
	\end{align*}	
	In the last step, we used that $H(\epsilon) < 1$, which holds as long as $\epsilon$ is chosen to be strictly less than $\frac{1}{2}$.
\end{proof}	
This proves Theorem \ref{thm:mainLowerBound}.
This result does not yet completely rule out the existence of a symmetric circuit family as required by 
Theorem \ref{thm:XOR_lowerboundProgramInstantiatedHypercubes}: Firstly, we assume the circuits here to be \emph{fully symmetric}, i.e.\ they are stabilised by \emph{every} permutation in $\Sym_n$; in Theorem \ref{thm:XOR_lowerboundProgramInstantiatedHypercubes}, the circuits need only have a polynomial orbit with respect to the automorphisms of the base graph, so they are stabilised by many, but not necessarily by \emph{all} these automorphisms. Secondly, in Theorem \ref{thm:XOR_lowerboundProgramInstantiatedHypercubes}, we only have a logarithmic bound on the \emph{fan-in dimension}, but it is not clear that this also entails a bound on the orbit-wise number of children and parents of each gate as in Property 3 above. Nevertheless, we hope that this negative result for the existence of fully symmetric bounded-degree XOR-circuits is a useful starting point to rule out further circuit classes over hypercubes, and eventually defeat all circuits from Theorem \ref{thm:XOR_lowerboundProgramInstantiatedHypercubes}.

\subsection{Supporting partitions of permutation groups}
The proof of the theorem relies on group-theoretic techniques, which allow us to approximate every permutation group of index $\text{poly}(2^n)$ by a product of large alternating groups. This idea comes from Anderson and Dawar's paper on symmetric circuits and fixed-point logic \cite{dawarAnderson}. There, it is shown that any group $G \leq \Sym_n$ has a unique \emph{coarsest supporting partition} $\Sp(G)$. This is the coarsest partition of $[n]$ such that every $\pi \in \Sym_n$ which stabilises each part setwise is in $G$. Here, we extend this concept and define \emph{alternating supporting partitions}. The alternating group $\Alt_n \leq \Sym_n$ is the group containing only permutations with even sign. A permutation has even sign if it inverts the order of an even number of pairs in $[n]$.
\begin{definition}[Alternating supporting partition]
	\label{def:G_altSupportingPartition}
	Let $A$ be a set and $\G \leq \Sym(A)$ be a permutation group acting on $A$. An \emph{alternating supporting partition} of $\G$ is a partition $\Pp$ of $A$ such that 
	\[
	\prod_{\stackrel{P \in \Pp}{|P| < 5 }} \Sym(P) \times \prod_{\stackrel{P \in \Pp}{|P| \geq 5 }} \Alt(P) \leq \G.
	\]
\end{definition}
The difference to the ``standard'' supporting partitions is that the odd permutations within the parts of an alternating supporting partition need not be contained in the supported group $\G$. On parts of size $< 5$, we require the full symmetric group to be in $\G$, because otherwise, the proof of the next lemma is problematic, and in our application, constant-size parts will not play a big role anyway.
We have to verify that alternating supporting partitions work just like the original supporting partitions from \cite{dawarAnderson}. Concretely, the desirable properties are that there always exists a unique coarsest supporting partition of a group $G$ and that $G$ is sandwiched between the pointwise and the setwise stabiliser of that partition.
\begin{lemma}
	\label{lem:G_uniqueAlternatingSupportingPartition}
	Each permutation group $\G \leq \Sym(A)$ has a unique \emph{coarsest alternating supporting partition}, denoted $\Sp_A(\G)$.
\end{lemma}
\begin{proof}
	The proof is similar to the one of Lemma 1 in \cite{dawarAnderson}. We need to prove that for any two alternating supporting partitions $\Pp, \Pp'$, the finest partition of which both of them are refinements is still an alternating supporting partition. Then the lemma follows directly. 
	This ``finest common coarsification'' of $\Pp$ and $\Pp'$, denoted $\Ee(\Pp,\Pp')$ is defined like this: Let $\sim \subseteq A^2$ be the transitive closure of the relation ``$a$ and $b$ occur together in some part $P$ of $\Pp$ or $\Pp'$''. The equivalence classes of $\sim$ are the parts of $\Ee(\Pp,\Pp') =: \Ee$. We want to show that any even permutation within any part $Q$ of $\Ee$ (that pointwise fixes everything outside of $Q$) is also in $\G$. We do this by proving: If $P \in \Pp, P' \in \Pp'$ have non-empty intersection, then $\Alt(P \cup P') \leq \G$.  
	Since $\Alt(P \cup P') $ is generated by pairs of transpositions $(x y)(x' y')$, it suffices to show that all such pairs are in $\G$. So consider $(x y)(x' y')$, for $x,y,x',y' \in P \cup P'$ pairwise distinct. We distinguish the following cases: If $x,x',y,y'$ are all in $P$ (or analogously, in $P'$), then $(x y)(x' y') \in \Alt(P) \leq \G$ (using that $\Pp$ is an alternating supporting partition of $\G$).\\ 
	The next case is that $x \in P \setminus P'$, $y \in P' \setminus P$, and $x',y'$ are in the same part, say, both are in $P'$. In this case, let $z \in P \cap P'$, and $a,b \in P \setminus \{ x,z\}$ be two distinct elements (if such $a,b$ do not exist, then $|P| \leq 3$, and so $\Sym(P) \leq \G$, which means that the following argument works even without these $a,b$). It holds $(x y)(x' y') = (x z)(a b)(z y)(x' y')(x z)(a b)$. The number of transpositions within $P$ and within $P'$ is even (we always swap $a$ and $b$ when we swap $x$ and $z$, and we swap $x'$ and $y'$ together with $(z y)$), so this product is in $\Alt(P) \times \Alt(P')$ and therefore in $\G$, because both $\Pp$ and $\Pp'$ are supporting partitions of $\G$.\\
	Another case is that $x,x' \in P \setminus P'$ and $y,y' \in P' \setminus P$. Let again $z \in P \cap P'$, and fix some $a,b \in P \setminus \{ x,x',z \}$ and $a',b' \in P' \setminus \{ y,y',z \}$. Again, if this is not possible, then $P$ and $P'$ are smaller than $5$ and so all permutations on them are in $\G$, which makes the next step only easier.
	Consider $(x z)(a b)(z y)(a' b')(x z)(a b)(x' z)(a b)(z y')(a' b')(x' z)(a b)$. This is equal to $(x y)(x' y')$ and again in $\Alt(P) \times \Alt(P')$ and hence in $\G$.\\
	It remains the case where $x,y \in P$ and $x',y' \in P'$. We can assume that $x,y \in P \setminus P'$ and $x',y' \in P'$ because $\Alt(P) \times \Alt(P') \leq \G$, and so we can move the elements that are to be swapped anywhere within $P$, $P'$, respectively. Then we can simulate the permutation $(x y)(x' y')$ by applying the previous case twice: First, we execute $(x x')(y y')$, and then $(x y')(y x')$. Both are in $\Alt(P) \times \Alt(P')$, as shown above, and hence, in total, we have $(x y)(x' y') \in \Alt(P) \times \Alt(P') \leq \G$.\\ 
	
	So we have shown that we can take the union of two intersecting parts from $\Pp$ and $\Pp'$, and the alternating group on this union will also be in $\G$. Iterating this, we can show that $\Alt(Q) \leq \G$, for every $Q \in \Ee$, because all parts of $\Ee$ can be obtained by iteratively taking the union of intersecting parts of $\Pp$ and $\Pp'$. It remains to show that for all $Q \in \Ee$ with $|Q| < 5$, $\Sym(Q) \leq \G$. But this is clear since such parts $Q \in \Ee$ can only be the union of small parts $P \in \Pp$ and $P' \in \Pp'$. Then we have $\Sym(P) \leq \G$ and $\Sym(P') \leq \G$. We can easily see that every transposition in $\Sym(P \cup P')$ is also in $\G$ by arguing as above, just that we do not need dummy-transpositions anymore in order to keep the sign even.
\end{proof}

\begin{lemma}[variation of Lemma 3 in \cite{dawarAnderson} for alternating supporting partitions]
	\label{lem:G_conjugacyOfAltSupportingPartitions}
	For any $\G \leq \Sym(A)$, and any $\sigma \in \Sym(A), \sigma \Sp_A(\G) = \Sp_A(\sigma G \sigma^{-1})$.
\end{lemma}	
\begin{proof}
	As mentioned in \cite{dawarAnderson}, it holds $\sigma \Ee(\Pp,\Pp') = \Ee(\sigma \Pp, \sigma \Pp')$ for any $\sigma \in \Sym(A)$. Therefore, it remains to show that for any alternating supporting partition $\Pp$ of $\G$, and any $\sigma \in \Sym(A)$, $\sigma \Pp$ is an alternating supporting partition of $\sigma G \sigma^{-1}$. So let
	\[
	\pi \in \prod_{\stackrel{P \in \sigma \Pp}{|P| < 5 }} \Sym(P) \times \prod_{\stackrel{P \in \sigma \Pp}{|P| \geq 5 }} \Alt(P). 
	\]
	Then 
	\[
	(\sigma^{-1} \pi \sigma) \in \prod_{\stackrel{P \in \Pp}{|P| < 5 }} \Sym(P) \times \prod_{\stackrel{P \in \Pp}{|P| \geq 5 }} \Alt(P). 
	\] 
	Therefore, $(\sigma^{-1} \pi \sigma) \in \G$ because $\Pp$ is an alternating supporting partition of $\G$. Consequently, $\pi \in \sigma \G \sigma^{-1}$, and so, $\sigma \Pp$ is an alternating supporting partition of $ \sigma \G \sigma^{-1}$.
\end{proof}

\begin{lemma}[Lemma 4 in \cite{dawarAnderson} for alternating supporting partitions]
	\label{lem:G_altSPsandwichLemma}
	Let $\G \leq \Sym(A)$. Then:
	\[	\prod_{\stackrel{P \in \Sp_A}{|P| < 5 }} \Sym(P) \times \prod_{\stackrel{P \in \Sp_A}{|P| \geq 5 }} \Alt(P) \leq \G \leq \Stab(\Sp_A(\G)).\]
\end{lemma}
\begin{proof}
	The first part is by definition of alternating supporting partitions. For the second part, let $\sigma \in G$. Then $\sigma G \sigma^{-1} = G$. So $\sigma \Sp_A(G) = \Sp_A(G)$ by the preceding lemma.
\end{proof}	
Thus, every group has a unique coarsest alternating supporting partition and is sandwiched between its pointwise and setwise stabiliser. Now the reason why we introduce these \emph{alternating} supporting partitions is because they have a useful property: For groups of index $\leq \text{poly}(2^n)$, the coarsest alternating supporting partition has at most a sublinear number of singleton parts. The proof of this hinges on the following lemma, which we only prove in the appendix because this is quite lengthy and requires a few more prerequisites. In short, the lemma works similarly as Theorem 5.2 B in \cite{dixonMortimer}. The difference is that here, the index is upper-bounded by $2^{nk}$, whereas in \cite{dixonMortimer}, the bound is much smaller, namely only $n^k$. In the following, we use the notation from \cite{dixonMortimer}, so for any group $H \leq \Sym_n$, and any subset $\Delta \subseteq [n]$, $H^{(\Delta)}$ denotes the subgroup of $H$ that fixes $\Delta$ pointwise.
\begin{restatable}{lemma}{restateContainmentOfAlternatingGroup}
	\label{lem:G_containmentOfAlternatingGroup}
	Let $(G_n)_{n \in \bbN}$ be a family of groups such that for all $n$, $G_n \leq \Sym_n$. Assume that there exists some constant $k \in \bbN$ such that asymptotically, $[\Sym_n : G_n] \leq 2^{nk}$.\\
	\\
	Then there exists a constant $0 < c \leq 1$ such that for all large enough $n$, $G_n$ has a subgroup $H_n$ such that $\Alt(A) \leq H_n^{([n] \setminus A)}$ for some $H_n$-orbit $A$ of size $|A| \geq cn$.
\end{restatable}

With this lemma, one can show:
\begin{restatable}{theorem}{restateAlternatingPartitionWithLogManyParts}
	\label{thm:G_alternatingPartitionWithLogManyParts}
	Let $k \in \bbN$ be a constant and $(G_n)_{n \in \bbN}$ be a family of groups such that $G_n \leq \Sym_n$ and for all large enough $n$, $[\Sym_n : G_n] \leq 2^{nk}$. Then the number of singleton parts in $\Sp_A(G_n)$ grows at most sublinearly. In other words: There is no constant $0 < c \leq n$ such that for all large enough $n$, there exists a $\Delta_n \subseteq [n]$ of size $|\Delta_n| \geq cn$ on which $\Sp_A(G_n)$ contains only singleton parts.
\end{restatable}	
\noindent
\textit{Proof sketch.} We let $\Delta_n \subseteq [n]$ be the set of points which are in singleton parts of $\Sp_A(G_n)$. For a contradiction, we assume that $|\Delta_n|$ grows linearly. Then we consider the action of $G_n$ on $\Delta_n$, while every point outside of $\Delta_n$ is fixed. Denote this subgroup of $G_n$ as $H_n$. Then
 Lemma \ref{lem:G_containmentOfAlternatingGroup} applied to $H_n$ yields a contradiction. One can show that $[\Sym(\Delta_n) : H_n] \leq 2^{nk}$, just like it is assumed for the index of $G_n$ in $\Sym_n$. Hence, Lemma \ref{lem:G_containmentOfAlternatingGroup} entails that $H_n$ contains an alternating group on a linearly-sized orbit $A$. However, then $A$ would be a single part in $\Sp_A(G_n)$, which contradicts the fact that the permutation domain of $H_n$ is the set of positions in singleton parts. The full proof is in the appendix. \hfill \qedsymbol \\ 

\subsection{The cancellation of input bits in highly symmetric circuits}
We now return to the proof of Theorem \ref{thm:XOR2_symmetricCircuitsDoNotCompute}. The alternating supporting partitions are used to approximate the stabiliser groups of the gates.
For a gate $g \in V_{C_n}$, we denote by $\Sp_A(g)$ the coarsest alternating supporting partition of the group 
\[
\Stab_n(g) := \{  \pi \in \Sym_n \mid \pi \text{ extends to an automorphism of } C_n \text{ that fixes } g \}.
\]
Here, we mean that at least one of the automorphisms that $\pi$ extends to fixes $g$. It can be seen that $\Stab_n(g)$ is indeed a subgroup of $\Sym_n$ because it contains the identity permutation, and: If $\pi, \pi' \in \Stab_n(g)$, then there exist circuit automorphisms $\sigma, \sigma'$ that $\pi, \pi'$ extend to such that $\sigma(g) = \sigma'(g) = g$. Thus, $\sigma \circ \sigma'$ fixes $g$ and is a circuit automorphism that $\pi \circ \pi'$ extends to. Thus, $\Stab_n(g) \leq \Sym_n$.\\

Importantly, $\Stab_n(g) \leq \Stab(\Sp_A(g))$ (Lemma \ref{lem:G_altSPsandwichLemma}), so every permutation in $\Stab_n(g)$ acts as a permutation on the parts of $\Sp_A(g)$. The supporting partition of an automorphic image of a gate can be obtained by applying a corresponding permutation in $\Sym_n$ to the supporting partition:
\begin{lemma}
	\label{lem:XOR2_supPartsGetShifted}
	Let $g, \sigma g$ be two gates in $C_n$, for a $\sigma \in \Aut(C_n)$. Let $\pi \in \Sym_n$ be a permutation that extends to the circuit automorphism $\sigma$. Then
	\[
	\Sp_A(\sigma g) = \pi (\Sp_A(g)).
	\]
\end{lemma}
\begin{proof}
	By Lemma \ref{lem:G_conjugacyOfAltSupportingPartitions}, the partition $\pi \Sp_A(g)$ is the coarsest alternating supporting partition of the group $\pi \Stab_n(g) \pi^{-1}$. It holds $\pi \Stab_n(g) \pi^{-1} = \Stab_n(\sigma g)$.
\end{proof}

Next, we would like to formalise what it means that a given input gate $g$ cancels itself out in the circuit:
\begin{lemma}
	\label{lem:XOR2_evenPaths}
	Let $g$ be an input gate of an XOR-circuit $C$. The circuit $C$ is sensitive to the input gate $g$ if and only if the number of distinct paths from the root to $g$ is odd.
\end{lemma}
\begin{proof}
	Via induction on the number of gates in $C$. In a circuit where $g$ is the root, there is only one path and the input determines the output. The smallest possible case where the number of paths from the root to $g$ is even is if $C$ consists of a root with two children $h_1,h_2$, and one input gate $g$ that is the child of both $h_1$ and $h_2$. Clearly, the input bit is canceled in the root.\\
	For the inductive step, let $h_1,...,h_m$ be the children of the root $r$. Let $p_i$ be the number of distinct paths from $h_i$ to the input $g$. By the induction hypothesis, the output of $h_i$ depends on $g$ iff $p_i$ is odd (use the statement for the smaller subcircuit rooted at $h_i$). The number of distinct paths from $r$ to $g$ is $\sum p_i$. That number is even iff an even number of the $p_i$ is odd. Then the input $g$ has no influence on the value computed at $r$ because only the $h_i$ with $p_i$ odd are sensitive to $g$, and these effects cancel at $r$ because it is an even number. If an odd number of the $p_i$ is odd, then $r$ is sensitive to $g$.
\end{proof}

Thus, our goal is to prove that the number of paths between the root and each input gate labelled with a ``too balanced edge'' of $\H_n$ is even.
Now the technical theorem that we want to prove in the next step reads as follows. From it, Theorem \ref{thm:XOR2_symmetricCircuitsDoNotCompute} follows with Lemma \ref{lem:XOR2_evenPaths}.

\begin{theorem}
	\label{thm:XOR2_evenPaths}
	Let $(C_n)_{n \in \bbN}$ be a family of XOR-circuits with the properties mentioned in Theorem \ref{thm:XOR2_symmetricCircuitsDoNotCompute}.
	Let a gate $g_n$ in every $C_n$ and a constant $\epsilon > 0$ be fixed such that $\Sp_A(g_n)$ contains at least two parts of size $\geq \epsilon \cdot n$.\\ 
	Then for all large enough $n$, the number of distinct paths from the root of $C_n$ to $g_n$ is even.
\end{theorem}

The proof idea is vaguely similar to a technique known as ``bottleneck counting'', that has been used in proof complexity to establish lower bounds for resolution. 
Roughly speaking, we associate with every gate in a circuit a certain quantity of which we know that it must be high in the root and much lower in the input gates. Furthermore, we will prove that this quantity can only change by a small amount as we move from a gate to its parents. In other words: The quantity cannot ``jump'' from the low value at the leafs to the high value at the root, but it has to pass through many intermediate values in the middle of the circuit. We will then show that certain intermediate values, which must necessarily occur at some gates, entail that the number of paths from the root to the gate is even.\\ 

First, here is an observation about supporting partitions (or partitions in general). In most cases, their orbit has at least quadratic size, unless the partition has a very particular shape.
\begin{lemma}
	\label{lem:XOR2_quadraticOrbits}
	Let $\Pp$ be some partition of $[n]$, for some $n \geq 4$, and let $\Stab_n(\Pp) \leq \Alt_n$ denote the setwise stabiliser of the partition in the alternating group. Then the orbit size of $\Pp$, that is, $(1/2)n!/|\Stab_n(\Pp)|$, is at least $\Omega(n^2)$ unless $\Pp$ has one of the following forms:
	\begin{itemize}
		\item $\Pp = \{ [n] \}$.
		\item $\Pp = \{ \{s\}, [n] \setminus \{s\}  \}$, for some $s \in [n]$.
		\item $\Pp = \{ \{s\} \mid s \in [n]  \}$.
	\end{itemize}
\end{lemma}
\textit{Proof sketch.} It is easy to see that in each of the three cases above, the $\Alt_n$-orbit of $\Pp$ has size one or $n$ (since we are assuming $n$ to be large enough such that $\Alt_n$ acts transitively on $[n]$). It remains to show that the orbit size is at least quadratic if $\Pp$ has any other form. In that case, $\Pp$ must contain some part $P$ of size $|P| \geq 2$, whose complement in $[n]$ is also of size $\geq 2$. Since $\Alt_n$ is transitive on the subsets of $[n]$ (for each fixed subset-size), the part $P$ has $\binom{n}{|P|}$ many $\Alt_n$-images, which is in $\Omega(n^2)$. If $|P| > n/2$, then we are done because any permutation that does not map $P$ to itself is not in $\Stab_n(\Pp)$ then, and so the $\Alt_n$-orbit of $\Pp_n$ is as large as claimed. Otherwise, if $|P| \leq n/2$, then $|\Stab_n(\Pp)| \leq (1/2) \cdot k! \cdot (|P|!)^k \cdot (n-k|P|)!$, where $k$ is the number of parts of size $|P|$ in $\Pp$. One can show that this is maximised for $k=n/|P|$ or $k=1$. Now the rest is just a calculation of a suitable lower bound on the orbit size using the Orbit-Stabiliser theorem. \hfill \qed\\

In combination with our assumption that orbits of parents and children have size $\Oo(n)$, this lemma will help us to get a handle on the interplay of the supporting partitions of parent and child gates. 
We now define the quantity that we associate with each gate, as described above. This quantity is actually rather a vector, that we call the \emph{size profile} (of the supporting partition). It is invariant under symmetries, so we define this measure not for individual gates but for their entire orbits. For a gate $g \in V_{C_n}$, we denote by $[g]$ its $\Aut(C_n)$-orbit in $V_{C_n}$.
A \emph{size profile} is a mapping $\zeta : \bbN \lra \bbN$. For an orbit $[g]$, we define $\zeta[g] : \bbN \lra \bbN$ as follows:
\[
\zeta[g](i) := |\{ P \in \Sp_A(g) \mid |P| = i \}|.
\]
Due to Lemma \ref{lem:XOR2_supPartsGetShifted}, this definition is indeed independent of the choice of the representative $g$ of the orbit.
Note that, as we promised earlier, the measure $\zeta$ differs considerably between the root and the input gates of a circuit. For the root $r$, $\zeta[r](n) = 1$, and $\zeta[r](i) = 0$ for all $i < n$. This is because our circuits are invariant under all permutations in $\Sym_n$ by assumption, so the root is stabilised by all permutations, and hence its coarsest alternating supporting partition contains just one large part. For the input gates, by contrast, we know that their supporting partition always has two parts: If an input gate is labelled with a hypercube edge, say, $\{0^a1^{n-a}, 0^{a-1}1^{n-a+1}\}$, then its supporting partition is $\{ \{ 1,...,a-1 \}, \{a \}, \{a+1,...,n \} \}$. Unless either $a$ or $b$ are very small, this partition has two large parts. We will prove that, with each layer in the circuit, the size profile of the gates cannot change very much, so for example, the largest part will grow by one, and another part will shrink by one, as we go one layer up in the circuit. As a consequence, in the middle of the circuit, we must encounter several different part sizes in the supporting partitions until we can reach the supporting partition $\{ [n]\}$ at the root. In particular, we will encounter even part sizes in some gates, and when that happens, this more or less leads to an even number of paths.
The next lemma is the key in our proof. It tells us precisely how the size-profiles of the gates can differ between children and parents. Essentially, the size of large parts can only change by at most one.
\begin{lemma}
	\label{lem:XOR2_mainPartSizeLemma}
	Let $0 < \epsilon < 1$ be any constant. Let $g$ a gate in $C_n$, $h$ a parent of $g$. Assume that $|\Orbit_{(g)}(h)| \in \Oo(n)$ and $|\Orbit_{(h)}(g)| \in \Oo(n)$.
	Let $\Delta : \{ m \in \bbN \mid m \geq \epsilon \cdot n \} \lra \bbN$ be the function defined as $\Delta(s) := \zeta[h](s) - \zeta[g](s)$. For all large enough $n \in \bbN$ and every $s \geq \epsilon \cdot n$ it holds:
	\begin{itemize}
	\item $|\Delta(s)| \leq 2$.
	\item If $\Delta(s) = 2$, then $\Delta(s+1) = -1, \Delta(s-1) = -1$. If additionally $s$ is odd, then $|\Orbit_{(g)}(h)|$ is even.
	\item If $\Delta(s) = -2$, then $\Delta(s+1) = 1, \Delta(s-1) = 1$. If additionally $s$ is even, then $|\Orbit_{(g)}(h)|$ is even.
	\item If no value of $\Delta$ is $2$ or $-2$, then one of the following is possible:
	\begin{itemize}
		\item If $\Delta(s) = 1$, then $\Delta(s-1) = -1$ or $\Delta(s+1) = -1$. In case that $s$ is odd and $\Delta(s+1) = -1$, then $|\Orbit_{(g)}(h)|$ is even.
		\item If $\Delta(s) = -1$, then $\Delta(s-1) = 1$ or $\Delta(s+1) = 1$. In case that $s$ is even and $\Delta(s-1) = 1$, then $|\Orbit_{(g)}(h)|$ is even. 
	\end{itemize}
	\item If $\Delta(s) \neq 0$, then for all other $s'$ except the ones mentioned in the cases above, it holds $\Delta(s') = 0$.
\end{itemize}
\end{lemma}
\begin{proof}
	Let $\Sp_A^*(g) \subseteq \Sp_A(g)$ be the set of parts of size $\geq \epsilon \cdot n$ and $\Sp_A^*(h) \subseteq \Sp(h)_A$ the parts of size $\geq \epsilon \cdot n -1$ in $\Sp_A(h)$.\\
	\textbf{Claim:} There is a bijection $\gamma : \Sp_A^*(g) \lra \Sp_A^*(h)$ such that:
	\begin{enumerate}[(a)]
		\item For every $Q \in \Sp_A^*(g)$, it holds $|\gamma(Q) \cap Q| \geq |Q|-1$.  
		\item For every $Q \in \Sp_A^*(g)$, it holds $|\gamma(Q) \setminus Q| \leq 1$. 
		\item There is at most one part $Q \in \Sp_A^*(g)$ such that $|\gamma(Q) \cap Q| = |Q|-1$.
		\item There is at most one part $Q \in \Sp_A^*(g)$ such that $|\gamma(Q) \setminus Q| = 1$. 
	\end{enumerate}
	\begin{figure}[H]
		\centering
		\begin{tikzpicture}[line width = 0.5mm]
			\draw[fill=green!20] (0,0) rectangle (2.8,-0.75) node[pos=.5] {$\gamma(Q)$};
			\draw[fill=green!20] (-0.2,-1) rectangle (2.6,-1.75) node[pos=.5] {$Q$};	
			
			\draw (0,0) -- (-1,0);
			\draw (0,-0.75) -- (-1,-0.75);
			
			\draw (0,-1) -- (-1,-1);
			\draw (0,-1.75) -- (-1,-1.75);
			
			\draw (2.8,0) -- (3.8,0);
			\draw (2.8,-0.75) -- (3.8,-0.75);
			
			\draw (2.6,-1) -- (3.8,-1);
			\draw (2.6,-1.75) -- (3.8,-1.75);	
		\end{tikzpicture}
		\caption{This is the most extreme way how $Q$ and $\gamma(Q)$ may differ: Each has at most one element that is not shared with the other part.}
	\end{figure}
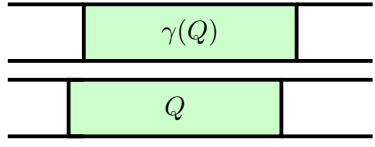
	\textit{Proof of claim:} Construct $\gamma$ by defining $\gamma(Q) \in \Sp_A^*(h)$ as the part whose intersection with $Q$ is largest possible. This is well-defined because there are only two cases how $\Sp_A(h)|_Q$ can look like: It either consists of one part or we have $\Sp_A(h)|_Q = \{\{s\}, Q \setminus \{s\} \}$, for some $s \in Q$. 
	Everything else is ruled out by Lemma \ref{lem:XOR2_quadraticOrbits} and Theorem \ref{thm:G_alternatingPartitionWithLogManyParts}. This is because by assumption, only $\Oo(n)$ many parents of $g$ are in $\Orbit_{(g)}(h)$. Therefore, the restriction of $\Sp_A(h)$ to $Q$, denoted $\Sp_A(h)|_Q$, can have at most $\Oo(n)$ many images under $\Alt(Q)$: Namely, every permutation in $\Alt(Q)$ extends to a circuit automorphism that fixes $g$ because $Q$ is a part in $\Sp_A(g)$. Moreover, by Lemma \ref{lem:XOR2_supPartsGetShifted}, any two distinct $\Alt(Q)$-images of $\Sp_A(h)|_Q$ must be the $Q$-restrictions of supporting partitions of distinct parents of $g$. So indeed, $\Orbit_{(g)}(h)$ contains at least as many gates as the size of the $\Alt(Q)$-orbit of $\Sp_A(h)|_Q$. This size can only be in $\Oo(n)$ if $\Sp_A(h)|_Q$ consists of singletons only or has just one big part or if it as a one-vs-rest split (because of Lemma \ref{lem:XOR2_quadraticOrbits}, where we also use that $|Q| \in \Theta(n)$). The case that $\Sp_A(h)|_Q$ consists only of singletons cannot happen because $Q$ has linear size and by Theorem \ref{thm:G_alternatingPartitionWithLogManyParts}, $\Sp_A(h)$ has at most $o(n)$ many singletons. Therefore, $\Sp_A(h)|_Q$ indeed either consists of one part or we have $\Sp_A(h)|_Q = \{\{s\}, Q \setminus \{s\} \}$, for some $s \in Q$. Thus, $\gamma(Q)$ is well-defined (and indeed, every part $\gamma(Q) \in \Sp_A(h)$ has size $\geq \epsilon \cdot n -1$).\\
	
	The above reasoning also directly proves statement (a). Statement (b) follows in a similar way because if it were not true, then $\Sp_A(g)|_{\gamma(Q)}$ would have $\Omega(n^2)$ many images under $\Alt(\gamma(Q))$ (again by combining Theorem \ref{thm:G_alternatingPartitionWithLogManyParts} and Lemma \ref{lem:XOR2_quadraticOrbits}), resulting in $\Orbit_{(h)}(g)$ being too large.\\
	
	We show that $\gamma$ is injective: If it were not, then there would be some $P \in \Sp_A^*(h)$ and $Q_1, Q_2 \in \Sp_A^*(g)$ such that $|Q_1 \cap P| \geq |Q_1|-1$ and $|Q_2 \cap P| \geq |Q_2|-1$. This is impossible because then, by Lemma \ref{lem:XOR2_quadraticOrbits}, $\Sp_A(g)|_P$ has $\Omega(n^2)$ many automorphic images under $\Alt(P)$, but $\Alt(P)$ fixes $h$, so $h$ has more than $\Oo(n)$ many children in $\Orbit_{(h)}(g)$, which is a contradiction.\\
	Also, $\gamma$ is surjective: Suppose there were a part $P \in \Sp_A(h)$ of size $\geq g(n) - 1$ that has no preimage. Then $\Sp_A(g)|_{P}$ must consist of parts smaller than $\epsilon \cdot n$.
	Then again, $\Sp_A(g)|_{P}$ has $\Omega(n^2)$ many images under $\Alt(P)$ by Lemma \ref{lem:XOR2_quadraticOrbits} (using also that the number of singleton parts in $\Sp_A(g)$ is sublinear and hence less than $|P|$ by Theorem \ref{thm:G_alternatingPartitionWithLogManyParts}).\\
	
	If statement (c) were not true, then there would be two parts $Q_1, Q_2 \in \Sp_A^*(g)$ such that $|\gamma(Q_i) \cap Q_i| = |Q_i|-1$. Then $\Alt(Q_1) \times \Alt(Q_2)$ fixes $g$ but generates $\Omega(n^2)$ many distinct automorphic images of $\Sp_A(h)$. Then again, $|\Orbit_{(g)}(h)|$ is greater than $\Oo(n)$, which contradicts the assumptions of the lemma.\\
	Similarly, statement (d) is shown: If it were not true, then $h$ would have too many children in $\Orbit_{(h)}(g)$. This proves the claim.\\
	\\
	Now with the claim we see that there are five possible cases:
	\begin{enumerate}
		\item $\gamma(Q) = Q$ for all $Q \in \Sp_A^*(g)$.
		\item There is one part $Q \in \Sp_A^*(g)$ such that $|\gamma(Q) \cap Q| = |Q|-1$ and $\gamma(Q) \subseteq Q$, and for all other parts $Q'$, $\gamma(Q') = Q'$.
		\item There is one part $Q$ with $|\gamma(Q) \setminus Q| = 1$ and $\gamma(Q) \supseteq Q$, and for all other parts $Q'$, $\gamma(Q') = Q'$.
		\item There is one part $Q$ that satisfies $|\gamma(Q) \cap Q| = |Q|-1$ \emph{and} $|\gamma(Q) \setminus Q| = 1$. For all other parts $Q'$, $\gamma(Q') = Q'$. 
		\item There is one part $Q_1$ that satisfies $|\gamma(Q_1) \cap Q_1| = |Q_1|-1$ (and $\gamma(Q_1) \subseteq Q_1$), and another part $Q_2$ that satisfies $|\gamma(Q_2) \setminus Q_2| = 1$ (and $Q_2 \subseteq \gamma(Q_2)$), and for all other parts $Q'$, $\gamma(Q') = Q'$.
	\end{enumerate}
	In Case 1, the $\Delta$-vector is zero.\\
	In Case 2, we have $\Delta(|Q|) = -1$ and $\Delta(|Q|-1) = 1$, and all other entries of $\Delta$ are zero.\\
	In Case 3, we have $\Delta(|Q|) = -1$ and $\Delta(|Q|+1) = 1$, and all other entries of $\Delta$ are zero.\\
	In Case 4, the $\Delta$-vector is zero.\\
	In Case 5, we have to distinguish several cases. If $|Q_1| = |Q_2|$, then $\Delta(|Q_1|) = -2$ and $\Delta(|Q_1|-1) = 1, \Delta(|Q_1|+1) = 1$. If $|Q_1| \neq |Q_2|$ and $|Q_1|-1 \neq |Q_2|+1$ , then $\Delta(|Q_1|) = -1, \Delta(|Q_2|) = -1, \Delta(|Q_1|-1) = 1, \Delta(|Q_2|+1) = 1$. If $|Q_1| = |Q_2|+2$, then $\Delta(|Q_1|) = -1, \Delta(|Q_2|) = -1, \Delta(|Q_1|-1) = 2$.\\
	\\
	In Case 2, assume that $|Q|$ is even. Then $\Sp_A(h)|_{Q}$ has an even number of images under $\Alt(Q)$ (because $\Sp_A(h)|_{Q}$ has one singleton part and the rest, and this singleton can be mapped to all $|Q|$ positions by $\Alt(Q)$). We now want to argue that therefore, $|\Orbit_{(g)}(h)|$ must be even. Let $H(Q) \subseteq \Orbit_{(g)}(h)$ be the set of parents $h'$ such that $\Sp_A(h')|_{Q}$ consists of one singleton part and the rest. It holds that $|H(Q)|$ is even:
	Every $\pi \in \Alt(Q)$ extends to a $\sigma \in \Aut(C_n)$ that fixes $g$. By Lemma \ref{lem:XOR2_supPartsGetShifted}, this $\sigma$ maps the parent $h$ of $g$ to another parent of $g$ with $\Sp_A(\sigma h)|_Q = \pi(\Sp_A(h))|_Q$. So the $\Alt(Q)$-orbit of every element of $H(Q)$ is even; hence, $H$ can be partitioned into $\Alt(Q)$-orbits, each of which is even, and so $|H(Q)|$ is even.\\
	
	Now if $H(Q)$ is equal to the whole set $\Orbit_{(g)}(h)$, then we are done. 
	Otherwise, $\Orbit_{(g)}(h)$ contains gates whose supporting partition on $Q$ does not split into a singleton and the rest. Let $h' \in \Orbit_{(g)}(h)$ be such a gate. There must exist a permutation $\pi \in \Sym_n$ that extends to a circuit automorphism $\sigma$ which maps $h$ to $h'$ and fixes $g$. So the corresponding $\pi$ must stabilise the partition $\Sp_A(g)$ setwise, and it will map $\Sp_A(h)$ to $\Sp_A(h')$. Therefore, $\pi(Q) \in \Sp_A(g)$ is a part for which we will again have Case 2 when we apply the above reasoning to $g$ and $h' = \sigma(h)$. Then we can define $H(\pi Q) \subseteq \Orbit_{(g)}(h)$ as the set of all parents whose supporting partition splits into singleton and rest on $\pi Q$, and we get that $|H(\pi Q)|$ is even. In total, with this reasoning we see that $\Orbit_{(g)}(h)$ is partitioned into even-size sets $H(\pi Q)$, for all $\pi \in \Sym_n$ which extend to circuit automorphisms that fix $g$ and permute its parents. So in total, $|\Orbit_{(g)}(h)|$ is even.\\
	
	Similarly, assume in Case 5 that $|Q_1|$ is even. Then the same argument shows that $|\Orbit_{(g)}(h)|$ is even.
	The lemma follows directly from these considerations. 
\end{proof}

\begin{corollary}
	\label{cor:XOR2_numberOfLargePartsStaysHigh}
	Let $0 < \epsilon < 1$ be any constant. Let $g$ be a gate in $C_n$ (for large enough $n$), $h$ a parent of $g$. Assume that $|\Orbit_{(g)}(h)| \in \Oo(n)$ and $|\Orbit_{(h)}(g)| \in \Oo(n)$, and that $|\Orbit_{(g)}(h)|$ is odd.\\
	Let $s$ with $\epsilon \cdot n \leq s < n$ be an even natural number.
	Then
	\[
	\sum_{i \geq s} \zeta[h](i) \geq \sum_{i \geq s} \zeta[g](i)  
	\]
\end{corollary}
\begin{proof}
	According to Lemma \ref{lem:XOR2_mainPartSizeLemma}, for any $i > s$, whenever $\zeta[h](i) < \zeta[g](i)$, then this is compensated by other values of $\zeta[h](j)$ in the sum $\sum_{i \geq s} \zeta[h](i)$.\\
It only remains to consider the case $\zeta[h](s) < \zeta[g](s)$. Since we are assuming that $|\Orbit_{(g)}(h)|$ is odd, and $s$ is even, Lemma \ref{lem:XOR2_mainPartSizeLemma} implies that $\zeta[h](s) = \zeta[g](s)-1$, and $\zeta[h](s+1) = \zeta[g](s+1)+1$. Therefore, the sum $\sum_{i \geq s} \zeta[h](i)$ cannot be strictly less than $\sum_{i \geq s} \zeta[g](i)$.    
\end{proof}

Intuitively speaking, this means that if along some path from the root to a gate $g$, the orbit size of the next parent gate in the stabiliser group of its child is always odd, then the number of large parts in the supporting partitions can only increase along the path towards the root. This will allow us to show that an even orbit must occur along each path. And this means that the path together with its orbit cancels itself out in the XOR computation.\\

When we look at a path $P = (r=h_1,h_2,...,g)$ from the root $r$ of a circuit to a certain gate $g$, then we can associate with $P$ its \emph{orbit-profile} $\Omega(P)$. This orbit profile says for every gate $h_i$ on the path, which orbit its predecessor $h_{i-1}$ belongs to. By orbit, we mean again $\Orbit_{(h_i)}(h_{i-1})$, so we refer to the partition of the parents of $h_i$ into the orbits with respect to the subgroup of $\Sym_n$ that fixes $h_i$. The orbit profile of a path is not supposed to describe that path uniquely but we rather want that several paths share the same orbit profile -- in a sense, we want the orbit profile to describe the ``path'' that we get when we factor out the respective orbits $\Orbit_{(h_i)}(h_{i-1})$. We have to show that this indeed makes sense:
\begin{lemma}
	\label{lem:XOR2_orbitsOfOrbits}
	Let $g$ be a gate and $h$ a parent of $g$. Let $g'$ be another gate such that there is a $\sigma \in \Aut(C)$ with $g' = \sigma(g)$. In the partition of $E_C g'$ into orbits $\Orbit_{(g')}(h')$, for $h' \in E_C g'$, there is a unique orbit $\Orbit_{(g')}(h')$ to which $\Orbit_{(g)}(h)$ can be mapped by $\Aut(C)$.
\end{lemma}
\begin{proof}
	Firstly, it is clear that every $\sigma \in \Aut(C)$ that takes $g$ to $g'$ must map $\Orbit_{(g)}(h)$ to some orbit $\Orbit_{(g')}(h')$, for a $h' \in E_C g'$. We now show that there cannot be two distinct $\Orbit_{(g')}(h'), \Orbit_{(g')}(h'')$ that $\Orbit_{(g)}(h)$ can be mapped to. Suppose for a contradiction that there were $\sigma,\sigma' \in \Aut(C)$ with $\sigma(g) = \sigma'(g) = g'$ and $\sigma(\Orbit_{(g)}(h)) = \Orbit_{(g')}(h')$ and $\sigma'(\Orbit_{(g)}(h)) = \Orbit_{(g')}(h'')$. Then $\sigma' \circ \sigma^{-1}$ maps $\Orbit_{(g')}(h')$ to $\Orbit_{(g')}(h'')$ while fixing $g'$. Thus, $\Orbit_{(g')}(h') = \Orbit_{(g')}(h'')$, which is a contradiction because these orbits are distinct.
\end{proof}		
Thus, for any gate $g$ in $C$, and $h$ a parent of $g$, we can define
\begin{align*}
	\Orbit( \Orbit_{(g)}(h) ) := \{  \sigma(\Orbit_{(g)}(h)) \mid  \sigma \in \Aut(C) \},
\end{align*}
and this orbit of orbits contains exactly one $\Orbit_{(g')}(h')$ for every $g' \in \Orbit(g) = \{ \sigma(g) \mid \sigma \in \Aut(C) \}$. The orbit profile $\Omega(P)$ of a path $P = (r = h_1,h_2,...,h_\ell = g)$ is defined as
\[
\Omega(P) := (\Orbit(\Orbit_{(h_\ell)}(h_{\ell-1}) ), \Orbit(\Orbit_{(h_{\ell-1})}(h_{\ell-2})),...,\Orbit( \Orbit_{(h_2)}(h_1))). 
\]
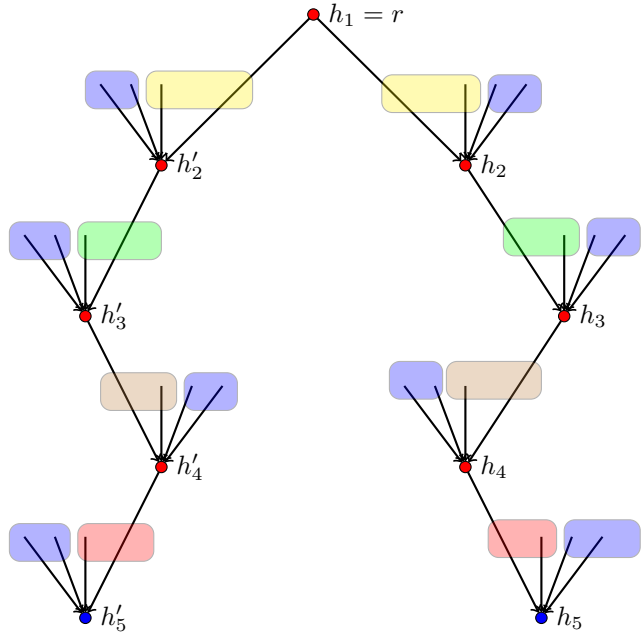
\begin{figure}[H]
	\centering
	\begin{tikzpicture}[dot/.style={draw,circle,minimum size=1.5mm,inner sep=0pt,outer sep=0pt,fill=red},circ/.style={draw,circle,minimum size=2.5mm,inner sep=0pt, fill=red},
		circY/.style={draw,circle,minimum size=2.5mm,inner sep=0pt, fill=yellow}]
		\node[dot,label=right:{$h_1 = r$}] (r) at (0,0) {};
		\node[dot,label=right:{$h_2$}] (h2) at (2,-2) {};	
		\node[dot,label=right:{$h_3$}] (h3) at (3.3,-4) {};	
		\node[dot,label=right:{$h_4$}] (h4) at (2,-6) {};	
		\node[dot,fill=blue,label=right:{$h_5$}] (h5) at (3,-8) {};	
		
		\node[dot,label=right:{$h'_2$}] (h2p) at (-2,-2) {};	
		\node[dot,label=right:{$h'_3$}] (h3p) at (-3,-4) {};	
		\node[dot,label=right:{$h'_4$}] (h4p) at (-2,-6) {};	
		\node[dot,fill=blue,label=right:{$h'_5$}] (h5p) at (-3,-8) {};

		\draw[->,thick] (r) -- (h2);	
		\draw[->,thick] (h2) -- (h3);	
		\draw[->,thick] (h3) -- (h4);	
		\draw[->,thick] (h4) -- (h5);		
		
		\draw[->,thick] (r) -- (h2p);	
		\draw[->,thick] (h2p) -- (h3p);	
		\draw[->,thick] (h3p) -- (h4p);	
		\draw[->,thick] (h4p) -- (h5p);		
		
		\node[draw=none,above=of h2] (h2ghost) {};
		\node[draw=none,right=2mm of h2ghost] (h2ghost2) {};
		\node[draw=none,right=2mm of h2ghost2] (h2ghost3) {};
		
		\node[draw=none,above=of h3] (h3ghost) {};
		\node[draw=none,right=2mm of h3ghost] (h3ghost2) {};
		\node[draw=none,right=2mm of h3ghost2] (h3ghost3) {};
		
		\node[draw=none,above=of h4] (h4ghost) {};
		\node[draw=none,left=2mm of h4ghost] (h4ghost2) {};
		\node[draw=none,left=2mm of h4ghost2] (h4ghost3) {};	
		
		\node[draw=none,above=of h5] (h5ghost) {};
		\node[draw=none,right=2mm of h5ghost] (h5ghost2) {};
		\node[draw=none,right=2mm of h5ghost2] (h5ghost3) {};		
		
		\draw[->,thick] (h2ghost) -- (h2);	
		\draw[->,thick] (h2ghost2) -- (h2);	
		\draw[->,thick] (h2ghost3) -- (h2);	
		
		\draw[->,thick] (h3ghost) -- (h3);	
		\draw[->,thick] (h3ghost2) -- (h3);	
		\draw[->,thick] (h3ghost3) -- (h3);	
		
		\draw[->,thick] (h4ghost) -- (h4);	
		\draw[->,thick] (h4ghost2) -- (h4);	
		\draw[->,thick] (h4ghost3) -- (h4);	
		
		\draw[->,thick] (h5ghost) -- (h5);	
		\draw[->,thick] (h5ghost2) -- (h5);	
		\draw[->,thick] (h5ghost3) -- (h5);

		\draw[rounded corners, fill=yellow, opacity=0.3] (0.9,-0.8) rectangle (2.2,-1.3);
		\draw[rounded corners, fill=blue, opacity=0.3] (2.3,-0.8) rectangle (3,-1.3);
		
		\draw[rounded corners, fill=green, opacity=0.3] (2.6-0.1,-3+0.3) rectangle (3.6-0.1,-3.5+0.3);
		\draw[rounded corners, fill=blue, opacity=0.3] (3.7-0.1,-3+0.3) rectangle (4.4-0.1,-3.5+0.3);	
		
		\draw[rounded corners, fill=blue, opacity=0.3] (1,-4.6) rectangle (1.7,-5.1);
		\draw[rounded corners, fill=brown, opacity=0.3] (1.75,-4.6) rectangle (3,-5.1);	
		
		\draw[rounded corners, fill=red, opacity=0.3] (2.3,-6.7) rectangle (3.2,-7.2);
		\draw[rounded corners, fill=blue, opacity=0.3] (3.3,-6.7) rectangle (4.3,-7.2);

		\node[draw=none,above=of h2p] (h2pghost) {};
		\node[draw=none,left=2mm of h2pghost] (h2pghost2) {};
		\node[draw=none,left=2mm of h2pghost2] (h2pghost3) {};
		
		\node[draw=none,above=of h3p] (h3pghost) {};
		\node[draw=none,left=2mm of h3pghost] (h3pghost2) {};
		\node[draw=none,left=2mm of h3pghost2] (h3pghost3) {};
		
		\node[draw=none,above=of h4p] (h4pghost) {};
		\node[draw=none,right=2mm of h4pghost] (h4pghost2) {};
		\node[draw=none,right=2mm of h4pghost2] (h4pghost3) {};	
		
		\node[draw=none,above=of h5p] (h5pghost) {};
		\node[draw=none,left=2mm of h5pghost] (h5pghost2) {};
		\node[draw=none,left=2mm of h5pghost2] (h5pghost3) {};		
		
		\draw[->,thick] (h2pghost) -- (h2p);	
		\draw[->,thick] (h2pghost2) -- (h2p);	
		\draw[->,thick] (h2pghost3) -- (h2p);	
		
		\draw[->,thick] (h3pghost) -- (h3p);	
		\draw[->,thick] (h3pghost2) -- (h3p);	
		\draw[->,thick] (h3pghost3) -- (h3p);	
		
		\draw[->,thick] (h4pghost) -- (h4p);	
		\draw[->,thick] (h4pghost2) -- (h4p);	
		\draw[->,thick] (h4pghost3) -- (h4p);	
		
		\draw[->,thick] (h5pghost) -- (h5p);	
		\draw[->,thick] (h5pghost2) -- (h5p);	
		\draw[->,thick] (h5pghost3) -- (h5p);

		\draw[rounded corners, fill=blue, opacity=0.3] (-3,-0.75) rectangle (-2.3,-1.25);
		\draw[rounded corners, fill=yellow, opacity=0.3] (-2.2,-0.75) rectangle (-0.8,-1.25);
		
		\draw[rounded corners, fill=blue, opacity=0.3] (-4,-2.75) rectangle (-3.2,-3.25);
		\draw[rounded corners, fill=green, opacity=0.3] (-3.1,-2.75) rectangle (-2,-3.25);	
		
		\draw[rounded corners, fill=brown, opacity=0.3] (-2.8,-4.75) rectangle (-1.8,-5.25);
		\draw[rounded corners, fill=blue, opacity=0.3] (-1.7,-4.75) rectangle (-1,-5.25);	
		
		\draw[rounded corners, fill=blue, opacity=0.3] (-4,-6.75) rectangle (-3.2,-7.25);
		\draw[rounded corners, fill=red, opacity=0.3] (-3.1,-6.75) rectangle (-2.1,-7.25);	
		
	\end{tikzpicture}
	\caption{The colours indicate the partition of each set of parents into orbits. Orbits with the same colour belong to the same orbit of orbits. The orbit profile of the two paths in the picture is thus ``red, brown, green, yellow'', and the two paths are related by an automorphism of the circuit.}
\end{figure}
\begin{lemma}
	\label{lem:XOR2_numPathsWithProfile}
	Let $P = (r=h_1,...,h_\ell = g)$ be a path from $r$ to $g$ in $C_n$. The number of paths in $C_n$ from $r$ to $g$ with orbit-profile $\Omega(P)$ is exactly
	\[
	\prod_{2 \leq i \leq \ell} |\Orbit_{(h_i)}(h_{i-1})|
	\]
\end{lemma}	
\begin{proof}
	We go backwards from $g=h_\ell$ to $r$ and count how many ways there are to construct a path with orbit-profile $\Omega(P)$. In the beginning, there are $|\Orbit_{(h_\ell)}(h_{\ell-1})|$ many options to choose a predecessor of $g$ that is in the orbit required by $\Omega(P)$. Let $h$ be the predecessor of $g$ that we choose. From there, we have $|\Orbit_{(h_{\ell-1})}(h_{\ell-2})|$ predecessors that we could continue with in a way that respects $\Omega(P)$. To see this, we use Lemma \ref{lem:XOR2_orbitsOfOrbits}: No matter which gate we chose for $h$, it is in $\Orbit(h_{\ell-1})$. Therefore, by Lemma \ref{lem:XOR2_orbitsOfOrbits}, there exists a unique $\Orbit_{(h)}(h')$ in $\{ \Orbit_{(h)}(h') \mid h' \in E_C h \}$ that is also a member of $\Orbit(\Orbit_{(h_{\ell-1})}(h_{\ell-2}))$. From this $\Orbit_{(h)}(h')$, we can choose the next gate on our path, and this orbit has the same size as $\Orbit_{(h_{\ell-1})}(h_{\ell-2})$. Hence, we have so far $|\Orbit_{(h_\ell)}(h_{\ell-1}) | \cdot | \Orbit_{(h_{\ell-1})}(h_{\ell-2})|$ possibilities to go two steps from $g$ towards $r$ in a way that complies with the orbit-profile $\Omega(P)$. In the same fashion, we continue counting until we reach the root, and obtain the number of paths that is stated in the lemma.
\end{proof}

\begin{lemma}
	\label{lem:XOR2_finalLemma}
	Let $0 < \epsilon \leq 1$. For each $n$, fix a gate $g_n$ in $C_n$ such that $\Sp_A(g_n)$ has at least two parts of size $\geq \epsilon n$. 
	For every possible orbit-profile $\Omega(P)$ that any path $P$ from the root of $C_n$ to $g_n$ can have, there exists an even number of distinct paths from the root to $g_n$ with exactly that orbit-profile.
\end{lemma}
\begin{proof}
	Fix a path $P$ from the root to $g_n$ in $C_n$ and the corresponding orbit-profile $\Omega(P)$. We are going to show that there exists an even number of distinct paths from the root to $g_n$ with orbit-profile $\Omega(P)$.\\
	By the assumption on $g_n$, it holds $\zeta[g_n](s_1) \geq 1$ and $\zeta[g_n](s_2) \geq 1$ for $s_1, s_2 \geq \epsilon \cdot n$. For the root $r_n$ it holds $\zeta[r_n](n) = 1$ and $\zeta[r_n](s) = 0$ for every $s \neq n$ (because the root is fixed by all permutations in $\Sym_n$). Therefore, the size profiles $\zeta$ must change along the path $P$ from $g$ to $r$. Let $s$ be an even natural number such that $\epsilon \cdot n \leq s \leq \min \{s_1,s_2\}$. This always exists because otherwise we can just make $\epsilon$ a bit smaller such that $\min \{s_1,s_2\}-1 \geq \epsilon \cdot n$.\\
	Assume for a contradiction that for every gate $h \in P$, for its predecessor $h'$ on the path $P$ it holds: $|\Orbit_{(h)}(h')|$ is odd. 
	Then applying Corollary \ref{cor:XOR2_numberOfLargePartsStaysHigh} inductively along the path $P$ shows that $\sum_{i \geq s} \zeta[h](i) \geq \sum_{i \geq s} \zeta[g_n](i) \geq 2$, for every $h \in P$. This is a contradiction to the fact that $\sum_{i \geq s} \zeta[r_n](i) = 1$.\\
	This shows that there must be some $h$ on the path $P$ such that the predecessor $h'$ of $h$ satisfies: $|\Orbit_{(h)}(h')|$ is even. Then the total number of paths from $r_n$ to $g_n$ with profile $\Omega(P)$ is even because by Lemma \ref{lem:XOR2_numPathsWithProfile}, this number is a product containing the even number $|\Orbit_{(h)}(h')|$.
\end{proof}

From this, our main technical theorem follows, which states that not only the number of paths with a given orbit profile, but the total number of paths from $r$ to $g$ is even:\\
\\
\textit{Proof of Theorem \ref{thm:XOR2_evenPaths}}: Every path from $r_n$ to $g_n$ has exactly one orbit-profile. Hence, the number of paths from $r_n$ to $g_n$ is just
\[
\sum_{\Omega \text{ an orbit profile of a path from } r_n \text{ to } g_n} \#(\Omega),
\]
where $\#(\Omega)$ denotes the number of paths with orbit-profile $\Omega$ that end in $g_n$. By Lemma \ref{lem:XOR2_finalLemma}, all summands in this sum are even. \hfill \qedsymbol \\

Finally, let us summarise why Theorem \ref{thm:XOR2_symmetricCircuitsDoNotCompute} (``fully symmetric XOR-circuits are insensitive to all inputs except those labelled with very imbalanced binary strings'') follows from Theorem \ref{thm:XOR2_evenPaths}.\\
\\
\textit{Proof of Theorem \ref{thm:XOR2_symmetricCircuitsDoNotCompute}}:\\ 
Fix any $\epsilon > 0$. Let $g_n$ be an input gate of $C_n$ labelled with a hypercube-edge  $\ell(g_n) = \{0^a1^{n-a}, 0^{a-1}1^{n-a+1}\}$ such that $a \geq \epsilon n$ and $n-a \geq \epsilon n$. It is easy to see that $\Sp_A(g_n) = \{ \{ 1,...,a-1 \}, \{a \}, \{a+1,...,n \} \}$. This contains two parts of size $\geq \epsilon n$, so Theorem \ref{thm:XOR2_evenPaths} applies and the number of paths from the root to $g_n$ is even. By Lemma \ref{lem:XOR2_evenPaths}, the circuit $C_n$ is not sensitive to the input gate $g_n$. \hfill \qedsymbol

\section{Conclusion and future research}
We have defined different classes of choiceless algorithms for the CFI-query, according to the properties of the h.f.\ sets that they necessarily activate. All three currently known algorithms are both super- as well CFI-symmetric. A generalisation of CFI-symmetry is the symmetric basis property (Definition \ref{def:XOR_symmetricBasis}). We have shown that super-symmetric objects which are also CFI-symmetric or have the symmetric basis property can be translated into polynomial-size symmetric XOR-circuits in a meaningful way.\\

We have used this circuit construction to make progress towards showing that no CFI-symmetric CPT-algorithm defines the CFI-query over unordered hypercubes: The existence of such an algorithm would entail the existence of a family of symmetric XOR-circuits whose sizes and orbit sizes are polynomial in the size of the $n$-dimensional hypercube (i.e.\ $2^n$), which compute the XOR over $\Omega(2^n/n^{1.5})$ many input bits, and whose fan-in degree is bounded by $\Oo(n)$ (i.e.\ logarithmic in the hypercube size). Our aim was to show that such circuit families do not exist. We have not fully accomplished this but at least we have identified interesting further restrictions on the circuits which altogether are unsatisfiable: If the orbit size of each circuit is assumed to be exactly $1$ instead just polynomial, and instead of the fan-in dimension bound of $\Oo(n)$, we impose an $\Oo(n)$-bound on the number of children and parents of each gate (per orbit in the stabiliser of the gate), then these circuits cannot compute the XOR over $\Omega(2^n/n^{1.5})$ many input bits. It follows that if nonetheless there does exist a CFI-symmetric algorithm for the hypercube CFI-query, then the corresponding circuit families either have orbit size $> 1$ or must violate the orbit-wise bound on the number of child or parent gates. Thus, the next step should be to try and lift our lower bound techniques to a more general setting. It seems plausible that this can be done but there are technical challenges involved:\\

The first problem is how to argue for circuits whose orbit size is not exactly $1$, but bounded by some polynomial in $2^n$. Then the supporting partition of the root does not necessarily consist of only one part, but it can be many more (although if it is too many, then the orbit size will be greater than $2^{nk}$, which is forbidden). Our argument exploited the fact that the number of linear-size parts in the supporting partition can never decrease along a path from an input gate to the root unless the number of parents is even at some point. But if the supporting partition of the root can now have multiple linear-size parts, then this no longer leads to a contradiction. It might be that with a much more careful analysis of the circuits, our argument could still be recovered in this case, though. Our key technical lemma (Lemma \ref{lem:XOR2_mainPartSizeLemma}) is actually stronger than what we needed in our proof because it gives us several cases in which the number of parents of a gate must be even. Thus, if even parent numbers are forbidden, then the ways in which the size profiles of the supporting partitions can change along a path are very limited. But surely we can expect that not all gates have the same size profile, so changes will occur somewhere, and then again, this will lead to even parent numbers. It is just not clear at this moment how to turn this into a formal argument.\\

The second problem concerns the relationship between the logarithmic bound on the \emph{fan-in dimension}, that we get from Theorem \ref{thm:XOR_lowerboundProgram}, and on the logarithmic orbit-wise fan-in and fan-out bounds that we imposed in the last section. Currently, we do not know if one of these bounds implies the other. Probably, the bound on the parent number is not directly related to fan-in dimension but the bound on the children might be. It would be nice if logarithmic fan-in dimension implied a logarithmic number of children per orbit. Then we would have this covered with our lower bound. In case that the gates in $\Orbit_{(h)}(g) \subseteq hE_C$ all have distinct sensitivity sets $\Xx(g')$, which are also linearly independent as vectors in $\bbF_2^E$, then $|\Orbit_{(h)}(g)|$ is indeed at most the fan-in dimension. But it is unclear how to reason about the properties of these sets $\Xx(g')$, for all $g' \in \Orbit_{(h)}(g)$.\\
For removing the $\Oo(n)$-bound on the orbit-wise parent number of the gates, we have a rough idea. Namely, because our circuits are single-rooted, their levels should get narrower closer to the root. Therefore, it seems plausible that close enough to the root, each gate indeed only has a bounded number of parents because otherwise, the circuit would get wider. The good thing about our even-paths theorem (Theorem \ref{thm:XOR2_evenPaths}) is that it can be applied to any gate in the circuit, not only input gates. So we could potentially focus on the top-most part of the circuit, where its levels only get narrower, and could show that in this top part, all paths cancel each other out. This would suffice to show that the circuit is not sensitive to enough input bits.\\

All in all, it feels like our even-paths technique has more potential and might also work for less restricted circuit classes, perhaps even for \emph{all} circuits satisfying the necessary properties for the existence of a CFI-symmetric algorithm for the hypercube CFI-problem. In particular, it might also be possible to improve our group-theoretic Theorem \ref{thm:G_alternatingPartitionWithLogManyParts}, which says that the alternating supporting partitions can not have linearly many singleton parts. Getting an even more fine-grained understanding of the alternating supporting partitions of groups of index $\leq 2^{nk}$ could be useful. 
So we seem to be in the situation where we probably have not yet reached the limitations of our technique, but nonetheless, making further progress might be technically very challenging.\\

One particular question that could guide further research is to try and prove that no \emph{tree-like} circuits can satisfy the conditions from Theorem \ref{thm:XOR_lowerboundProgramInstantiatedHypercubes}. Our result indirectly shows that no tree-like circuits with logarithmic fan-in degree can have orbit size one with respect to the hypercube automorphisms: Such circuits would satisfy all properties from Theorem \ref{thm:XOR2_symmetricCircuitsDoNotCompute} \emph{and} they would be sensitive to all of their input gates because no cancellations can happen in tree-like circuits. This would contradict Theorem  \ref{thm:XOR2_symmetricCircuitsDoNotCompute}. It remains the question if tree-like circuits with a greater orbit size than one or fewer degree restrictions can satisfy the conditions from Theorem \ref{thm:XOR_lowerboundProgramInstantiatedHypercubes}. We hope the answer will turn out to be negative.\\

Finally, it may be interesting to investigate in how far the new lower bound technique against symmetric XOR-circuits that we developed here can be applied to other scenarios as well. As mentioned in the introduction, studying lower bounds for symmetric circuits also seems to be a promising approach towards separating the algebraic complexity classes VNP and VP. There exist lower bounds against symmetric arithmetic circuits for computing the determinant and permanent polynomials by Dawar and Wilsenach \cite{symmetricCircuitsDeterminant, dawar2020symmetric}. They raise the question in how far these lower bounds can be improved to weaker symmetry groups, and perhaps our technique can be adapted to that end. Of course, the even-paths theorem is probably only useful for circuits which purely consist of XOR-gates; but the statement that the alternating supporting partitions of the gates cannot change much between the layers could lead to new insights. A novelty of our technique in comparison with \cite{dawarAnderson, dawarGreg, symmetricCircuitsDeterminant, dawar2020symmetric} is that it does not use any ``support theorem''. Support theorems are a key ingredient in all these previous works, and they usually state that any gate in a highly symmetric circuit is supported by a constant number of elements of the permutation domain. For poly-size circuits with hypercube-symmetries, as we study here, we believe that a support theorem in that strong form does not hold. Thus, our approach via alternating supporting partitions might perhaps open up a perspective to study such weaker symmetry groups as well.
	
\bibliographystyle{plainurl}
\bibliography{references.bib}	

\newpage
\section{Appendix}

\subsection{Correctness of the inductive matrix construction in Section \ref{sec:generalisedCircuitConstruction}}
\restateNMatricesPermuteWithStabiliser*

\begin{proof}
	Let $[y] \in \Cc[x]$ be the primer of $\Omega_{[y']}$. Let $\pi_{[x'][y']} \in \Stab_G(\mu)$ be the permutation that was used to define $N[x'][y']$ from $N[x][y]$. For the matrix $N[x][y]$, there exists such a homomorphism $h : g_{[y]}(\Stab_G([y] \cap x)) \lra \Sym(J_{[x][y]})$ by definition of the matrix. We define the desired homomorphism $h': g_{[y']}(\Stab_G([y'] \cap x')) \lra \Sym(J_{[x'][y']})$ as follows. For every $\pi \in \Stab_G([y'] \cap x')$ and every $(i,[x'],[y']) \in J_{[x'][y']}$, let
	\[
	h'(g_{[y']}(\pi))(i,[x'],[y']) := (j,[x'],[y']),   
	\]
	where $j$ is the number such that 
	\[
	h(g_{[y]}(\pi^{-1}_{[x'][y']} \circ \pi \circ \pi_{[x'][y']}))(i,[x],[y]) = (j,[x],[y]).
	\]
	For every $\pi \in \Stab_G([y'] \cap x')$, $h'(g_{[y']}(\pi))$ is indeed a permutation in $\Sym(J_{[x'][y']})$, because $N[x][y]$ and $N[x'][y']$ have the same number of rows (and the index sets $J_{[x][y]}$ and $J_{[x'][y']}$ differ only with respect to the second and third entry of the index triples), and $(\pi^{-1}_{[x'][y']} \circ \pi \circ \pi_{[x'][y']}) \in \Stab_G([y] \cap x)$, so $h(g_{[y]}(\pi^{-1}_{[x'][y']} \circ \pi \circ \pi_{[x'][y']}))$ is a permutation on the rows of $N[x][y]$.\\
	The fact that $h'$ is a group homomorphism follows directly from the fact that $h$ is one, and because $(\pi^{-1}_{[x'][y']} \circ \pi_1 \circ \pi_{[x'][y']}) \circ (\pi^{-1}_{[x'][y']} \circ \pi_2 \circ \pi_{[x'][y']}) = \pi_1 \circ \pi_2$.\\
	Finally, we have to show that for every $\sigma \in g_{[y']}(\Stab_G([y'] \cap x'))$, it holds $(h'(\sigma), \sigma)(N[x'][y']) = N[x'][y']$. To prove this, we show that $N[x'][y']_{h'(\sigma)(j,[x'],[y']),\sigma(i)} = N[x'][y']_{(j,[x'],[y']),i}$ for every $(j,[x'],[y']) \in J_{[x'][y']}$ and $i \in I_{[y']}$. In the following, we will use that by definition of $N[x'][y']$, we have: $N[x'][y'] = (\text{id},g_{[y]}(\pi_{[x'][y']}))(N[x][y])$. Let $\sigma = g_{[y']}(\pi) = g_{[y]}(\pi)$ for a $\pi \in \Stab_G([y'] \cap x')$. Then for any $(j,[x'],[y']) \in J_{[x'][y']}, i \in I_{[y']}$, we obtain: 
	\begin{align*}
		N[x'][y']_{h'(\sigma)(j,[x'],[y']),\sigma(i)} &= (N[x][y])_{h(g_{[y]}(\pi^{-1}_{[x'][y']} \circ \pi \circ \pi_{[x'][y']}))(j,[x],[y]),(g_{[y]}(\pi_{[x'][y']})^{-1} \circ \sigma)(i)}\\
		&= (N[x][y])_{h(g_{[y]}(\pi^{-1}_{[x'][y']} \circ \pi \circ \pi_{[x'][y']}))(j,[x],[y]),g_{[y]}(\pi_{[x'][y']}^{-1} \circ \pi \circ \pi_{[x'][y']})(i')}\\
		&=  (N[x][y])_{(j,[x],[y]),i'}.
	\end{align*}
	In the second step, we used that $\sigma(i) = g_{[y']}(\pi)(i)$, and we replaced $i \in I_{[y']}$ with $g_{[y]}(\pi_{[x'][y']})(i')$ for some $i' \in I_{[y]}$ (which can be done because of item (c) of the induction hypothesis for $I_{[y]}$ and $I_{[y']} = I_{\pi_{[x'][y']}[y]}$). We also used that $g_{[y]}$ is a group homomorphism. The last step holds because we already know that $h$ satisfies the property that we are trying to prove for $h'$, i.e.\ $(h(\sigma'),\sigma')N[x][y] = N[x][y]$ for any $\sigma' \in g_{[y]}(\Stab_G([y] \cap x))$ (this is by construction of $N[x][y]$ and $h$). We can continue the equation, using the definition of $N[x'][y']$ again:
	\begin{align*}
		(N[x][y])_{(j,[x],[y]),i'} &= (N[x'][y'])_{(j,[x'],[y']),g_{[y]}(\pi_{[x'][y']})(i')}\\
		&= (N[x'][y'])_{(j,[x'],[y']),i}.
	\end{align*}		
	This proves that $(h'(\sigma), \sigma)(N[x'][y']) = N[x'][y']$, as desired.
\end{proof}

\restateHwellDefined*
\begin{proof}
	Let $[y] \in \Cc[x]$ be the primer of $\Omega_{[y']}$, and write $\sigma := \pi_{[x'][y']}, \sigma' := \pi_{\pi[x']\pi[y']}$. So these are the two permutations in $\Stab_G(\mu)$ that were used to construct $N[x'][y']$ and $N\pi[x']\pi[y']$ from $N[x][y]$. By construction of these matrices, it holds $N[x'][y'] = (\text{id},g_{[y]}(\sigma))(N[x][y])$ and $N[x'][y'] = (\text{id},g_{[y]}(\sigma'))(N[x][y])$. Thus, $(N\pi[x']\pi[y'])_{k,-} = (g_{[y]}(\sigma') \circ g_{[y]}(\sigma)^{-1})((N[x'][y'])_{k,-})$. Since $g_{[y]}$ is a group homomorphism, we can also write this as:
	\[
	(N\pi[x']\pi[y'])_{k,-} = g_{[y]}(\sigma' \circ \sigma^{-1})((N[x'][y'])_{k,-}) \tag{$\star$}
	\]
	If $\sigma' \circ \sigma^{-1}$ were equal to $\pi$, then this would suffice to prove the lemma. However, we only know that $(\sigma' \circ \sigma^{-1})([y'] \cap x') = \pi([y'] \cap x')$. We use this to show the following\\
	\textbf{Claim:} Let $m$ be the number of rows of $N[x'][y']$ and $N\pi[x']\pi[y']$. There exists a permutation $\theta \in \Sym_m$ such that $(N\pi[x']\pi[y'])_{\theta(k),-} = g_{[y]}(\pi)((N[x'][y'])_{k,-})$, for every $k \in [m]$.\\
	\textit{Proof.} It holds that $(\pi \circ \sigma \circ (\sigma')^{-1}) \in \Stab_G(\pi[y'] \cap \pi x')$. By Lemma \ref{lem:XOR_NmatricesPermuteWithStabiliser}, there exists a $\theta \in \Sym_m$ such that $(\theta, g_{[y]}(\pi \circ \sigma \circ (\sigma')^{-1})))(N\pi[x']\pi[y']) = N\pi[x']\pi[y']$. It holds $\sigma \circ (\sigma')^{-1} = (\sigma' \circ \sigma^{-1})^{-1}$. Thus, by $(\star)$ we have: $(N[x'][y'])_{k,-} = g_{[y]}(\sigma \circ (\sigma')^{-1})((N\pi[x']\pi[y'])_{k,-})$. It follows that $(N\pi[x']\pi[y'])_{\theta(k),-} = g_{[y]}(\pi)(N[x'][y'])_{k,-}$ for every $k \in [m]$. This proves the claim.\\
	The claim entails the lemma because $g_{[y']} = g_{[y]}$ (as $\Omega_{[y]} = \Omega_{[y']})$, and so we know that the rows of $N\pi[x']\pi[y']$ are the rows of $N[x'][y']$, with an application of $g_{[y']}(\pi)$ to the columns, and a potential reordering of the rows.
\end{proof}	

\subsection{Objects with and without symmetric bases}
\restateCFIsymmetricHasSymmetricBasis*
\begin{proof}
	Let $x \in \tc(\mu)$ and $[y] \in \Cc[x]$. Let $\Gamma := \Stab_E([y] \cap x) \leq \bbF_2^E$. We have to define two bases $\Bb_\Gamma \subseteq \Bb$ of $\Gamma$ and of $\bbF_2^E$, respectively, such that the group $\Stab_{\Stab_G([y] \cap x)}(\Bb_\Gamma) \cap \Stab_{\Stab_G([y] \cap x)}(\Bb)$ has polynomial index in $\Stab_G([y] \cap x)$. 
	Since $\mu$ is CFI-symmetric, by Definition \ref{def:XOR_CFIsymmetric}, the $\Aut_{\CFI}(\GG)$-orbit of $[y] \cap x$ has size exactly two. We have $\Aut_{\CFI}(\GG) \cong \bbF_2^E$, so by the Orbit-Stabiliser Theorem, $\Stab_E([y] \cap x)$ is a subspace of $\bbF_2^E$ with co-dimension one.
	We use this to analyse the structure of the space $\Stab_E([y] \cap x)$. The group $\Stab_G([y] \cap x)$ is a subgroup of $\Aut(G) \leq \Sym(V)$ and therefore also acts on the edge set $E$. Thus, we can partition $E$ into its $\Stab_G([y] \cap x)$-orbits. Let $\Pp = \{P_1,...,P_m\}$ denote this orbit partition of $E$.\\
	\\
	\textbf{Claim:} There is a partition $E = A \uplus B$ such that $\Stab_E([y] \cap x) = \bbF_2^A \oplus \widetilde{\bbF}_2^B$ and $B \neq \emptyset$. Moreover, $A$ and $B$ are unions of $\Stab_G([y] \cap x)$-orbits.\\
	\textit{Proof of claim.} Let $\Pp' \subseteq \Pp$ denote the set of orbits $P_i$ such that for any $e \in P_i$, the unit vector $\chi(e)$ is in $\Stab_E([y] \cap x)$. Note that whenever $\chi(e) \in \Stab_E([y] \cap x)$, then $\chi(e') \in \Stab_E([y] \cap x)$ for every $e'$ in the orbit of $E$ because $\Stab_G([y] \cap x)$ is transitive on each orbit and the space $\Stab_E([y] \cap x)$ is invariant under the action of $\Stab_G([y] \cap x)$ on the coordinates (Lemma \ref{lem:XOR_actionOfAutomorphismsOnStabiliserSpaces}). We let $A := \bigcup \Pp'$ and $B := E \setminus A$. It remains to show that $\widetilde{\bbF}_2^B$ is a subspace of $\Stab_E([y] \cap x)$. Assume for a contradiction that there is some vector $\mathbf{v} \in \bbF_2^E$ with even Hamming weight on $B$ and zero on $A$ which is not contained in $\Stab_E([y] \cap x)$. 
	Since $\Stab_E([y] \cap x)$ has co-dimension exactly one in $\bbF_2^E$, and since moreover, by definition of $\Pp'$, no unit vector $\chi(e)$ with $e \in B$ is in $\Stab_E([y] \cap x)$, we know that for any such unit vector $\chi(e)$ with $e \in B$, there exists some $\mathbf{w}_e \in \Stab_E([y] \cap x)$ such that $\mathbf{v} = \mathbf{w}_e + \chi(e)$. But then, every vector with Hamming-weight exactly two on $B$ is in $\Stab_E([y] \cap x)$. Namely, for any two $e,e' \in B$, it then holds that $\mathbf{w}_e + \mathbf{w}_{e'} = \chi(e) + \chi(e')$, and we have $\mathbf{w}_e, \mathbf{w}_{e'} \in \Stab_E([y] \cap x)$. So then, $\Stab_E([y] \cap x)$ does contain $\widetilde{\bbF}_2^B$. In total, this proves the claim (it holds $B \neq \emptyset$ because otherwise, the co-dimension would be zero).\\
	
	Now it is not hard to define a symmetric basis for $\Stab_E([y] \cap x)$.  
	Fix an arbitrary edge $f \in B$. We define 
	\[
	\Bb_\Gamma := \{ \chi(e) \mid e \in A\} \cup \{ \chi(\{ e,f\}) \mid e \in B \setminus \{f \} \}.
	\]
	One can check that this is indeed a basis of $\bbF_2^A \oplus \widetilde{\bbF}_2^B$. The basis $\Bb$ of $\bbF_2^E$ is then simply defined as $\Bb := \Bb_\Gamma \cup \{ \chi(f) \}$.\\
	Now the group
	$\Stab_{\Stab_G([y] \cap x)}(\Bb_\Gamma) \cap \Stab_{\Stab_G([y] \cap x)}(\Bb)$ contains all permutations in $\Stab_G([y] \cap x)$ that fix the edge $f$ and fix the sets $A$ and $B$ (setwise). By the Claim, $A$ and $B$ are unions of $\Stab_G([y] \cap x)$-orbits, so the latter condition is fulfilled by all permutations in $\Stab_G([y] \cap x)$. Therefore,  $\Stab_{\Stab_G([y] \cap x)}(\Bb_\Gamma) \cap \Stab_{\Stab_G([y] \cap x)}(\Bb)$ is simply the pointwise stabiliser of $f$ in $\Stab_G([y] \cap x) \leq \Sym(E)$, and this has index at most $|E|$. This is polynomial in $|\GG^S|$.
\end{proof}

\restateCounterExample*	
\begin{proof}
	Define $t(n)$ as the next even natural number $\geq n$. We now construct $\Gamma_n$ and $\mathbf{G}_n$. Let $\Pp_n$ be a partition of $[t(n)]$ into $\approx \log n$ many parts such that each part is roughly of the same size, namely $\approx \frac{n}{\log n}$. Importantly, every part must be of even size; such a partition exists because $t(n)$ is even. For a part $P \in \Pp_n$, let $\widetilde{\bbF}_2^P \leq \bbF_2 ^{t(n)}$ denote the Boolean vector space that contains all vectors whose projection to $P$ has even Hamming weight and which are zero outside of $P$. Then we define
	\[
	\Gamma_n := \bigoplus_{P \in \Pp_n} \widetilde{\bbF}_2^P.
	\]
	In other words, $\Gamma_n$ contains exactly those vectors that have even Hamming weight on each of the parts in $\Pp_n$ (but not all vectors with even Hamming weight in $\bbF_2^{t(n)}$, namely not the vectors which are odd on an even number of parts). The permutation group $\mathbf{G}_n \leq \Sym_{t(n)}$ is defined as the largest group that setwise stabilises the partition $\Pp_n$. So $\mathbf{G}_n$ contains the direct product $\mathbf{H_n} := \prod_{P \in \Pp_n} \Sym(P)$ and all permutations that map each part of $\Pp_n$ to another part.\\
	
	It is clear that $\Gamma_n$ is invariant under $\mathbf{G}_n$. Furthermore, the codimension of $\Gamma_n$ is logarithmic in $t(n) \approx n$: Suppose $\Bb_\Gamma$ is any basis of $\Gamma_n$. Then it can be extended to a basis of $\bbF_2^{t(n)}$ by adding one unit vector $e_P$ for each part $P \in \Pp_n$, such that $e_P$ has a $1$-entry in $P$ and is zero otherwise. The number of parts is logarithmic, so the same holds for the codimension. Finally, we have to prove the third condition.\\
	Let $\Bb_\Gamma \subseteq \Bb$ be arbitrary bases for $\Gamma$ and $\bbF_2^{t(n)}$, respectively. Observe that $\Stab_{\mathbf{G}_n}(\Bb) \leq \Stab_{\mathbf{G}_n}(\Bb \setminus \Bb_\Gamma)$ because $\Gamma_n$ is $\mathbf{G}_n$-invariant and so, the vectors in $\Bb \setminus \Bb_\Gamma$ cannot be moved into $\Gamma$. Therefore:
	$[\mathbf{G}_n : \Stab(\Bb \setminus \Bb_\Gamma)] \leq [\mathbf{G}_n : \Stab(\Bb)]$. Thus, it suffices to show the desired lower bound for $[\mathbf{G}_n : \Stab(\Bb \setminus \Bb_\Gamma)]$. Let $\mathbf{w}_1,...,\mathbf{w}_{\log n}$ be an enumeration of $\Bb \setminus \Bb_\Gamma$. For each $i \leq \log n$, let $\Qq_i \subseteq \Pp_n$ denote the set of parts $P \in \Pp_n$ such that $\mathbf{w}_i$ has odd Hamming weight on $P$. We know that for each $i$, $\Qq_i \neq \emptyset$ because otherwise, $\mathbf{w}_i$ would be in $\Gamma$. Moreover, each $P \in \Pp_n$ is in at least one of the $\Qq_i$ because otherwise, $\Bb$ would not generate the whole space $\bbF_2^{t(n)}$.\\
	Now in order to estimate $|\Stab(\Bb \setminus \Bb_\Gamma)|$, we first estimate the size of the pointwise stabiliser of $\Bb \setminus \Bb_\Gamma$ in $\mathbf{H}_n$, $\StabP_{\mathbf{H}_n}(\Bb \setminus \Bb_\Gamma)$. This is the subgroup of $\mathbf{H}_n$ that stabilises each $\mathbf{w}_i \in \Bb \setminus \Bb_\Gamma$, so it consists of all permutations that fix each part $P \in \Pp_n$ and each $\mathbf{w}_i$. 
	We can bound this stabiliser as follows:
	\[
	|\StabP_{\mathbf{H}_n}(\Bb \setminus \Bb_\Gamma)| \leq \prod_{P \in \Pp_n} (|P|-1)!
	\]
	This holds because for each part $P \in \Pp_n$, there is a vector $\mathbf{w}_i \in \Bb \setminus \Bb_\Gamma$ which has odd weight on $P$. Since each part $P$ has even size, $\mathbf{w}_i$ is not the all-$1$-vector on $P$ (nor the all-zero vector, of course); therefore, the vector $\mathbf{w}_i$ is not fixed by all permutations in $\Sym(P)$ but at most by $(|P|-1)!$ many of them (more precisely by $k! \cdot (|P|-k)!$ many, if $k$ is the number of $1$-entries in $P$ -- but this is at most $(|P|-1)!$). 
	Now because $[\mathbf{G}_n  :  \mathbf{H}_n] = (\log n)!$, we have
	\[
	| \StabP_{\mathbf{G}_n}(\Bb \setminus \Bb_\Gamma)| \leq  (\log n)! \cdot |\StabP_{\mathbf{H}_n}(\Bb \setminus \Bb_\Gamma)|.
	\]
	Furthermore, $[\Stab_{\mathbf{G}_n}(\Bb \setminus \Bb_\Gamma)  :  \StabP_{\mathbf{G}_n}(\Bb \setminus \Bb_\Gamma) ] \leq (\log n)!$. So in total, we get:
	\[
	\Stab_{\mathbf{G}_n}(\Bb \setminus \Bb_\Gamma)  \leq  (\log n)!^2 \cdot  \prod_{P \in \Pp_n} (|P|-1)!
	\]
	Since $|\mathbf{G}_n| = (\log n)! \cdot \prod_{P \in \Pp_n} |P|!$, we get for the index:
	\[
	[\mathbf{G}_n : \Stab_{\mathbf{G}_n}(\Bb)] \geq [\mathbf{G}_n : \Stab_{\mathbf{G}_n}(\Bb \setminus \Bb_\Gamma)] \geq 1/(\log n)! \cdot \prod_{P \in \Pp_n} |P| \geq \Big( \frac{n}{(\log n)^2}  \Big)^{\log n}.
	\]
	The last inequality follows because $\Pp_n$ consists of $\log n$ many parts of size $\frac{n}{\log n}$ each, and because $(\log n)! \leq (\log n)^{\log n}$.
\end{proof}

\subsection{Homogeneity of hypercubes}

\restateMainHomogeneity*

\begin{proof}
	We assume that $\tp(\bar{\alpha}\gamma) = \tp(\bar{\alpha}\gamma')$.
	The first half of the proof consists in establishing that then, $\gamma, \gamma'$ are in the same edge or vertex gadget. Assume first that $\gamma, \gamma' \in \widehat{E}$. Now suppose for a contradiction that the edge gadget of $\gamma$ and $\gamma'$ is not the same. Let $\{u,v\} \in E_n$ be the edge in whose gadget $\gamma$ is, and let $\{u',v'\} \in E_n$ be the corresponding edge for $\gamma'$.
	W.l.o.g.\ we may assume $u \neq u'$ and $u \neq v'$. Now let $s_1,...,s_n \in E_n$ be the edges that form the star which is covered by $\bar{\alpha}$ according to the assumption. We may assume that the centre of the star is the string $0^n$, because the automorphism group of the hypercube is transitive and so we can always move the centre of the star to $0^n$. Let $s_i$ denote the edge $\{0^n, 0^{i-1}10^{n-i}\}$, i.e.\ the edge along which the position $i$ is flipped.\\
	Let $U \subseteq [n]$ be the positions at which the string $u$ is $1$. We construct a $\Cc^{5}$-formula $\phi_u(x)$ that defines the gadget $u^*$ in $\HH^i_n$ using the star $s_1,...,s_n$ as parameters. More precisely, let $s'_1,...,s'_n$ be the respective vertices in the edge gadgets that occur in $\bar{\alpha}$.\\
	Our formula uses some auxiliary formulas: $\psi_{\dist=\ell}(x,y) \in \Cc^3$ which asserts that there is a path of length $\ell$ from $x$ to $y$. This can be expressed with only three variables by requantifying variables in an alternating way (see e.g.\ Proposition 3.2 in \cite{immermanLander}). Also, we use a formula $\psi_\approx(x,y)$ which asserts that $x$ and $y$ are in the same vertex-gadget. This can be expressed by saying that both of them have exactly $n$ neighbours, and: For every neighbour $z$ of $x$, $z$ is either also a neighbour of $y$ or adjacent to a neighbour of $y$ (in the same edge-gadget). The same must hold for every neighbour $z$ of $y$. Expressing this requires not more than five variables in total. Now we define:
	\begin{align*}
		\phi_u(x) := &\bigwedge_{i \in U} \exists z  (\psi_{\dist=1+2(|U|-1)}(s'_i,z) \land \psi_\approx(z,x) ) \land \\
		&\bigwedge_{i \in [n] \setminus U} \neg \exists z  (\psi_{\dist=1+2(|U|-1)}(s'_i,z) \land \psi_\approx(z,x) ).
	\end{align*}
	\textbf{Claim:} $\HH_n^i \models \phi_u(a)$ iff $a \in u^*$.\\
	\textit{Proof of claim:} We are assuming that $s_i$ is the edge between $0^n$ and the string with a $1$ at position $i$. Now we show that $\HH_n^i \models \phi_u(a)$ if $a \in u^*$: We have to check that for every $i \in [n]$, the respective conjunct of the formula is satisfied. If $i \in U$, then for any $s_i' \in s_i^*$, there exists $z \in u^*$ such that there is a path from $s_i'$ to $a$ of length exactly $1+2(|U|-1)$ in $\HH_n^i$: The path goes one step from $s_i'$ into the vertex gadget for $0^{i-1}10^{n-i}$, and from there, the path follows a shortest path of length $|U|-1$ in $\H_n$ that goes from  $0^{i-1}10^{n-i}$ to the vertex $u$ and flips the remaining $|U|-1$ zeros on the way. That path in the CFI-structure $\HH^i_n$ is twice as long because every edge is subdivided by a gadget. The path will end in some node in the vertex-gadget $u^*$. If $i \notin U$, then there is no path from $s_i'$ of length $1+2(|U|-1)$ that ends in a node in $u^*$: The shortest path from $s_i'$ into the gadget $u^*$ requires $1+2|U|$ steps. Hence, $\HH_n^i \models \phi_u(a)$.\\
	If $a \notin u^*$, then $\HH_n^i \not\models \phi_u(a)$, because $\phi_u(a)$ can only be satisfied if the required paths exist (do not exist, respectively) in $\H_n$, and the above arguments also show that these conditions are only satisfiable if $a$ is in the gadget of $u$. This proves the claim.\\
	
	Therefore, we have
	\[
	\HH_n^i \models \exists x ( E x \gamma \land \phi_u(x)  ),
	\]
	but
	\[
	\HH_n^i \not\models \exists x ( E x \gamma' \land \phi_u(x)  ).
	\]
	This is a contradiction to the assumption that $\bar{\alpha}\gamma$ and  $\bar{\alpha}\gamma'$ have the same $\Cc^{\tw_n}$-type, because $\phi_u$ uses only $\Oo(n)$ many variables. Thus we have shown that there is one edge $g \in E_n$ such that $\gamma, \gamma' \in g^*$.\\
	In case that $\gamma$ and $\gamma'$ are both in vertex gadgets, then the same argument shows that they must be in the same gadget $u^*$ because we can define this gadget with the above formula. Similarly, we can argue if $\gamma$ is in a vertex gadget and $\gamma'$ is in an edge gadget: Then we define the vertex gadget of $\gamma$ with the above formula, and $\gamma'$ will not satisfy it. So the types of the tuples being equal entails that $\gamma$ and $\gamma'$ must be in the same gadget, be it of an edge or vertex.\\
	
	In the second half of the proof we show that there is an automorphism $\rho \in \Aut(\HH_n^i)$ that maps $\bar{\alpha}\gamma$ to $\bar{\alpha}\gamma'$, again under the assumption that $\tp(\bar{\alpha}\gamma) = \tp(\bar{\alpha}\gamma')$. 
	We first deal with the case that $\gamma, \gamma' \in g^*$ for some edge gadget $g^*$. 
	It is not necessary to permute the hypercube, so it suffices to find an edge-flip automorphism, i.e.\ $\rho \in \Aut_{\CFI}(\HH_n^i)$.\\
	In case that $\gamma = \gamma'$, there is nothing to show. So let us assume that w.l.o.g.\ $\gamma = g_0$ and $\gamma' = g_1$. We need to find $\rho \in \Aut(\HH_n^i)$ such that $\rho(g_0) = g_1$, and such that $\rho$ fixes $\bar{\alpha}$. Now we call an edge $e = \{u,v\} \in E_n$ \emph{fixed} if $e_0$ or $e_1$ occurs in $\bar{\alpha}$. We call a vertex $v \in V_n$ \emph{fixed} if some node in$v^*$ occurs in $\bar{\alpha}$. We know that $g$ is not fixed because if it is, then $\bar{\alpha}g_0$ and $\bar{\alpha}g_1$ have different types. A cycle in $\H_n$ is called \emph{fixed} if at least one edge or one vertex on it is fixed. Else, the cycle is \emph{free}. If there exists a free cycle in $\H_n$ on which $g$ lies, then the desired automorphism is $\rho_F \in \Aut_{\CFI}(\HH_n^i)$, where $F$ is the edge-set of the free cycle.\\
	Otherwise, every cycle on which $g$ lies is fixed. We want to show that in this case, $\bar{\alpha}\gamma$ and $\bar{\alpha}\gamma'$ do not have the same $\Cc^{\tw_n}$-type and so, this situation cannot occur.\\
	
	Let $u,v \in \{0,1\}^n$ be the endpoints of $g$. Let $X_u \subseteq \{0,1\}^n$ be the set of vertices in $\H_n$ that are reachable from $u$ via paths using only free vertices and free edges in $E_n \setminus \{g\}$. The set $X_u$ is meant to include the fixed vertices which are reachable in this way. Similarly, we define $X_v$ as the set of reachable vertices from $v$ via such free paths (also including fixed vertices). Because no free cycle exists, the sets $X_u$ and $X_v$ must be disjoint (except for potential shared fixed vertices). We restrict the graph to the smaller of these two sets, w.l.o.g.\ this is $X_u$. So let $X := X_u$.
	We now consider the graphs $G$ and $G'$, which are induced subgraphs of $\HH^i_n$ on the universe 
	\begin{align*}
		\Xx := \{ v^Y \mid v^Y \in v^*, v \in X \} &\cup \{ e_j \mid j \in \{0,1 \}, e \in E(\H_n[X]) \}\\
		&\cup \{ a \in V(\HH^i_n) \mid a \text{ an entry in } \bar{\alpha}\gamma \} \cup \{ \gamma'\}.
	\end{align*}
	Now $G := (\HH_n^i[\Xx], \bar{\alpha}\gamma)$ and $G' := (\HH_n^i[\Xx], \bar{\alpha}\gamma')$, that is, they are both the same induced subgraph of the CFI-graph, expanded with the respective tuples of constants (strictly speaking, the constant symbol for $\gamma$ and $\gamma'$ should be the same, but with the interpretation $\gamma$ and $\gamma'$, respectively).\\ 
	As for the size of $X \subseteq V(\H_n)$, we have: $|X| <  |\bar{\alpha}\gamma| \cdot n= \tw_n-n$. To see this, let $\delta X \subseteq E_n$ be the cut of $X$ in $\H_n$, i.e.\ the set of edges between $X$ and its complement. The Cheeger number of $\H_n$, which denotes the minimum of $\frac{|\delta A|}{|A|}$ over all $A \subseteq V(\H_n)$ with $|A| \leq |V(\H_n)|/2$, is between $1$ and $2 \sqrt{n}$. This can be seen from the Cheeger inequalities (see e.g.\ in \cite{cheegerInequality}) together with the fact that the smallest non-zero Eigenvalue of the Laplacian of any Hypercube is $2$ \cite{hypercubeLaplacian} (the fastest way to look this up is actually Wikipedia). This means that $|\delta X| \geq |X|$. From this it follows that $|\bar{\alpha}\gamma| \cdot n \geq |X|$ because the edges in the cut $\delta X$ are exactly the fixed edges, and each entry of $\bar{\alpha}\gamma$ fixes at most $n$ edges. Moreover, we can say that the inequality must actually be strict, so $|\bar{\alpha}\gamma| \cdot n > |X|$. This is because not all entries in $\bar{\alpha}$ are used to fix the edges in $\delta X$; some entries of $\bar{\alpha}$ must also be in $\delta X_v \setminus \delta X_u$. Hence, we have $|X| < |\bar{\alpha}\gamma| \cdot n= \tw_n-n$.\\
	\\
	\textbf{Claim:} Spoiler wins the bijective 
	$\tw_n$-pebble game on $G$ and $G'$.
	\\
	\textit{Proof of claim:} Trivially, the treewidth of $\H_n[X]$ is strictly less than $\tw_n-n$ because $|X| < \tw_n-n$. Hence, the Cops win the Cops and Robber game on $\H_n[X]$ with $\tw_n-n$ many cops.
	Spoiler's goal is to pebble in both structures $G$ and $G'$ the vertex $e_1$ in every edge in $E_n(u) \setminus \{g\}$, i.e.\ every edge incident to $u$, except $g$. If he achieves that, then he wins in the next round: Assume w.l.o.g.\ that the gadget $u^*$ is even. Then in $G'$, an even number of neighbours of every vertex in $u^*$ is pebbled or equal to the constant $\gamma' = g_1$. In $G$, this number is odd for every vertex in $u^*$, because $\gamma = g_0$. Therefore, Spoiler can then place an additional pebble on an arbitrary vertex in $u^*$ and wins because Duplicator's bijection must map the gadget $u^*$ in $G$ to $u^*$ in $G'$.\\
	Now Spoiler can achieve this goal by a standard argument, as for example given in \cite{atseriasBulatovDawar}: Initially, the ``target vertex'' for Spoiler is $u$. This means that he has to pebble the $e_1$-vertices in all its incident edges in both graphs. Duplicator's bijections can flip edges, which changes the set of target vertices for Spoiler. Suppose $F \subseteq E(\H_n[X])$ is the set of edges flipped by Duplicator in a given round. This has the effect that every vertex in $\H_n[X]$ whose $F$-degree is odd changes its role from target- to non-target vertex and vice versa. However, Duplicator cannot flip edges which are pebbled by Spoiler. If Spoiler places his pebbles on edges (or their endpoints) according to the Cops' winning strategy (while Duplicator ``moves the robber'' by flipping paths), then he can eventually pin down a target vertex that Duplicator cannot move anywhere else. From such a position, he can enforce a situation as described above and wins the game. This argument is well-known; the only additional difficulty in our setting is that the base graph is not ordered, so we have to argue that Duplicator cannot win by playing bijections other than edge-flips. This can be enforced by Spoiler, using at most $n$ extra pebbles: The important observation is that Duplicator's bijection must respect distances to all pebbles on the board and to the parameters $\bar{\alpha}\gamma$. So Duplicator can only map a gadget $w^*$ to some other gadget $\pi(w^*)$ if the vertices $w$ and $\pi(w)$ have the same distance in $\H_n[X]$ to every pebbled vertex and parameter. This holds because if Duplicator disrespects distances between such marked elements, then Spoiler can easily win using three pebbles, that he moves along the shortest paths.\\
	Now Spoiler can simply start by pebbling some star in $\H_n[X]$ with $n$ pebbles (this is possible because we are playing with $\tw_n$ pebbles, but only $\tw_n-n$ many are needed to simulate the Cops' winning strategy). Once a star is pebbled in $G$ and $G'$, it follows with the argument used earlier in the proof of this lemma that any $w \in V(\H_n[X])$ has a unique set of distances to the edges of the star, and therefore, Duplicator is then forced to map every $w^*$ to a unique vertex gadget $\pi(w^*)$. This entails that there is also a unique edge gadget $\pi(e^*)$ that she has to map each $e^*$ to. So, from that moment on, Duplicator is indeed limited to playing only edge-flips, and then, Spoiler wins in the aforementioned way using the Cops' strategy. This proves the claim.\\
	
	The claim directly entails that $\bar{\alpha}\gamma$ and $\bar{\alpha}\gamma'$ do not have the same $\Cc^{\tw_n}$-type if $g$ does not lie on a free cycle. This finishes the case where $\gamma$ and $\gamma'$ are in an edge gadget. The other case is that $\gamma = u^*_Y$ and $\gamma' = u^*_{Y'}$ are both in some vertex gadget $u^*$. Then there is an even-sized set of edges $F = Y \triangle Y'$ incident with $u$ in $\H_n$ such that we have to flip the edges in $F$ (and no other edges in $E(u)$) in order to map $\gamma$ to $\gamma'$. This is possible if we can pair up the edges in $F$ in such a way that each pair $\{f_1,f_2\} \subseteq F$ lies on a free cycle (that avoids all other edges in $E(u)$). Suppose for a contradiction that there is some pair $\{f_1,f_2\} \subseteq F$ which is not on a free cycle. Then let $g := f_1$, and make the same argument as above in the case where we wanted to flip the edge $g$ (with the difference that we now remove both $f_1$ and $f_2$ in order to get the two sets $X_u$ and $X_v$ that are not connected by any free path). Then the above proof shows that the two nodes in the gadget $f_1^*$ are distinguishable in $\Cc^{\tw_n}$ using the parameters $\bar{\alpha}$. But then also $\gamma$ and $\gamma'$ in $u^*$ are distinguishable because one of them is adjacent to the $0$-node in $f_1^*$, and the other is adjacent to the $1$-node in $f_1^*$. So again, the required free cycles must exist because otherwise, $\bar{\alpha}\gamma$ and $\bar{\alpha}\gamma'$ have distinct $\Cc^{\tw_n}$-types.
\end{proof}

\subsection{Alternating supporting partitions only have sublinearly many singleton parts}

Here, we prove Theorem \ref{thm:G_alternatingPartitionWithLogManyParts}, which depends on the following key lemma:
	
\restateContainmentOfAlternatingGroup*

For the proof of this lemma, we have to introduce a few notions from group theory first. These can be found for example in the textbook \cite{dixonMortimer}.  	
For a group $G \leq \Sym(\Omega)$, and a subset $\Delta \subseteq \Omega$, $G^{(\Delta)}$ denotes the pointwise stabiliser of $\Delta$ in $G$. If $\Delta$ is a union of orbits of $G$, then we write $G^\Delta$ to denote the restriction of $G$ to its action on $\Delta$. This is a subgroup of $\Sym(\Delta)$.\\

A group $G \leq \Sym(\Omega)$ acts \emph{transitively} on $\Omega$ if every element can be mapped to every other element by $G$, so if $\Omega$ is itself an orbit. If $G$ acts transitively on $\Omega$, then a non-empty set $\Delta \subseteq \Omega$ is called a \emph{block} if for each $\pi \in G$, $\pi(\Delta) = \Delta$ or $\pi(\Delta) \cap \Delta = \emptyset$. A \emph{block system}, or \emph{system of imprimitivity}, is a partition of $\Omega$ into blocks (of equal size). The group $G$ acts as a permutation group on the set of blocks because it always maps blocks to blocks. Every transitive group has the trivial block systems in which each point forms a singleton block, or the whole point set is one block, respectively. If a transitive group $G$ has other block systems than these two, then $G$ is called \emph{imprimitive}, and otherwise, \emph{primitive}. In particular, primitive groups are always transitive.\\

A subgroup $N \leq G$ is called \emph{normal} (denoted $N \lhd G$) if its right and left cosets coincide, i.e.\ if $\gamma N = N \gamma$ for every $\gamma \in G$. An equivalent formulation is that for all $g \in G, h \in N$, we have $ghg^{-1} \in N$. We are interested in normal subgroups because they can be factored out: If $N \lhd G$, then $G/N$ is the group whose elements are the cosets of $N$, that is: For any two $\gamma N, \gamma' N$, the group operation in the factor group $G/N$ is defined as $(\gamma N) \circ (\gamma' N) = (\gamma \circ \gamma') N$. Thus, the order $|G/N|$ of the factor group is equal to the index $[G : N]$, and so, $|G| = |N| \cdot |G/N|$.\\

If $G \leq \Sym(\Omega)$ is intransitive and $\Delta \subseteq \Omega$ is an orbit of $G$, then $G^{(\Delta)}$, the pointwise stabiliser of the orbit, is a normal subgroup of $G$, as one can easily verify. The factor group $G/G^{(\Delta)}$ is isomorphic to $G^{\Delta}$, the action of $G$ on $\Delta$. Also, if $G$ has a non-trivial block system, then the subgroup of $G$ that fixes every block setwise is normal in $G$. Factoring out this stabiliser yields a group that is isomorphic to the action of $G$ on the blocks. Thus, if $G$ is intransitive or imprimitive, it has these mentioned ``canonical'' normal subgroups. These can then be factored out, which is useful in inductive proofs. The only case where it is not clear how to factor out a normal subgroup is if $G$ is primitive. Note that the primitive cases $G = \Sym(\Omega)$ or $G = \Alt(\Omega)$ are not difficult: The symmetric group has the alternating group as a normal subgroup, which leaves $\bbZ_2$ when it is factored out. The alternating group is \emph{simple}, which means that it only has itself and the trivial group $\{1\}$ as normal subgroups. The other primitive cases are less clear but thankfully, the finite primitive groups have been classified completely. For our proof, the following theorem by Babai, which essentially summarises the relevant primitive cases, is sufficient:
\begin{theorem}[Theorem 3.2.1 in \cite{babai}]
	\label{thm:G_primitiveGroups}
	Let $G \leq \Sym_n$ be a primitive group of order $|G| \geq n^{1+\log n}$ where $n$ is greater
	than some absolute constant. Then $G$ has a normal subgroup $N$ of index $\leq n$ such that $N$ has a system of imprimitivity on which $N$ acts as a Johnson group $\Alt^{(t)}_k$ with $k \geq \log n$.
\end{theorem}
Note that the theorem in \cite{babai} has a typo in the order of $G$, which we have corrected here. The \emph{Johnson group} $\Alt^{(t)}_k$ is isomorphic to $\Alt_k$, the alternating group on $k$ elements, but $\Alt^{(t)}_k$ acts on the set of all $t$-tuples over $[k]$ (in the natural way). So the above theorem guarantees the existence of a normal subgroup $N$ in any large enough primitive group, and moreover, it tells us that $N$ more or less looks like an alternating group. This will essentially be one of the base cases in the proof of Lemma \ref{lem:G_containmentOfAlternatingGroup}.\\

Before we can start with that proof, we need one more concept, namely the \emph{composition series} of a group $G$. This is a series $1 = H_0 \lhd H_1 \lhd ... \lhd H_n = G$ such that each $H_i$ is a maximal proper normal subgroup of $H_{i+1}$. The factors $H_{i+1}/H_i$ are called the \emph{composition factors} of $G$. Every finite group has such a composition series, which is not necessarily unique. But by the Jordan-Hölder theorem, every composition series yields the same composition factors (see for example \cite{Rose2009}). Therefore, no matter in which order we factor out normal subgroups of a given group $G$, we will eventually encounter the same composition factors (just like in the prime factorisation of a natural number). This holds even if we do not factor out a \emph{maximal} normal subgroup in each step. Therefore, it holds:
\begin{lemma}
	\label{lem:G_recursiveStep}
	Let $G$ be a group and $H$ be a composition factor of $G$. Let $N \lhd G$ be a normal subgroup. Then $H \cong G/N$, or $H$ is a composition factor of $N$ or of $G / N$.
\end{lemma}	
\begin{proof}
	If $N$ is a maximal normal subgroup in $G$, then there exists a composition series of $G$ of the form $1 \lhd ... \lhd N \lhd G$. Then either $H = G/N$, or $H$ appears as a composition factor later in the series, which means that it is a composition factor of $N$. If $N$ is not a maximal normal subgroup in $G$, then
	we have $N \lhd N_1 \lhd ... \lhd N_k = G$ for $k \geq 1$ normal subgroups of $G$ containing $N$. 	
	Then either $H$ is a composition factor of $N$, or if it is not, then it must be equal to $N_{i+1}/N_i$, for some $i \in [k]$. 
	By the Third Isomorphism Theorem, $1 \lhd N_1/N \lhd ... \lhd N_k/N \lhd G/N$ is a composition series of $G/N$, and $(N_{i+1}/N) /(N_i/N) \cong H$, so $H$ is a composition factor of $G/N$ in this case.
\end{proof}	

What we will also need is that a group which is alternating on one of its orbits $A$ is either still alternating on $A$ when the rest is fixed pointwise, or the action on $A$ is always completely determined by the action outside of $A$. For the proof idea of this lemma, I thank Daniel Wiebking.
\begin{lemma}
	\label{lem:G_fixingTheRestPointwise}
	Let $H \leq G \leq \Sym_n$ and let $A$ be an orbit of $H$ such that $\Alt(A) \leq H^A$. Then either, $\Alt(A) \leq (H^{([n] \setminus A)})^A$, or for every $h \in H$, the action of $h$ on $[n] \setminus A$ also determines the action of $h$ on $A$. The latter means that there are no two distinct $g,h \in H$ which induce the same permutation on $[n] \setminus A$ but distinct permutations on $A$.
\end{lemma}	
\begin{proof}
	It holds that $N := (H^{([n] \setminus A)})^A$ is a subgroup of $H^A$. Moreover, this subgroup is normal. To see this, let $h \in N, g \in H^A$. We want to show that $ghg^{-1} \in N$. There exist $h' \in H^{([n] \setminus A)}, g' \in H$ such that $h', g'$ are extensions of $h$ and $g$, i.e.\ their restriction to $A$ corresponds to $h$, $g$, respectively. It is clear that $g'h'g'^{-1} \in H^{([n] \setminus A)}$ because this permutation fixes every point outside of $A$. Therefore, $ghg^{-1} \in N$, because this is just the action of $g'h'g'^{-1}$ on $A$. So $N \lhd H^A$. Since $\Alt(A) \leq H^A$, $H^A$ is either $\Alt(A)$ or $\Sym(A)$. If it is $\Alt(A)$, then $N$ is either also $\Alt(A)$ or the trivial group $\{1\}$, because $\Alt(A)$ has no other normal subgroups. If $H^A = \Sym(A)$, then $N$ is trivial, $N = \Alt(A)$, or $N = \Sym(A)$. So if $N$ is not trivial, then $\Alt(A) \leq N = (H^{([n] \setminus A)})^A$. Otherwise, if $N$ is trivial, then every permutation in $H^{([n] \setminus A)}$ also fixes $A$ pointwise. It follows that there do not exist any two distinct $g,h \in H$ such that $g,h$ are equal on $[n] \setminus A$ but different on $A$. If they existed, then $gh^{-1} \in H^{([n] \setminus A)}$, but $gh^{-1}$ is not the identity on $A$.  
\end{proof}	
Now we will prove the main technical result that is needed for Lemma \ref{lem:G_containmentOfAlternatingGroup}. It essentially says that if a group $G$ has a large alternating group as a \emph{composition factor}, then it also has this large alternating group as a \emph{subgroup} in some sense, or otherwise, the index of $G$ in $\Sym_n$ must be large. The proof is by induction on the compositional structure of $G$, i.e.\ in the inductive step, we choose a normal subgroup and factor it out and then continue inductively with the normal subgroup or with the factor group, in the spirit of Lemma \ref{lem:G_recursiveStep}. With the next lemma, we can prove Lemma \ref{lem:G_containmentOfAlternatingGroup} using a fact from the literature: Every large group must also have a large alternating group as a composition factor. So the key step is the one from composition factor to subgroup. Again, I thank Daniel Wiebking for his help with the proof, especially for solving the primitive case.
\begin{lemma}
	\label{lem:G_groupInduction}
	Let $0 < c \leq 1$ be a constant. 
	Let $(G_n)_{n \in \bbN}$ be a family of groups such that for all $n$, $G_n \leq \Sym_d$, where $c \cdot n \leq d \leq n$ (to be precise: this can be a different $d$ for every $n$), and such that for all large enough $n$, $G_n$ has a composition factor isomorphic to $\Alt_{m}$, for some $m \geq c \cdot n$. Then, for every large enough $n$, one of following two cases can arise:
	\begin{enumerate}[(i)]
		\item There exists a subgroup $H_n \leq G_n$ and an orbit $A$ of $H_n$ with $|A| \geq c\cdot n$ such that $\Alt(A) \leq H_n^{([d] \setminus A)}$.
		\item $[\Sym_d : G_n ] \geq \frac{1}{n} \cdot \Big(\frac{2d \cdot \alpha}{e } \Big)^{cn/2} \cdot \alpha^\ell$, where $d$ is the degree of $G_n$, $\ell$ is the number of occurrences of the transitive imprimitive case in the recursion starting with $G$ and ending with the composition factor $\Alt_m$ or in another non-recursive case. The factor $\alpha$ is $1/(\lfloor c^{-1} \rfloor!)$. 
	\end{enumerate}	
\end{lemma}	
\begin{proof}
	Fix $n \in \bbN$ and let $d$ be such that $G_n \leq \Sym_d$ (we suppress the subscript $n$ in the following). We prove the lemma by induction on the compositional structure of $G$ and choose a normal subgroup that we factor out in each step, until we arrive at the composition factor $\Alt_m$, which occurs in $G$ according to the assumption of the lemma. If $G = \Alt_m$, then we are in case (i) and are done. Otherwise, $G$ must have a normal subgroup because else, $G$ would be simple and would not contain the composition factor $\Alt_m$.
	We choose this subgroup depending on which of the following is the case:\\
	\\
	\textit{Case 1: $G$ is intransitive.}\\
	Let $\Omega \subseteq [d]$ be an arbitrary orbit with $|\Omega| \geq cn$. Such an orbit must exist because otherwise, $G$ cannot have $\Alt_m$ for $m \geq cn$ as a composition factor. To see this, consider a chain of normal subgroups that pointwise fix an orbit, one after the other. The corresponding factor groups are always the restrictions of the next normal subgroup to one orbit, and so they can never contain $\Alt_m$ if all orbits are too small. Eventually, we have fixed every orbit pointwise, which leads to the trivial group, and this cannot contain $\Alt_m$, either. Therefore, at least one large enough orbit $\Omega$ must exist.\\	
	Let $N := G^{(\Omega)}$ be the pointwise stabiliser of that orbit. It holds $N \lhd G$, and the factor $G / N$ is isomorphic to $G^{\Omega}$, the action of $G$ on $\Omega$. Let $\Delta := [d] \setminus \Omega$ be the complement of the orbit. Now we apply Lemma \ref{lem:G_recursiveStep}. It tells us that the large alternating group which must appear as a composition factor in $G$ is either isomorphic to $G / N$, or it is a composition factor of $N$ or of $G / N$.\\
	
	If $\Alt_{m}$ is a composition factor of $N$, then we apply the inductive hypothesis to $N^{\Delta} \cong N$. It yields in case (i) a subgroup $H \leq N \leq G$ and a set $A \subseteq [d] \setminus \Delta$ with $|A| \geq c\cdot n$ that is an $H$-orbit and satisfies $\Alt(A) \leq H^{(\Delta \setminus A)}$. Note that in the induction, only the degree $d$ decreases, but $c \cdot n$ remains fixed for each group $G_n$. Therefore, the size of the set $A$ that we get by induction is indeed $\geq cn$. Since $N$ fixes $\Omega$ pointwise, so does $H$, and thus, we have $\Alt(A) \leq H^{([d] \setminus A)}$. So we also have case (i) for $G$.\\
	If case (ii) applies to $N$, then let $d' := |\Delta|$ be the degree of $N^{\Delta} \cong N$. The induction hypothesis yields $[\Sym_d^{(\Omega)} : N] = [\Sym(\Delta) : N ] \geq  \frac{1}{n} \cdot \Big(\frac{2d' \cdot \alpha}{e} \Big)^{cn/2} \cdot \alpha^\ell$. We have
	\[
	[\Sym_d : N ] = [\Sym_d : \Sym_d^{(\Omega)}] \cdot [\Sym_d^{(\Omega)} : N] = (d!/d'!) \cdot [\Sym_d^{(\Omega)} : N].
	\]
	Since $N \leq G \leq \Sym_d$, we also have $[\Sym_d : N ] = [\Sym_d : G] \cdot [G : N]$. We know that $[G : N] \leq (d - d')!$ because $[G : N]$ is the order of $G/N$, whose permutation domain is $\Omega$.
	In total we get for $[\Sym_d : G]$, using the Stirling approximation $n! \approx \sqrt{2\pi n} \cdot \Big( \frac{n}{e}\Big)^n$: 
	\begin{align*}
		[\Sym_d :  G] &= \frac{[\Sym_d : N ]}{[G : N]} \geq \frac{d!}{d'! \cdot (d-d')!} \cdot \frac{1}{n} \cdot \Big(\frac{2d' \cdot \alpha}{e} \Big)^{cn/2} \cdot \alpha^\ell\\
		& \approx \frac{1}{n} \cdot \Big(\frac{2d' \cdot \alpha}{e} \Big)^{cn/2} \cdot \alpha^\ell \cdot \frac{\sqrt{2\pi d} \cdot (d/e)^d}{d'! \cdot (d-d')!}\\
		&=  \frac{1}{n} \cdot \Big(\frac{2d \cdot \alpha}{e} \Big)^{cn/2}  \cdot \alpha^\ell \cdot \frac{d^{d-(cn/2)} \cdot d'^{cn/2} \cdot \sqrt{2\pi d}}{d'! \cdot (d-d')! \cdot e^{d}}\\
		&\approx  \frac{1}{n} \cdot \Big(\frac{2d \cdot \alpha}{e} \Big)^{cn/2}  \cdot \alpha^\ell \cdot \frac{d^{d-(cn/2)} \cdot d'^{cn/2} \cdot \sqrt{2\pi d}}{d'^{d'} \cdot (d-d')^{d-d'} \cdot 2\pi \sqrt{d'(d-d')}}
	\end{align*}
	Now since $d' \geq cn$ (otherwise $N$ cannot have $\Alt_m$ as a composition factor), we can use the $d-(cn/2)$ many $d$-factors in the numerator to dominate all factors $(d-d')$ in the denominator. This yields:
	\begin{align*}
		&\geq \frac{1}{n} \cdot \Big(\frac{2d \cdot \alpha}{e} \Big)^{cn/2}  \cdot \alpha^\ell  \cdot \frac{d^{d'-(cn/2)} \cdot \sqrt{2\pi d}}{d'^{d'-(cn/2)} \cdot 2\pi \sqrt{d'(d-d')}}\\
		&\geq \frac{1}{n} \cdot \Big(\frac{2d \cdot \alpha}{e} \Big)^{cn/2}  \cdot \alpha^\ell  \cdot \frac{(1+(c/(1-c)))^{d'-(cn/2)}}{\sqrt{2\pi d'}}
	\end{align*}
	In the last step, we used that $d/d' \geq (d'+cn) /d' \geq 1+\frac{c}{1-c}$. This holds because $d-d' \geq cn$ by the choice of $\Omega$ (and therefore, $d' \leq (1-c)n$). Furthermore, we cancelled $\sqrt{d}$ and $\sqrt{d-d'}$. Now the exponent $d'-(cn/2)$ is at least $d'/2$ (because $d' \geq cn$), and the base is some constant $>1$, so it can be checked that the whole fraction in the right factor is $\geq 1$, unless $d'$ and hence $cn$ is smaller than some constant depending on $c$. So, for large enough $n$, we can remove the factor on the right and are left with:
	\[
	\geq  \frac{1}{n} \cdot \Big(\frac{2d \cdot \alpha}{e} \Big)^{cn/2} \alpha^\ell.
	\]
	So (ii) holds for $G$.
	It remains to deal with the case that the large alternating group is a composition factor of $G/N \cong G^{\Omega}$ or that it is isomorphic to $G/N$. In the latter case, $G$ acts on $\Omega$ as $\Alt_m$, for $m = |\Omega| \geq cn$. If $\Alt(\Omega) \leq G^{(\Delta)}$, then we are done and have case (i) for $G$. If $G$ does not act as $\Alt(\Omega)$ when it fixes $\Delta$ pointwise, then by Lemma \ref{lem:G_fixingTheRestPointwise}, for every $g \in G$, the effect of $g$ on $\Omega$ is fully determined by the effect of $g$ on $\Delta$. Thus, $|G| \leq |\Delta|! = d'!$. Then $[\Sym_d : G] \geq \frac{d!}{d'!}$. We have $d' \leq d - cn$ because $|\Omega| \geq cn$. So
	\begin{align*}
		[\Sym_d : G] \geq \frac{d!}{d'!} \geq \frac{d!}{(d-cn)!} \geq \Big(\frac{d}{e}\Big)^{cn} \cdot \sqrt{d/(d-cn)} \geq \Big(\frac{2d}{e \cdot {\lfloor c^{-1} \rfloor!}}\Big)^{cn/2} = \Big(\frac{2d \cdot \alpha}{e}\Big)^{cn/2} .
	\end{align*}	
	The third inequality uses again the Stirling approximation for the factorials, and the last inequality holds because $\frac{d}{e} \geq \frac{2}{c^{-1}}$. Therefore, if this happens, we have case (ii) for $G$ (the additional factors $\frac{1}{n} \cdot \alpha^\ell$ in case (ii) only make the expression smaller).\\
	
	If $G/N$ is not isomorphic to the alternating group, then $\Alt_m$ must be a composition factor of $G/N = G^\Omega$. We apply the induction hypothesis. 
	Again, this gives us two cases that can arise for $G/N$. In case (i), there exists $H \leq G^{\Omega}$ and an orbit $A \subseteq \Omega$ of $H$ with $|A| \geq c\cdot n$ such that $\Alt(A) \leq H^{(\Omega \setminus A)}$. Then let $H'$ be a subgroup of $G$ whose restriction to $\Omega$ is $H$. If $\Alt(A) \leq H'^{([d] \setminus A)}$, then we have case (i) for $G$, as witnessed by $H'$. Otherwise, we use again Lemma \ref{lem:G_fixingTheRestPointwise}, which says that the action of $H'$ outside of $A$ determines the action in $A$. We claim that this also true for $G$ itself, i.e.\ there are no two $g,g' \in G$ which are different on $A$ and equal on $[d] \setminus A$.\\ 
	\textit{Proof of claim:} If the claim is not true, then $G^{([d] \setminus A)}$ is non-trivial. Hence $G^{([d] \setminus A)}$ is not a subgroup of $H'$ because then, the action of $H'$ on $A$ would not be determined by its action outside of $A$. If $G^{([d] \setminus A)}$ is not a subgroup of $H'$, then $(G^{([d] \setminus A)})^\Omega$ is not a subgroup of $H$, either. But we can always assume that $H$ contains $(G^{([d] \setminus A)})^\Omega$: The fact that $\Alt(A) \leq H^{(\Omega \setminus A)}$ will still hold if we add to $H$ all permutations in $G^{([d] \setminus A)}$. This can at most push $H^{(\Omega \setminus A)}$ from $\Alt(A)$ to $\Sym(A)$. Thus, we can assume that $H$ is such that the claim holds and the action of $G$ on $A$ is indeed determined by its action on $[d] \setminus A$.\\
	
	It follows that $|G| \leq (d - |A|)!$, and since $|A| \geq cn$, we get the same lower bound for $[\Sym_d : G]$ as above in the case where $G/N \cong \Alt_m$. So in this case, case (ii) applies to $G$.\\
	
	It remains to check what happens if the induction gives us case (ii) for $G/N$. Then we have $[\Sym(\Omega) : G/N ] \geq \frac{1}{n} \cdot \Big(\frac{2(d-d') \cdot \alpha}{e} \Big)^{cn/2} \cdot \alpha^\ell$. It holds 
	\[
	|G/N| = |\Sym(\Omega)|/[\Sym(\Omega) : G/N ] = \frac{(d-d')!}{[\Sym(\Omega) : G/N ]}.
	\]
	We have $|G| = |N| \cdot |G/N|$, and $|N| \leq d'!$. We get a similar chain of inequalities as before:
	\begin{align*}
		[\Sym_d : G] &= \frac{d!}{|G|} = \frac{d!}{|N| \cdot |G/N|}\\
		&\geq \frac{d! \cdot (2(d-d')\alpha)^{cn/2} \cdot\alpha^\ell}{d'! \cdot (d-d')! \cdot n \cdot e^{cn/2}}\\
		&\approx  \frac{1}{n} \cdot \Big(\frac{2d\alpha}{e} \Big)^{cn/2} \cdot \alpha^\ell \cdot \frac{d^{d-(cn/2)} \cdot (d-d')^{cn/2} \cdot \sqrt{2\pi d}}{d'^{d'} \cdot (d-d')^{d-d'} \cdot 2\pi \sqrt{d'(d-d')}} \\
		&\geq \frac{1}{n} \cdot \Big(\frac{2d\alpha}{e} \Big)^{cn/2} \cdot \alpha^\ell \cdot \frac{d^{d'}}{d'^{d'} \cdot \sqrt{2\pi \cdot d'}} \tag{$\star$}
	\end{align*}	
	In the last step, we cancelled all factors $(d-d')$ in the denominator with the factors $d$ and $(d-d')$ in the numerator, and also removed $\frac{\sqrt{d}}{\sqrt{d-d'}}$. Now we can almost continue as in the case before, except that we need a case distinction. The important difference is that now, we have no lower bound for $d'$ because this time, the factor $\Alt_m$ is in $G^\Omega$, so $[d] \setminus \Omega$ might be arbitrarily small (at least size $1$ because it contains at least one orbit). We distinguish the cases whether $d' \geq cn$ or $d' < cn$. Let us start with the latter case. As above, we have $d/d' \geq (d'+cn) /d'$, because this only depends on the fact that $|\Omega| \geq cn$, which still holds by choice of $\Omega$. If $d' < cn$, then this becomes: $d/d' \geq (d'+cn) /d' \geq 2$. So then, the right factor in $(\star)$ is at least $\frac{2^{d'}}{\sqrt{2\pi d'}}$. This is greater than one for all values of $d' \geq 2$. If $d' = 1$, then we can argue differently. Then, the right fraction in $(\star)$ is equal to $d/\sqrt{2 \pi}$, which is also greater than one for all large enough $n$ (because $d \geq cn$). In case that $d' \geq cn$, we make the same argument as before, and use the bound $d/d' \geq (d'+cn) /d' \geq 1+\frac{c}{1-c}$. Then the right factor in $(\star)$ is at least $\frac{(1+\frac{c}{1-c})^{d'}}{\sqrt{2\pi d'}}$. Again, for large enough $d'$, this is $\geq 1$. Since $d' \geq cn$, this happens for large enough $n$. So in all cases, we can remove the right factor and the product only gets smaller. So all in all, we have 
	\[
	[\Sym_d : G] \geq \frac{1}{n} \cdot \Big(\frac{2d \cdot \alpha }{e} \Big)^{cn/2} \alpha^\ell,
	\]
	as desired.
	This finishes the case where $G$ is intransitive.\\
	\\
	\textit{Case 2: $G$ is transitive, but not primitive.}\\
	In this case, $G$ has a non-trivial block system. Let $N \lhd G$ be the normal subgroup that stabilises each block setwise. Then $G/N$ is the action of $G$ on the set of blocks. Again, according to Lemma \ref{lem:G_recursiveStep}, the large alternating group is either a composition factor of $N$, of $G/N$, or it is isomorphic to $G/N$. In the latter case, (ii) applies to $G$: If $G/N \cong \Alt_m$ for some $m \geq c \cdot n$, then the block system has $m \leq d/2$ blocks. Each block has size $t$, with $2 \leq t \leq \lfloor c^{-1} \rfloor$ (note that this case can only happen if $c \leq \frac{1}{2}$). Then we have $|G| \leq (t!)^m \cdot m!$. Thus, 
	\begin{align*}
		[\Sym_d : G] &\geq \frac{d!}{(t!)^m \cdot m!} \geq \frac{d!}{(t!)^m \cdot (d/2)!}\\
		&\geq \frac{\sqrt{2} \cdot (2d/e)^{(d/2)}}{(t!)^{(d/2)}} \geq \Big( \frac{2d \cdot \alpha}{e} \Big)^{d/2}\\
		& \geq \Big( \frac{2d \cdot \alpha}{e} \Big)^{c \cdot n/2} \geq  \frac{1}{n} \cdot \Big( \frac{2d \cdot \alpha}{e} \Big)^{c \cdot n/2} \cdot  \alpha^1.
	\end{align*}	
	The last step holds because $\alpha \leq 1$. 
	The next case is that $\Alt_m$ is a composition factor of $G/N$. Then we do not need the inductive step either and can immediately conclude with the same lower bound as above for $[\Sym_d : G]$. This is because if $\Alt_m$ is a composition factor of $G/N$, the degree of $G/N$ must be at least $m \geq cn$, and so, the number of blocks must also be greater than $cn$, and the block size can be at most $\lfloor c^{-1} \rfloor$.\\
	
	It remains the case that $N$ has $\Alt_m$ as a composition factor. We apply the inductive hypothesis to $N$. Should case (i) hold for $N$, then we immediately know that case (i) also applies to $G$, because then we have some $H \leq N \leq G$ (so in particular, $H \leq G$) and an $H$-orbit $A \subseteq [d]$ such that $\Alt(A) \leq H^{([d] \setminus A)}$ (this is easier than before because $N$ and $G$ have the same permutation domain now).\\ 
	If we instead have case (ii) for $N$, then $[\Sym_{d} : N] \geq \frac{1}{n} \cdot \Big( \frac{2d \cdot \alpha}{e} \Big)^{c \cdot n/2} \cdot \alpha^{\ell-1}$.\\
	Each block must be of size at least $c \cdot n$ because $N$ has an alternating group of that degree as a composition factor, and the blocks are the orbits of $N$. If these orbits were smaller than $c \cdot n$, then $N$ could not have a composition factor of that degree, as we already argued in the intransitive case.
	Therefore, the number of blocks is at most $\lfloor c^{-1} \rfloor$ and so, $|G/N| = [G : N] \leq \lfloor c^{-1} \rfloor!$. In total, we have
	\begin{align*}
		[\Sym_d : G] &= [\Sym_d : N] / [G : N] \geq \frac{1}{n} \cdot \Big( \frac{2d \cdot \alpha}{e} \Big)^{c \cdot n/2} \cdot \frac{\alpha^{\ell-1}}{\lfloor c^{-1} \rfloor!}\\
		&= \frac{1}{n} \cdot \Big( \frac{2d \cdot \alpha}{e} \Big)^{c \cdot n/2} \cdot \alpha^{\ell}.
	\end{align*}
	This finishes the transitive and imprimitive case.\\
	\\
	\textit{Case 3: $G$ is primitive.}\\
	It remains the case that $G$ is primitive. Since $G$ has $\Alt_m$ as a composition factor, with $m \geq cn$, we have for the order of $G$: $|G| \geq (cn)!/2 \approx \frac{1}{2}\Big(\frac{cn}{e}\Big)^{cn} \cdot \sqrt{2\pi cn} \geq n^{1+\log n} \geq d^{1+\log d}$. Thus, by Theorem \ref{thm:G_primitiveGroups}, $G$ has a normal subgroup $N$ of index $\leq d$ such that $N$ has a block system on which it acts as a Johnson group $\Alt_k^{(t)}$, where $k \geq \log d$. Here, $k$ is the size of the permutation domain, and $t$ is the length of the tuples over $[k]$ of the Johnson action. Since the index of $N$ in $G$ is $\leq d$, $G/N$ has order at most $d$, so it is too small to have the composition factor $\Alt_m$. Therefore, $N$ must have $\Alt_m$ as a composition factor. 
	Now $N$ acts as $\Alt_k^{(t)}$ on the blocks of its block system. We have $k \geq \log d$, and every block is identified with a $t$-tuple over a $k$-element domain. So the number of blocks is $k^t \geq (\log d)^t$. Thus, the block-size is at most $d/(\log d)^t$. Even for $t = 1$, this is asymptotically less than $cn$ (recall that $d \leq n$). 
	So if $n$ is large enough, then our sought alternating group cannot be a factor of the blockwise stabiliser $N' \lhd N$ (whose orbits are the blocks), but it must be a factor of $N/N'$, which is the action of $N$ on the blocks. This requires the number of blocks to be $\geq cn$. Then the block-size can be at most $\lfloor c^{-1} \rfloor$, as in Case 2 above. If the block-size is $\geq 2$, then we obtain $[\Sym_d : N] \geq \Big( \frac{2d\alpha}{e} \Big)^{cn/2}$ with exactly the same calculation as in Case 2. Then 
	\[
	[\Sym_d : G] = [\Sym_d : N] / [G : N] \geq  \frac{1}{n} \cdot \Big( \frac{2d \cdot \alpha}{e} \Big)^{cn/2}.
	\]
	Here, we used that $[G:N] \leq d \leq n$. Thus, we get the desired index-bound for $G$ without having to recurse any further. In the case that the block system of $N$ is trivial and the blocks are singletons, then we must have $t=1$: If $t \geq 2$, then $N$ is isomorphic to an alternating group on $\leq \sqrt{d} < cn$ points. Alternating groups have no normal subgroups other than the trivial group and itself, so in this case, $N$ cannot have $\Alt_m$ with $m \geq cn$ as a composition factor. So it remains the case that $t=1$, and $N = \Alt_d$. Since $N \leq G$, we have case (i) for $G$ then. 
\end{proof}	

We now combine this lemma with the following result from the literature, which guarantees the existence of a large alternating group as a composition factor in any sufficiently large group.
\begin{lemma}[Lemma 2.2 in \cite{babaiCameronPalfy}]
	\label{lem:G_babai_alternatingGroup}
	Let $G$ be a permutation group of degree $n$. If $G$ has no composition factor isomorphic to an alternating group of degree $> C$, then $|G| \leq C^{n-1} (C \geq 6)$.
\end{lemma}

We get the following consequence for groups of bounded index:
\begin{corollary}	
	\label{cor:G_groupsOfBoundedIndexHaveAlt}
	Let $(G_n)_{n \in \bbN}$ be a family of groups such that for all $n$, $G_n \leq \Sym_n$. Assume that there exists some constant $k \in \bbN$ such that asymptotically, $[\Sym_n : G_n] \leq 2^{nk}$.\\
	Then there is a function $f(n) \in \Theta(n)$ such that for all large enough $n$, $G_n$ has a composition factor isomorphic to $\Alt_{m}$, for some $m \geq f(n)$.
\end{corollary}	
\begin{proof}
	Let $f(n): \bbN \lra \bbN$ be the smallest upper bound for the maximum degree of an alternating group that appears as a composition factor of $G_n$. We have to show that $f(n)$ is linear in $n$. Suppose for a contradiction that $f(n) \in o(n)$. Let $g(n) := f(n)+1$, for example. Then $g(n) \in o(n)$, and it holds for all large enough $n$ that $G_n$ has no composition factor isomorphic to $\Alt_{g(n)}$. By Lemma \ref{lem:G_babai_alternatingGroup}, we have $|G_n| \leq g(n)^{n-1}$. Then we get for the index:
	\[
	[\Sym_n :  G_n  ] \geq \frac{n!}{g(n)^{n-1}}.
	\]
	Using the Stirling approximation, we obtain:
	\[
	\frac{n!}{g(n)^{n-1}} \approx \Big(\frac{n}{e \cdot g(n)}\Big)^{n-1} \cdot \Big(\frac{n \cdot \sqrt{2 \pi n} }{e} \Big).
	\]
	Since $g(n)$ is sublinear, the fraction $\frac{n}{e \cdot g(n)}$ is not bounded from above by any constant, so we have $[\Sym_n :  G_n  ] > 2^{nk}$, for every constant $k \in \bbN$. This is a contradiction, so $f(n)$ must be linear. 
\end{proof}

Finally, we can put everything together to prove Lemma \ref{lem:G_containmentOfAlternatingGroup}.
\begin{proof}[Proof of Lemma \ref{lem:G_containmentOfAlternatingGroup}]
	Corollary \ref{cor:G_groupsOfBoundedIndexHaveAlt} states that all large enough $G_n$ have a composition factor isomorphic to $\Alt_m$, where $m \geq f(n)$, for some function $f(n) \in \Theta(n)$. There is a constant $0 < c \leq 1$ such that $f(n) = c \cdot n$. Hence, Lemma \ref{lem:G_groupInduction} applies (where $d := n$).\\
	We now want to show that case (ii) from Lemma \ref{lem:G_groupInduction} cannot occur for any of the $G_n$. Suppose for a contradiction that $[\Sym_n : G_n] \geq \frac{1}{n} \cdot \Big(\frac{2n \cdot \alpha}{e} \Big)^{cn/2} \cdot \alpha^\ell$ for some large enough $n$. First, we bound $\ell$. In Lemma \ref{lem:G_groupInduction}, this was defined as the number of times the transitive and imprimitive case occurs in the recursion. By inspection of the proof, we see that this case occurs only if $G$ is transitive and imprimitive, \emph{and} the sought composition factor $\Alt_m$ is in the blockwise stabiliser $N \lhd G$. In this case, the block-size must be $\geq cn$. Whenever this case occurs, we continue with the intransitive case where the orbits are the blocks, which means that we eventually end up with a group whose degree is at most the block size of the block system of $G$. So whenever we have the transitive imprimitive case and continue the recursion, we at least halve the degree. We may only do this until the degree drops below $cn$, so we have $\ell \leq \log n$ (actually, the bound could be made even smaller, but $\log n$ is already much tighter than we really need). Thus,
	\begin{align*}
		[\Sym_n : G_n] &\geq \frac{1}{n} \cdot \Big(\frac{2n}{e \cdot \lfloor c^{-1} \rfloor!} \Big)^{cn/2} \cdot \Big(\frac{1}{\lfloor c^{-1}\rfloor!} \Big)^{\log n}\\
		&\geq \Big(\frac{2n}{e \cdot \lfloor c^{-1}\rfloor!^2} \Big)^{cn/2} 
	\end{align*}	
	Now for every constant $k\in \bbN$ that we could choose, $2^{kn}$ can be written as $a^n$ for some constant $a$ (that can be arbitrarily large but still constant). Comparing the above expression with $a^n$, we find:
	\begin{align*}
		\frac{[\Sym_n : G_n]}{a^n} &\geq \Big(\frac{2n}{e \cdot (c^{-1})!^2} \Big)^{cn/2} \cdot \frac{1}{a^n}\\
		&=  \Big(\frac{2n}{a^{2/c} \cdot e \cdot (c^{-1})!^2} \Big)^{cn/2}. 
	\end{align*}	
	Clearly, for any constant $a$, this fraction tends to infinity in the limit. So, if $[\Sym_n : G_n]$ is bounded from above by $2^{kn}$, then for all large enough $n$, case (ii) from Lemma \ref{lem:G_groupInduction} cannot possibly apply to $G_n$.
\end{proof}

Now we can prove Theorem \ref{thm:G_alternatingPartitionWithLogManyParts} using Lemma \ref{lem:G_containmentOfAlternatingGroup}:

\restateAlternatingPartitionWithLogManyParts*

\begin{proof}
Assume for a contradiction that there was a constant $0 < c \leq 1$ such that for every large enough $n$, there exists a subset $\Delta_n \subseteq [n]$ of size $|\Delta_n| \geq cn$, which is precisely the set of elements in singleton parts in $\Sp_A(G_n)$. In the following, we assume that the size of $\Delta_n$ is not only lower-bounded linearly, but that there also exists some other constant $1 > d > c$ such that $|\Delta_n| \leq dn$, for all large enough $n$. We deal with the other case in the end.\\ 
Let $H_n := (G_n^{([n] \setminus \Delta)})^{\Delta}$ be the subgroup of $G_n$ that fixes every point outside of $\Delta$, restricted to its action on $\Delta$. We now show that $[\Sym(\Delta) : H_n] > 2^{nk}$, for every $k \in \bbN$. Indeed, if there were a $k \in \bbN$ such that $[\Sym(\Delta) : H_n] \leq 2^{nk}$ for all large enough $n$, then by Lemma \ref{lem:G_containmentOfAlternatingGroup}, all the $H_n$ would have a subgroup $H'_n$ containing the alternating group on an orbit $A$ of size $|A| \geq c' \cdot |\Delta|$, where $c'$ is the constant from the lemma. More precisely, $\Alt(A) \leq H_n'^{(\Delta \setminus A)}$. Thus, every even permutation on $A$ (that fixes everything in $\Delta$ outside of $A$) would be contained in $H_n$, and since $H_n$ is a subgroup of $G_n$ that fixes everything outside of $\Delta$, this means that $\Alt(A) \leq G^{([n] \setminus A)}$. But then, $A \subseteq \Delta$ would be a part of the coarsest alternating supporting partition $\Sp_A(G_n)$. This contradicts the fact that $\Sp_A(G_n)$ only has singleton parts on $\Delta$. Therefore, we must have that $[\Sym(\Delta_n) : H_n] > 2^{nk}$, for every $k \in \bbN$.\\
As a next step, we calculate that this entails a violation of the assumption that $[\Sym_n : G_n]$ can be upper-bounded by $2^{nk}$, for some $k$. Namely, if $[\Sym(\Delta_n) : H_n] > 2^{nk}$ for all $k \in \bbN$, then $|H_n| < \frac{|\Delta_n|!}{2^{nk}}$, for every choice of $k$. Then we also have $|G_n| < \frac{|\Delta_n|!}{2^{nk}} \cdot (n-|\Delta_n|)!$. This is because by Lemma \ref{lem:G_altSPsandwichLemma}, every element of $G_n$ stabilises $\Sp_A(G_n)$ setwise and so in particular, maps singleton parts only to singleton parts; therefore, $[n] \setminus \Delta$ is a union of $G$-orbits and so $H_n$ is normal in $G_n$. It follows for the index, for every $k \in \bbN$:
\begin{align*}
[\Sym_n : G_n] &= \frac{n!}{|G_n|} > \frac{n! \cdot 2^{nk}}{|\Delta_n|! \cdot (n-|\Delta_n|)!}.
\end{align*}	
Since $cn \leq |\Delta_n| \leq dn$, we have $|\Delta_n| \leq dn$ and $n-|\Delta_n| \leq (1-c)n$. So we can 	continue: 
\begin{align*}
\frac{n! \cdot 2^{nk}}{|\Delta_n|! \cdot (n-|\Delta_n|)!} &\geq \frac{n! \cdot 2^{nk}}{(dn)! \cdot ((1-c)n)!}\\
& \approx 2^{nk} \cdot \Big( \frac{1}{d}\Big)^{dn} \cdot \Big( \frac{1}{1-c}\Big)^{(1-c)n} \cdot \frac{1}{\sqrt{2\pi \cdot d(1-c) \cdot n}}
\end{align*}
Since $d < 1$ and $(1-c) < 1$, the above product is greater than $2^{nk}$ (for large enough $n$). Because this lower bound for $[\Sym_n : G_n]$ holds for every $k \in \bbN$, we can conclude that $[\Sym_n : G_n] > 2^{nk}$, for every $k$. But the assumption of Theorem \ref{thm:G_alternatingPartitionWithLogManyParts} says that there exists a $k \in \bbN$ such that for all large enough $n$, $[\Sym_n : G_n] \leq 2^{nk}$. This is a contradiction. This proves the theorem in case that $|\Delta_n|$ can be upper-bounded by some linear function $d \cdot n$, for $d < 1$. The case that $|\Delta_n| > dn$, for all $d < 1$, cannot occur: Lemma \ref{lem:G_containmentOfAlternatingGroup} applied to $G_n$ states that $G_n$ contains an alternating group of linear degree (which fixes the rest pointwise). Therefore, for all large enough $n$, $\Sp_A(G_n)$ must have at least one part of linear size. So it is impossible that all but sublinearly many elements of $[n]$ are in singleton parts in $\Sp_A(G_n)$. 
\end{proof}	
	
\end{document}